\documentclass[11pt]{article}
\usepackage{amsfonts}
\usepackage[latin9]{inputenc}
\usepackage{geometry}
\usepackage{setspace}
\usepackage{amsmath}
\usepackage{amssymb}
\usepackage{bbm}
\usepackage{dcolumn}
\usepackage{graphicx}
\usepackage{enumerate}
\usepackage[inline, shortlabels]{enumitem}
\usepackage[multiple]{footmisc}
\usepackage{hhline}
\usepackage[position=bottom]{subfig}
\usepackage[authoryear]{natbib}
\usepackage[colorlinks,linkcolor=blue,citecolor=blue]{hyperref}
\usepackage{algorithm,algorithmic}

\newcolumntype{d}{D{.}{.}{-1}}
\newtheorem{theorem}{Theorem}
\newtheorem{remark}{Remark}
\newtheorem{lemma}{Lemma}
\newtheorem{corollary}{Corollary}

\newtheorem{assumption}{Assumption}

\newenvironment{proof}[1][Proof]{\noindent \textbf{#1.} }{\  \rule{0.5em}{0.5em}}

\begin{document}

\title{Inference on Linear Regressions \\
with Two-Way Unobserved Heterogeneity}
\author{Hugo Freeman\thanks{
Department of Economics, Michigan State University, East Lansing, USA.
E-mail: freem391@msu.edu} \and Dennis Kristensen\thanks{
Department of Economics, University College London, London, UK. E-mail:
d.kristensen@ucl.ac.uk} }
\date{}
\maketitle

\begin{abstract}
We develop a general estimation and inference procedure for the common
parameters in linear panel data regression models with nonparametric two-way
specification of unobserved heterogeneity. The procedure takes as input any
first-step estimators of the nonparametric regression function and the fixed
effects and relies on two key ingredients: First, we develop moment
conditions for the common parameters that are Neyman orthogonal with respect
to the nonparametric regression function. Second, we employ a novel
adjustment of the nonparametric regression estimator so the estimated fixed
effects do not generate incidental parameter biases. Together, these ensure
that the resulting estimator of the common parameters is $\sqrt{NT}$%
--asymptotically normally distributed under weak conditions on the
estimators of fixed effects and regression function. Next, we propose a
novel two-step estimator of the nonparametric regression function and the
fixed effects and verify that this particular estimator satisfies the
conditions of our general theory. A numerical study shows that the proposed
estimators perform well in finite samples.
\end{abstract}

\section{Introduction}

We are interested in inference on $\beta $ in the model, 
\begin{equation}
Y_{it}=X_{it}^{\prime }\beta +g(\alpha _{i},\gamma _{t})+\varepsilon
_{it},\quad i=1,\dots ,N,\quad t=1,\dots ,T,  \label{eqn:model}
\end{equation}%
where $X_{it}\in \mathbb{R}^{K}$ are a set of observed covariates, $\alpha
_{i}$ and $\gamma _{t}$ are unobserved finite--dimensional fixed effects
that potentially correlate with $X_{it}$, and $\varepsilon _{it}$ is an
idiosyncratic shock that is mean--independent of $X_{i,t}$, $\alpha _{i}$
and $\gamma _{t}$. The unknown function $g(\cdot ,\cdot )$ adheres to
certain smoothness conditions but is otherwise left unrestricted and treated
as a nonparametric object.

Our framework includes as special cases well-known parametric models such as
the two-way additive fixed-effects model, $g(\alpha _{i},\gamma _{t})=\alpha
_{i}+\gamma _{t}$ and the interactive fixed-effects model, $g(\alpha
_{i},\gamma _{t})=\sum_{r=1}^{R}u_{r}\left( \alpha _{i}\right) v_{r}\left(
\gamma _{t}\right) =\sum_{r=1}^{R}\lambda _{ir}f_{tr}$, where $\lambda
_{ir}=\sigma _{r}u_{r}(\alpha _{i})$ and $f_{tr}=v_{r}(\gamma _{t})$. In
both cases, the (transformed) individual effects and time effects can be
controlled for with, e.g., the least square estimator %
\citep{Bai2009,MoonWeidner2015}, or correlated common effects estimator %
\citep{Pesaran2006} and the resulting estimator of $\beta $ is $\sqrt{NT}$%
--asymptotically normally distributed.

However, the nonparametric case presents a challenging scenario that does
not trivially lead to $\sqrt{NT}$-consistent estimates for $\beta $. For
example, \cite{freeman2023linear} develop least-squares estimator of $\beta $
where a sieve--estimator of $g$ is employed, but find that the resulting
estimator of $\beta $ does not enjoy $\sqrt{NT}$-asymptotic normality; see
also \cite{fernandez2021low}.

The contribution of this paper is two--fold: First, we propose a general
estimation procedure for $\beta $ that takes as input any first-step
estimators of $g$ and $(\alpha _{i},\gamma _{t})$. Under weak conditions on
the first--step estimators, we show that the resulting estimator of $\beta $
is $\sqrt{NT}$-asymptotically normally distributed, thereby allowing for
standard inference tools to be employed. Second, we propose novel estimators
of $g$ and $(\alpha _{i},\gamma _{t})$ and verify that these satisfy the
regularity conditions for the theory of our general estimation procedure to
hold.

The general estimation procedure relies on two main ingredients: First, we
follow \cite{freeman2022multidimensional} and develop moment conditions for
the estimation of $\beta $ that are Neyman orthogonal w.r.t. to $g$; see,
e.g., \cite{chernozhukov2022locally} for an overview of this method in the
context of semiparametric estimation. Unsurprisingly, the Neyman orthogonal
moment conditions take the same form as the ones used for estimation of the
partially linear model, where $\left( \alpha _{i},\gamma _{t}\right) $ are
treated as observed co--variates; see, e.g., \cite{robinson1988root}. In
order to operationalise these moment conditions, the researcher will have to
plug in first-step estimators of the nonparametric functions $g_{X}\left(
\alpha _{i},\gamma _{t}\right) :=\mathbb{E}\left[ X_{it}|\alpha _{i},\gamma
_{t}\right] $ and $g_{Y}\left( \alpha _{i},\gamma _{t}\right) :=\mathbb{E}%
\left[ Y_{it}|\alpha _{i},\gamma _{t}\right] $ together with estimators of
the unknown fixed effects $\left( \alpha _{i},\gamma _{t}\right) $. In the
ideal scenario where $\left( \alpha _{i},\gamma _{t}\right) $ are observed,
the Neyman orthogonality of the moment conditions w.r.t. $g_{Z}:=\left(
g_{X},g_{Y}\right) $ ensure that the nonparametric estimation of $g_{Z}$ has
no first-order effect on the estimator of $\beta $. In particular, it
ensures that the estimated fixed effects used in the estimation of $g_{Z}$
do not contribute to the asymptotic variance the estimator of $\beta $.

However, the resulting estimator of $\beta $ will generally suffer from
well-known incidental parameter biases due to estimated fixed effects. The
leading term of these biases could be adjusted for using standard
techniques, e.g., analytical bias adjustment or the Jackknife. The second
ingredient of our procedure avoids any such post--estimation bias adjustment
by combining sample splitting and a generalised version of the two--way
nonparametric regression estimator of \cite{freeman2023linear} that removes
the incidental parameter biases. Under weak regularity conditions on the
first-step nonparametric estimator, we show that the resulting estimator of $%
\beta $ will be $\sqrt{NT}$-asymptotically normally distributed and with the
same distribution as the oracle estimator where $g$ is known and $\left(
\alpha _{i},\gamma _{t}\right) $ are observed.

Next, we develop a specific estimator of $\left( \alpha _{i},\gamma
_{t}\right) $ and $g_{Z}$ that satisfy the regularity conditions for our
general estimation theory to hold. These novel estimators impose weak
conditions on $g_{Z}$ and the fixed effects. Under weak regularity
conditions, \cite{FreemanKristensen2026} show that a finite number of
leading eigenfunctions of $g$ evaluated at $\left( \alpha _{i},\gamma
_{t}\right) $ can be used as proxies for $\left( \alpha _{i},\gamma
_{t}\right) $. Our proposed estimator is then obtained in two steps: In the
first step, we employ the approximate factor model estimator of \cite%
{freeman2023linear} to obtain estimates of the leading eigenfunctions of $g$%
. In the second step, we run a nonparametric regression of $Y_{it}$ and $%
X_{it}$, respectively, onto the estimated leading eigenfunctions from the
first step to obtain our final estimator of $g_{Z}$. The nonparametric
regression approach is related to the nonparametric smoothing technique
employed in \cite{freeman2024multidimensional}.

The nonparametric regression in the second step comes in two versions: The
first version constructs a multivariate index of the eigenfunctions that are
then used as covariates in a nonparametric regression. The second one
employs an nonparametric additive regression procedure with the attractive
feature of not suffering from any curse--of--dimensionality. The two
estimators rely on different assumptions on $g_{Z}$ and come with different
convergence rates: The first one imposes weaker conditions and comes with
smaller biases but bigger variances compared to the second one. We
demonstrate that both of the two estimators exhibit sufficiently fast rate
of convergence so that they can be combined with the Neyman orthogonal
moment conditions to obtain the desired result.

Alternative estimators of $\left( \alpha _{i},\gamma _{t}\right) $ and $%
g_{Z} $ are proposed in \cite{deaner2025inferring} and \cite%
{beyhum2025inference} with the latter using their proposal to estimate $%
\beta $ based on the Neyman orthogonal moments described above. The
regression estimator of \cite{deaner2025inferring} is closest in spirit to
ours but employ a different proxy for the fixed effects, namely an estimated
pseudo--distance, while \cite{beyhum2025inference} use sample moments of $%
Z_{it}$ to estimate the fixed effects and then $k$-means methods to estimate 
$g_{Z}$. We expect our estimators of $g_{Z}$ to come with smaller biases
compared to the one of \cite{beyhum2025inference} if $g_{Z}$ is a smooth
function. In addition, the formal assumptions under which the fixed effects
estimators of \cite{deaner2025inferring} and \cite{beyhum2025inference} are
valid are different from ours. As such, the three papers complement each
other.

One particularly attractive feature of our proposed estimator of $\beta $
over the one of \cite{beyhum2025inference} is that it will in general be
more efficient:\ Suppose that 
\begin{equation}
X_{it}=g_{X}(\alpha _{i},\gamma _{t})+\tilde{g}_{X}(\tilde{\alpha}_{i},%
\tilde{\gamma}_{t})+e_{it},\text{ \ \ }\mathbb{E}\left[ e_{it}|\alpha
_{i},\gamma _{t},\tilde{\alpha}_{i},\tilde{\gamma}_{t}\right] =0,
\label{eq: X DGP Beyhum}
\end{equation}%
where $(\tilde{\alpha}_{i},\tilde{\gamma}_{t})$ are additional latent
individual and time effects that are not present in (\ref{eqn:model}). The
proposed algorithm of \cite{beyhum2025inference} will control for both $%
g_{X}(\alpha _{i},\gamma _{t})$ and $\tilde{g}_{X}(\tilde{\alpha}_{i},\tilde{%
\gamma}_{t})$ in its first step and so in effect use an estimator of $e_{it}$
as regressors in their proposed estimator of $\beta $. In contast, our
procedure uses an estimator of $X_{it}-g_{X,1}(\alpha _{i},\gamma _{t})-%
\mathbb{E}\left[ g_{X,2}(\tilde{\alpha}_{i},\tilde{\gamma}_{t})|\alpha
_{i},\gamma _{t}\right] $ as regressors in estimation of $\beta $.
Importantly, the variance of the latter will be larger than the one of $%
e_{it}$ and so leads to lower asymptotic variance of our estimator of $\beta 
$ compared to the one of \cite{beyhum2025inference}. Moreover, our procedure
suffers from a lower curse--of--dimensionality since it only involves
learning about/estimating $(\alpha _{i},\gamma _{t})$; in contrast, the one
of \cite{beyhum2025inference} requires learning about $(\tilde{\alpha}_{i},%
\tilde{\gamma}_{t})$ in addition. Thus, unless $\tilde{g}_{X}(\tilde{\alpha}%
_{i},\tilde{\gamma}_{t})=0$, our procedure should dominate theirs both in
small and large samples.

We carry out an extensive simulation study that support our theoretical
results: The estimator of $\beta $ that relies on our novel
eigenfunction--based estimators of the fixed effects suffers from only small
finite-sample biases and dominates the estimators of $\beta $ that takes as
input fixed effects estimators based on aforementioned pseudo-distance and
sample moments across a range of different DGP's. In particular, in
scenarios where $(\alpha _{i},\gamma _{t})$ are multivariate, our proposal
significantly outperforms these alternative fixed effects estimators.

The remainder of the paper is organised as follows: In Section \ref{sect:NO}
, we present the general theory for two--step estimation of $\beta $ that
allows for a broad range of first--step estimators of $g_{Z}$ and fixed
effects. We present our novel estimators of $g_{Z}$ and fixed effects in
Section \ref{sect:g_Z ID}. The asymptotic properties of this estimator are
analysed in Section \ref{sect:Asymptotics}. The results of our simulation
study are presented in Section \ref{sect:Simulations} and we conclude in
Section \ref{sect:Conclusion}. All proofs have been relegated to the
Appendix.

\section{A general theory for estimation of $\protect\beta $}

\label{sect:NO}

We here first present our general estimation approach, the Neyman
Orthogonalised estimator of $\beta $, and show that it can achieve $\sqrt{NT}
$-asymptotic normality under weak conditions on the first-step estimation of
the nonparametric component $g$. However, the estimator will generally
suffer from incidental parameter biases. We show how a Jackknife-type
procedure can be used to remove such biases when combined with sample
splitting.

Recall the following definitions, 
\begin{equation}
g_{X}\left( \alpha _{i},\gamma _{t}\right) :=\mathbb{E}[X_{it}|\alpha
_{i},\gamma _{t}],\text{ \ \ \ }g_{Y}\left( \alpha _{i},\gamma _{t}\right) :=%
\mathbb{E}[Y_{it}|\alpha _{i},\gamma _{t}],
\label{eqn:conditionalExpectation}
\end{equation}%
and write $\Gamma _{X,it}=g_{X}\left( \alpha _{i},\gamma _{t}\right) $ and $%
\Gamma _{Y,it}=g_{Y}\left( \alpha _{i},\gamma _{t}\right) $ for brevity.
Note that $\Gamma _{Y,it}=\Gamma _{X,it}^{\prime }\beta +\Gamma _{it}$,
where $\Gamma _{it}:=g\left( \alpha _{i},\gamma _{t}\right) $. Moreover, 
\begin{equation*}
\mathbb{E}[\eta _{it}|\alpha _{i},\gamma _{t}]=0,\text{ \ \ }\eta
_{it}:=X_{it}-\Gamma _{X,it},
\end{equation*}%
holds by definition of $g_{X}\left( \alpha _{i},\gamma _{t}\right) $. For
ease of notation, we let $Z_{it}=\left( Y_{it},X_{it}\right) $, $%
g_{Z}=\left( g_{Y},g_{X}\right) $, $\Gamma _{Z,it}=\left( \Gamma
_{Y,it},\Gamma _{X,it}^{\prime }\right) ^{\prime }$ and $\varepsilon
_{Z,it}=\left( \varepsilon _{it},\eta _{it}\right) $ so that $Z_{it}=\Gamma
_{Z,it}+\varepsilon _{Z,it}$, where $\mathbb{E}[\varepsilon _{Z,it}|\alpha
_{i},\gamma _{t}]=0$.

Our estimator of $\beta $ will be based on the following moment function, 
\begin{equation*}
m\left( Z_{it},\beta ,\Gamma _{Z,it}\right) =\left( X_{i,t}-\Gamma
_{X,it}\right) \left( Y_{it}-\Gamma _{Y,it}-\left( X_{it}-\Gamma
_{X,it}\right) ^{\prime }\beta \right) \in \mathbb{R}^{K},
\end{equation*}%
where we note that%
\begin{equation*}
Y_{it}-\Gamma _{Y,it}-\left( X_{i,t}-\Gamma _{X,it}\right) ^{\prime }\beta
_{0}=\varepsilon _{it}.
\end{equation*}%
Note that this is the same moment function used in the estimation of the
partially lineear model with cross-sectional data, c.f. \cite%
{robinson1988root}. It is easily checked that the data--generating
parameter, denoted by $\beta _{0}$, is identified as the unique solution to $%
\mathbb{E}\left[ m\left( Z_{it},\beta ,\Gamma _{Z,it}\right) \right] =0$
w.r.t. $\beta $,%
\begin{equation*}
\beta _{0}=\mathbb{E}\left[ \left( X_{it}-\Gamma _{X,it}\right) \left(
X_{it}-\Gamma _{X,it}\right) ^{\prime }\right] ^{-1}\mathbb{E}\left[ \left(
X_{it}-\Gamma _{X,it}\right) \left( Y_{it}-\Gamma _{Y,it}\right) \right] ,
\end{equation*}%
under the following assumptions:

\begin{assumption}
\label{ass:ID}$\mathbb{E}\left[ \left( X_{it}-\Gamma _{X,it}\right) \left(
X_{it}-\Gamma _{X,it}\right) ^{\prime }\right] $ has full rank; $\mathbb{E}%
[\varepsilon _{Z,it}|\left\{ X_{js},\alpha _{j},\gamma _{s}\right\} _{js}]=0$%
.
\end{assumption}

We here impose strict exogeneity. Identification holds under the weaker
condition of contemporaneous exogeneity, $\mathbb{E}[\varepsilon
_{Z,it}|X_{it},\alpha _{i},\gamma _{t}]=0$, but parts of our theoretical
analysis will make use of strict exogeneity to simplify the arguments.

Importantly, $m\left( Z_{i,t},\beta ,\Gamma _{Z,it}\right) $ is
Neyman-orthogonal w.r.t. $\Gamma _{Z,it}$ at $\beta =\beta _{0}$,%
\begin{eqnarray*}
\mathbb{E}\left[ \frac{m\left( Z_{it},\beta _{0},\Gamma _{Z,it}\right) }{%
\partial \Gamma _{Y,it}}\right] &=&-\mathbb{E}\left[ X_{i,t}-\Gamma _{X,it}%
\right] =0, \\
\mathbb{E}\left[ \frac{m\left( Z_{i,t},\beta _{0},\Gamma _{Z,it}\right) }{%
\partial \Gamma _{X,it}}\right] &=&\mathbb{E}\left[ X_{it}-\Gamma _{X,it}%
\right] ^{\prime }\beta _{0}-\mathbb{E}\left[ Y_{it}-\Gamma _{Y,it}-\left(
X_{i,t}-\Gamma _{X,it}\right) ^{\prime }\beta _{0}\right] =0.
\end{eqnarray*}%
An important consequence of this feature is that estimation of $\beta $
based on these moment conditions is less sensitive to the first-step
estimation of $\Gamma _{Z,it}$.

To develop such estimator, we here take as given any first--step estimators $%
\hat{\Gamma}_{Z}=(\hat{\Gamma}_{Y},\hat{\Gamma}_{X})$ of $\Gamma _{Z}$ as
chosen by the researcher. We then define our Neyman--orthogonal estimator of 
$\beta $ as the solution to $\sum_{it}m(Z_{it},\hat{\beta}_{NO},\hat{\Gamma}%
_{Z,it})=0$, where $\sum_{it}:=\sum_{i=1}^{N}\sum_{t=1}^{T}$, which takes
the form%
\begin{equation}
\hat{\beta}_{NO}=\left[ \sum_{it}(X_{it}-\widehat{\Gamma }_{X,it})(X_{it}-%
\widehat{\Gamma }_{X,it})^{\prime }\right] ^{-1}\sum_{it}(X_{it}-\widehat{%
\Gamma }_{X,it})(Y_{it}-\widehat{\Gamma }_{Y,it}),
\label{Eq: hat-beta_NO def}
\end{equation}%
Due to the moment conditions being orthogonal w.r.t. $\hat{\Gamma}_{Z}$,
this estimator of $\beta $ will be $\sqrt{NT}$-asymptotically normally
distributed convergence as long as $\hat{\Gamma}_{Z}-\Gamma _{Z}=o_{P}\left(
(NT)^{-1/4}\right) $. This is a well-known result from the literature on
semiparametric estimators.

However, in our setting, above rate result is in general not achievable when
we treat $\alpha _{i}$ and $f_{t}$ as unknown fixed effects since in this
case $\hat{\Gamma}_{Z}$ will involve estimates of these. To see this,
suppose that in fact $g$ in eq. (\ref{eqn:model}) is known to us, in which
case we have fully parametric fixed effects model whose parameters could be
estimated by $(\hat{\beta},\{\hat{\alpha}_{i}\}_{i=1}^{N},\{\hat{\gamma}%
_{t}\}_{t=1}^{T})=\arg \min_{\beta ,\{\alpha _{i}\}_{i=1}^{N},\{\gamma
_{t}\}_{t=1}^{T}}\sum_{it}(Y_{it}-X_{it}^{\prime }\beta -g(\alpha
_{i},\gamma _{t}))^{2}$. This in turn yields the estimator $\hat{\Gamma}%
_{it}=g(\hat{\alpha}_{i},\hat{\gamma}_{t})$. Except in a few special cases,
such as when $g(\alpha _{i},\gamma _{t})$ is linear or multiplicative in its
arguments, it is well-known that $\hat{\beta}$ will generally suffer from
incidental parameter biases, $\mathbb{E[}\hat{\beta}]\simeq \beta
_{0}+B_{\alpha }/T+B_{\gamma }/N$, where $B_{\alpha }$ and $B_{\gamma }$ are
due to the incidental parameter biases caused by $\{\hat{\alpha}%
_{i}\}_{i=1}^{N}$ and $\{\hat{\gamma}_{t}\}_{t=1}^{T}$, respectively; see,
e.g., Theorem 4.1 of \cite{FernandezValWeidner2016}. Obviously, in our
setting with $g$ unknown, we cannot hope to do better than in the parametric
submodel with $g$ known, unless we are willing to impose further
restrictions on the unknown fixed effects $\alpha _{i}$ and $\gamma _{t}$ or
the mapping $g$.

In light of above, we develop a general asymptotic result for $\hat{\beta}%
_{NO}$ when $\hat{\Gamma}_{Z,it}=\hat{g}_{Z}(\hat{\alpha}_{i},\hat{\gamma}%
_{t})$ contains non--negiglible biases due to $\left( \hat{\alpha}_{i},\hat{%
\gamma}_{t}\right) $ being used in place of $(\alpha _{i},\gamma _{t})$. We
first develop an expansion of $\hat{\beta}_{NO}$ w.r.t. $\hat{\Gamma}_{Z}$:
Applying the mean--value theorem to $\sum_{it}m(Z_{it},\hat{\beta}_{NO},\hat{%
\Gamma}_{Z,it})=0$,%
\begin{equation*}
0=\frac{1}{NT}\sum_{it}m\left( Z_{it},\beta _{0},\hat{\Gamma}_{Z,it}\right) +%
\frac{1}{NT}\sum_{it}\frac{\partial m\left( Z_{it},\beta _{0},\hat{\Gamma}%
_{Z,it}\right) }{\partial \beta ^{\prime }}(\hat{\beta}_{NO}-\beta _{0}),
\end{equation*}%
where%
\begin{equation*}
\frac{1}{NT}\sum_{it}\frac{\partial m\left( Z_{it},\beta _{0},\hat{\Gamma}%
_{Z,it}\right) }{\partial \beta ^{\prime }}=\frac{1}{NT}\sum_{it}\left(
X_{i,t}-\hat{\Gamma}_{X,it}\right) \left( X_{it}-\hat{\Gamma}_{X,it}\right)
^{\prime }.
\end{equation*}%
A second--order Taylor expansion of the right--hand side of above equation
w.r.t. $\hat{\Gamma}_{X,it}$ at $\Gamma _{X,it}$ yields%
\begin{align*}
\frac{1}{NT}\sum_{it}(X_{it}-\widehat{\Gamma }_{X,it})& (X_{it}-\widehat{%
\Gamma }_{X,it})^{\prime }=\frac{1}{NT}\sum_{it}\eta _{it}\eta _{it}^{\prime
}+\frac{1}{NT}\sum_{it}\eta _{it}(\hat{\Gamma}_{X,it}-\Gamma
_{X,it})^{\prime } \\
& +\frac{1}{NT}\sum_{it}(\hat{\Gamma}_{X,it}-\Gamma _{X,it})\eta
_{it}^{\prime }+\frac{1}{NT}\sum_{it}(\hat{\Gamma}_{X,it}-\Gamma _{X,it})(%
\hat{\Gamma}_{X,it}-\Gamma _{X,it})^{\prime },
\end{align*}%
where, assuming $\sum_{it}\left\Vert \eta _{it}\right\Vert ^{2}/\left(
NT\right) \rightarrow ^{p}\mathbb{E}\left[ \left\Vert \eta _{it}\right\Vert
^{2}\right] $ and $\sum_{it}\left\Vert \hat{\Gamma}_{it}-\Gamma
_{X,it}\right\Vert ^{2}/\left( NT\right) =o_{P}\left( 1\right) $, 
\begin{eqnarray*}
\left\Vert \frac{1}{NT}\sum_{it}\eta _{it}(\hat{\Gamma}_{it}-\Gamma
_{X,it})^{\prime }\right\Vert &\leq &\sqrt{\frac{1}{NT}\sum_{it}\left\Vert
\eta _{it}\right\Vert ^{2}\times \frac{1}{NT}\sum_{it}\left\Vert \hat{\Gamma}%
_{it}-\Gamma _{X,it}\right\Vert ^{2}} \\
&=&O_{P}\left( 1\right) \times o_{P}\left( 1\right) =o_{P}\left( 1\right) ,
\end{eqnarray*}%
and similar for the other terms. Thus, $\frac{1}{NT}\sum_{it}(X_{it}-\hat{%
\Gamma}_{X,it})(X_{it}-\hat{\Gamma}_{X,it})^{\prime }=\frac{1}{NT}%
\sum_{it}\eta _{it}\eta _{it}^{\prime }+o_{P}\left( 1\right) $.

Next, another second--order Taylor expansion combined with $Y_{i,t}-\Gamma
_{Y,it}=\beta _{0}^{\prime }\eta _{it}+\varepsilon _{i,t}$ gives us%
\begin{eqnarray*}
\frac{1}{NT}\sum_{it}m\left( Z_{it},\beta _{0},\hat{\Gamma}_{Z,it}\right) &=&%
\frac{1}{NT}\sum_{it}\eta _{it}\varepsilon _{i,t}+\frac{1}{NT}\sum_{it}(\hat{%
\Gamma}_{X,it}-\Gamma _{X,it})(\hat{\Gamma}_{Y,it}-\Gamma _{Y,it}) \\
&&+\frac{1}{NT}\sum_{it}\varepsilon _{it}(\hat{\Gamma}_{X,it}-\Gamma
_{X,it})+\frac{1}{NT}\sum_{it}\eta _{it}(\hat{\Gamma}_{Y,it}-\Gamma _{Y,it}),
\end{eqnarray*}%
where the last three terms contain the first and second--order effects of $%
\hat{\Gamma}_{Z,it}$. The following result provides conditions under which $%
\hat{\beta}_{NO}$ is asymptotically normally distributed but suffers from
incidental parameter biases due to the presence of these terms:

\begin{theorem}
\label{Th: NO general}Suppose that Assumption \ref{ass:ID} holds and that $%
\frac{1}{NT}\sum_{it}||{\Gamma }_{X,it}-\widehat{\Gamma }%
_{X,it}||^{2}=o_{p}(1)$, 
\begin{eqnarray}
\frac{1}{NT}\sum_{it}(\hat{\Gamma}_{X,it}-\Gamma _{X,it})(\hat{\Gamma}%
_{Y,it}-\Gamma _{Y,it}) &=&\frac{B_{\alpha ,1}}{T}+\frac{B_{\gamma ,1}}{N}%
+o_{P}\left( 1/\sqrt{NT}\right)  \label{eq: rate cond 0} \\
\frac{1}{NT}\sum_{it}\eta _{it}(\hat{\Gamma}_{Y,it}-\Gamma _{Y,it}) &=&\frac{%
B_{\alpha ,2}}{T}+\frac{B_{\gamma ,2}}{N}+o_{P}\left( 1/\sqrt{NT}\right) ,
\label{eq: rate cond 1} \\
\frac{1}{NT}\sum_{it}\varepsilon _{it}(\hat{\Gamma}_{X,it}-\Gamma _{X,it})
&=&\frac{B_{\alpha ,3}}{T}+\frac{B_{\gamma ,3}}{N}+o_{P}\left( 1/\sqrt{NT}%
\right)  \label{eq: rate cond 2} \\
\frac{1}{NT}\sum_{it}\eta _{it}\eta _{it}^{\prime }\rightarrow ^{P}\Omega
_{X}:=\mathbb{E}\left[ \eta _{it}\eta _{it}^{\prime }\right] , &&\frac{1}{%
\sqrt{NT}}\sum_{it}\eta _{it}\varepsilon _{i,t}\rightarrow ^{d}\mathcal{N}%
\left( 0,\Sigma \right) .  \label{eq: rate cond 3}
\end{eqnarray}%
Then, with $B_{\alpha }=\Omega _{X}^{-1}\sum_{k=1}^{3}B_{\alpha ,k}$ and $%
B_{\gamma }=\Omega _{X}^{-1}\sum_{k=1}^{3}B_{\gamma ,k}$,%
\begin{equation*}
\sqrt{NT}(\hat{\beta}_{NO}-\beta _{0}-B_{\alpha }/T-B_{\gamma
}/N)\rightarrow ^{d}\mathcal{N}\left( 0,\Omega _{X}^{-1}\Sigma \Omega
_{X}^{-1}\right) .
\end{equation*}
\end{theorem}

\begin{remark}
Sufficient conditions for the LLN and the CLT stated in eq. (\ref{eq: rate
cond 3}) to hold can be found in, e.g., \cite{Bai2009} and \cite%
{hahn2011bias}. These sufficient conditions allow for potential
cross-sectional and time series correlation/dependence and
heteroskedasticity.
\end{remark}

The result is quite general and allows for a broad class of estimators ${%
\hat{\Gamma}}_{Z,it}$. The main restriction is the implicit assumption that
the leading bias components of this estimator are of order $O\left(
1/T\right) +O\left( 1/N\right) $. The $o_{P}(1/\sqrt{NT})$ terms in
equations (\ref{eq: rate cond 0})--(\ref{eq: rate cond 3}) capture both the
full error from the nonparametric estimation of $g$ together with the
variance components of the fixed effects estimators. As we shall see, this
rate result will hold under weak restrictions on $g$, $\alpha _{i}$ and $%
\gamma _{t}$ and their estimators. The asymptotic biases $B_{\alpha }/T$ and 
$B_{\gamma }/N$ will be present due to first--order biases of the estimators
of $\alpha _{i}$ and $\gamma _{t}$ together with potentially covariation
between their estimation errors and $(\eta _{it},\varepsilon _{it})$, c.f.
Theorem 4.1 of \cite{FernandezValWeidner2016}.

Below, we first give more primitive conditions under which (\ref{eq: rate
cond 0})--(\ref{eq: rate cond 3}) will hold. Next, we show how sample
splitting combined with bias adjustment can be employed to remove the
incidental parameter biases.

\subsection{Sources of incidental parameter biases}

We here focus on a particular class of first-step estimators of $\Gamma _{Z}$
that combine fixed-effects estimation with nonparametric regression
techniques and for these provide a partial characterisation of the
incidental parameter biases introduced in (\ref{eq: rate cond 0})--(\ref{eq:
rate cond 2}). The class of estimators and accompanying theory includes as a
special case the particular estimators that we develop in Section \ref%
{sect:FM}. The estimators of \cite{beyhum2025inference} also fit into our
general framework. However, parts of our analysis requires the estimator of $%
g_{Z}$ to be sufficiently regular (smooth) while their procedure involves
discretisation and so these parts do not apply to their estimator.

We first note that, without further normalisations and restrictions on the
model, we cannot separately identify $g_{Z}$, $\alpha _{i}$ and $\gamma _{t}$%
. We will here assume that suitably identifying restrictions/normalisations
have been developed so that the following conditions are satisfied:

First, we will require that there exists a normalised version of $g_{Z}$,
denoted $g_{0,Z}$, and associated normalised version of the fixed effects $%
\left( \alpha _{i},\gamma _{t}\right) $, denoted $\left( \lambda
_{i},f_{t}\right) $, so that%
\begin{equation}
\text{(i) \ }\Gamma _{Z,it}=g_{0,Z}\left( \lambda _{i},f_{t}\right) ,\text{
\ \ (ii) }\left( \lambda _{i},f_{t}\right) \text{ are identified in the
population}  \label{eq: ID restrictions}
\end{equation}%
Second, we will require that fixed-effects estimators $(\hat{\lambda}_{i},%
\hat{f}_{t})$ of $\left( \lambda _{i},f_{t}\right) $ are available to us.
The precise forms of the normalised versions and their estimators depend on
the assumptions the researcher is willing to impose on the model. One
example of such is provided in Section \ref{sect:FM}; see \cite%
{FreemanKristensen2026} for further details. We will focus on the ideal
setting where the estimators are regular in the sense that they satisfy%
\begin{equation}
\hat{\lambda}_{i}=\lambda _{i}+\frac{1}{T}\sum_{t=1}^{T}\psi _{\lambda
}\left( Z_{it}\right) +\frac{1}{T}b_{\lambda ,i}+o_{P}\left( 1/T\right) ,%
\text{ \ \ }\hat{f}_{t}=f_{t}+\frac{1}{N}\sum_{i=1}^{N}\psi _{f}\left(
Z_{it}\right) +\frac{1}{N}b_{f,t}+o_{P}\left( 1/N\right) ,
\label{eq: lambda and f expan}
\end{equation}%
where $\mathbb{E}\left[ \psi _{\lambda }\left( Z_{it}\right) \right] =%
\mathbb{E}\left[ \psi _{f}\left( Z_{it}\right) \right] =0$ and $b_{\lambda
,i}$ and $b_{f,t}$ capture the leading bias terms of the two estimators; see 
\cite{FernandezValWeidner2016} for sufficient conditions for above to hold
in a parametric setting. The fixed-effects estimators of \cite%
{beyhum2025inference} also satisfy above with $b_{\lambda ,i}=b_{f,t}=0$.-
Under above assumption, we will now provide a characterisation of the biases
appearing in (\ref{eq: rate cond 0})--(\ref{eq: rate cond 2}). We expect all
qualitative conclusions reached in the following to also apply to irregular
estimators but the precise forms of the biases contained in the first and
second-order terms in (\ref{eq: rate cond 0})--(\ref{eq: rate cond 2}) will
in this case become more complicated, including their rates which no longer
will be of the parametric kind.

Given this set-up, we will then consider estimators of $\Gamma _{Z,it}$ that
runs a nonparametric regression of $Z_{it}$ onto the estimated fixed effects 
$(\hat{\lambda}_{i},\hat{f}_{t})$ to obtain $\hat{g}_{0,Z}$. This estimator
could take many forms; one example would be kernel regression, in which case%
\begin{equation}
\hat{\Gamma}_{Z,it}=\hat{g}_{0,Z}(\hat{\lambda}_{i},\hat{f}_{t}),\text{ \ \ }%
\hat{g}_{0,Z}\left( \lambda ,f\right) =\frac{\sum_{it}Z_{it}K_{1,h_{\lambda
}}(\hat{\lambda}_{i}-\lambda )K_{2,h_{f}}(\hat{f}_{t}-f)}{%
\sum_{it}K_{1,h_{\lambda }}(\hat{\lambda}_{i}-\lambda )K_{2,h_{f}}(\hat{f}%
_{t}-f)},  \label{eq: g_Z kernel reg}
\end{equation}%
where $K_{j}$, $j=1,2$, are kernel functions and $h_{\lambda },h_{f}>0$ are
bandwidths, but the subsequent theory allows for other nonparametric
regresion techniques. It is here important to note that the estimated fixed
effects enter $\hat{g}_{0,Z}(\hat{\lambda}_{i},\hat{f}_{t})$ through two
channels: First, as the values $(\hat{\lambda}_{i},\hat{f}_{t})$ of the
argument $\left( \lambda ,f\right) $ in $\hat{g}_{0,Z}\left( \lambda
,f\right) $; second, as the regressors used to compute $\hat{g}_{0,Z}\left(
\lambda ,f\right) $. Each channel could potentially generate incidental
parameter biases. The first channel is well-known, c.f. \cite%
{FernandezValWeidner2016}, while the second one is novel and has only been
analysed in the context of nonparametric estimation of the density of
one--way fixed effects; see, e.g., \cite{Okui2019} and \cite{Barras2021}.

We first analyse the first--order terms appearing in eqs. (\ref{eq: rate
cond 1})--(\ref{eq: rate cond 2}). For a given function $g_{Z}$, define 
\begin{equation}
\hat{\nu}\left( g_{Z}\right) :=\frac{1}{NT}\sum_{it}\left[ 
\begin{array}{c}
\varepsilon _{it}g_{X}(\hat{\lambda}_{i},\hat{f}_{t}) \\ 
\eta _{it}g_{Y}(\hat{\lambda}_{i},\hat{f}_{t})%
\end{array}%
\right] ,\text{ \ \ }\hat{\nu}^{\ast }\left( g_{Z}\right) :=\frac{1}{NT}%
\sum_{it}\left[ 
\begin{array}{c}
\varepsilon _{it}g_{X}(\lambda _{i},f_{t}) \\ 
\eta _{it}g_{Y}(\lambda _{i},f_{t})%
\end{array}%
\right] .  \label{eq: emp proc def}
\end{equation}%
Let $\hat{g}_{0,Z}$ denote the feasible nonparametric estimator of $g_{0,Z}$
that takes as input $\{\hat{\lambda}_{i}\}_{i=1}^{N}$ and $\{\hat{f}%
_{t}\}_{t=1}^{T}$, such as the kernel regression estimator in (\ref{eq: g_Z
kernel reg}), while $\hat{g}_{0,Z}^{\ast }$ is the infeasible version that
takes as input $\{\lambda _{i}\}_{i=1}^{N}$ and $\{f_{t}\}_{t=1}^{T}$. We
can then decompose the vector of first-order error terms in eqs. (\ref{eq:
rate cond 1})--(\ref{eq: rate cond 2}) as%
\begin{equation}
\frac{1}{NT}\sum_{it}\left[ 
\begin{array}{c}
\varepsilon _{it}(\hat{\Gamma}_{X,it}-\Gamma _{X,it}) \\ 
\eta _{it}(\hat{\Gamma}_{Y,it}-\Gamma _{Y,it})%
\end{array}%
\right] =\left\{ \hat{\nu}(\hat{g}_{0,Z})-\hat{\nu}^{\ast }(\hat{g}%
_{0,Z})\right\} +\left\{ \hat{\nu}^{\ast }(\hat{g}_{0,Z})-\hat{\nu}^{\ast }(%
\hat{g}_{0,Z}^{\ast })\right\} +\left\{ \hat{\nu}^{\ast }\left( \hat{g}%
_{0,Z}^{\ast }\right) -\hat{\nu}^{\ast }\left( g_{0,Z}\right) \right\} .
\label{eq: bias decomp}
\end{equation}%
The first two terms on the right--hand side of above display contain the
estimation errors due to $\{\hat{\lambda}_{i}\}_{i=1}^{N}$ and $\{\hat{f}%
_{t}\}_{t=1}^{T}$, while the third term contains the error of the infeasible
nonparametric estimator $\hat{g}_{0,Z}^{\ast }$; under weak regularity
conditions, the third term will be negiglible:

\begin{lemma}
\label{Lem: Stoch Eq}Suppose that $g_{0,Z},\hat{g}_{0,Z}\in \left( \mathcal{G%
},\left\Vert \cdot \right\Vert _{\mathcal{G}}\right) $ w.p.a.1 and $%
\left\Vert \hat{g}_{0,Z}^{\ast }-g_{0,Z}\right\Vert _{\mathcal{G}%
}=o_{P}\left( 1\right) $; $g_{Z}\mapsto \sqrt{NT}\hat{\nu}^{\ast }(g_{Z})$
defined in (\ref{eq: emp proc def}) is stochastically equicontinuous at $%
g_{0,Z}$ w.r.t. $\left\Vert \cdot \right\Vert _{\mathcal{G}}$. Then $\hat{\nu%
}^{\ast }\left( \hat{g}_{0,Z}\right) -\hat{\nu}^{\ast }\left( g_{0,Z}\right)
=o_{P}\left( 1/\sqrt{NT}\right) $.

Suppose that $\left\{ \varepsilon _{Z,it},\lambda _{i},f_{t}\right\} _{i,t}$
are mutually independent across $i$ and $t$ with $\max_{i,t}\mathbb{E}\left[
\left\Vert \varepsilon _{Z,it}\right\Vert ^{2}\right] <\infty $. Then the
stochastic equicontinuity condition holds if $(\mathcal{G},\left\Vert \cdot
\right\Vert _{\mathcal{G}})$ is chosen as the $L_{2}$-Sobolev space defined
in Eq. (2.14) of \cite{Andrews1994}.
\end{lemma}

The Donsker-type condition of \cite{Andrews1994} cited in above lemma to
ensure stochastic equicontinuity is satisfied by many nonparametric
regression estimators, including above kernel regression estimators. On the
other hand, it is unclear whether the k-means clustering algorithm of \cite%
{beyhum2025inference} satisfies a similar condition; in particular, their
estimator is by construction non--smooth and so will not belong to the $%
L_{2} $-Sobolev space of \cite{Andrews1994}. As such, their estimator may
not be covered by above lemma.

What remains is the analysis of the first two terms in eq. (\ref{eq: bias
decomp}). Consider the first component of the vector $\hat{\nu}(\hat{g}%
_{0,Z})-\hat{\nu}^{\ast }(\hat{g}_{0,Z})$, denoted $\hat{\nu}_{Y}\left( \hat{%
g}_{0,X}\right) -\hat{\nu}_{Y}^{\ast }(\hat{g}_{0,X}):=\frac{1}{NT}%
\sum_{it}\varepsilon _{it}\{\hat{g}_{X}(\hat{\lambda}_{i},\hat{f}_{t})-\hat{g%
}_{0,X}(\lambda _{i},f_{t})\}$. A second--order Taylor expansion yields 
\begin{align}
\hat{\nu}_{Y}\left( \hat{g}_{0,X}\right) -\hat{\nu}_{Y}^{\ast }(\hat{g}%
_{0,X})& =\frac{1}{NT}\sum_{it}\varepsilon _{it}\frac{\partial \hat{g}%
_{0,X}(\lambda _{i},f_{t})}{\partial \lambda ^{\prime }}(\hat{\lambda}%
_{i}-\lambda _{i})+\frac{1}{NT}\sum_{it}\varepsilon _{it}\frac{\partial \hat{%
g}_{0,X}(\lambda _{i},f_{t})}{\partial f^{\prime }}(\hat{f}_{t}-f_{t})
\label{eq: bias 1} \\
& +\frac{1}{2NT}\sum_{it}\varepsilon _{it}(\hat{\lambda}_{i}-\lambda
_{i})^{\prime }\frac{\partial ^{2}\hat{g}_{0,X}(\lambda _{i},f_{t})}{%
\partial \lambda \partial \lambda ^{\prime }}(\hat{\lambda}_{i}-\lambda _{i})
\notag \\
& +\frac{1}{2NT}\sum_{it}\varepsilon _{it}(\hat{f}_{t}-f_{t})^{\prime }\frac{%
\partial ^{2}\hat{g}_{0,X}(\lambda _{i},f_{t})}{\partial f\partial f^{\prime
}}(\hat{f}_{t}-f_{t})+R_{N,T},  \notag
\end{align}%
where the remainder term $R_{N,T}$ in great generality will be of higher
order than the other terms in the final expression. A similar expansion
holds for the second component of $\hat{\nu}(\hat{g}_{0,Z})-\hat{\nu}^{\ast
}(\hat{g}_{0,Z})$, $\hat{\nu}_{X}\left( \hat{g}_{0,Y}\right) -\hat{\nu}%
_{X}^{\ast }(\hat{g}_{0,Y}):=\frac{1}{NT}\sum_{it}\eta _{it}\{\hat{g}_{0,Y}(%
\hat{\lambda}_{i},\hat{f}_{t})-\hat{g}_{0,Y}(\lambda _{i},f_{t})\}$. The
leading terms of these expansions will create incidental parameter biases if
the first-order bias and/or variance components of $\hat{\lambda}_{i}$ or $%
\hat{f}_{t}$ in (\ref{eq: lambda and f expan}) correlate with $\varepsilon
_{Z,it}$. The precise expressions of these bias components can be derived
using arguments similar to the ones in \cite{FernandezValWeidner2016}. At
the same time, $\hat{\nu}(\hat{g}_{0,Z})-\hat{\nu}^{\ast }(\hat{g}_{0,Z})$
will generally not contribute to the variance of $\hat{\beta}$ due to the
use of Neyman orthogonality of the moment conditions.

The analysis of the second term in eq. (\ref{eq: bias decomp}) requires us
to take a firmer stand on the nonparametric regression method being
employed. For example, if kernel regression with a second--order kernel is
employed, then we expect 
\begin{equation}
\mathbb{E}\left[ \hat{g}_{0,Z}\left( \lambda ,f\right) -\hat{g}_{0,Z}^{\ast
}\left( \lambda ,f\right) \right] =\frac{1}{T}b_{g,1}\left( \lambda
,f\right) +\frac{1}{N}b_{g,2}\left( \lambda ,f\right) +o\left( 1/T\right)
+o\left( 1/N\right) ,  \label{eq: nonpar reg bias}
\end{equation}%
for some functions $b_{g,1}\left( \lambda ,f\right) $ and $b_{g,2}\left(
\lambda ,f\right) $. The expressions of these can be obtained by combining (%
\ref{eq: lambda and f expan}) with the arguments of \cite{Okui2019} and \cite%
{Barras2021}. Given that $\varepsilon _{it}$ and $\eta _{it}$ do not
correlate with $\lambda _{i}$ and $f_{t}$, these terms will not contribute
to the incidental parameter biases in (\ref{eq: rate cond 1})--(\ref{eq:
rate cond 2}). We expect similar results to hold for other nonparametric
regression estimators. In conclusion, the incidental parameter biases $%
B_{\alpha ,k}/T+B_{\gamma ,k}/N$, $k=2,3$, in Theorem \ref{Th: NO general}
are expected to arrive from $\hat{\nu}(\hat{g}_{0,A})-\hat{\nu}^{\ast }(\hat{%
g}_{0,A})$ if $\hat{\lambda}_{i}$ or $\hat{f}_{t}$ correlate with $%
\varepsilon _{Z,it}$. Moreover, again due to the Neyman orthogonal moment
conditions, the variance component of $\hat{\nu}(\hat{g}_{0,Z})-\hat{\nu}%
^{\ast }(\hat{g}_{0,Z}^{\ast })$ will generally be asymptotically negiglible.

For the second--order term in (\ref{eq: rate cond 0}), we can carry out a
similar decomposition. Assuming that $\hat{g}_{0,Z}^{\ast }$ is sufficiently
regular, $\frac{1}{NT}\sum_{it}\left\Vert \hat{g}_{0,Z}^{\ast }(\lambda
_{i},f_{t})-g_{0,Z}(\lambda _{i},f_{t})))\right\Vert ^{2}=o_{P}(1/\sqrt{NT})$%
; most known nonparametric regression estimators will satisfy this under
standard regularity conditions. Moreover, we expect the second--order
quadratic term will contain an incidental parameter bias terms on the form 
\begin{eqnarray}
\frac{1}{NT}\sum_{it}\frac{\partial g_{0,Y}(\lambda _{i},f_{t})}{\partial
\lambda ^{\prime }}(\hat{\lambda}_{i}-\lambda _{i})(\hat{\lambda}%
_{i}-\lambda _{i})^{\prime }\frac{\partial g_{0,Y}(\lambda _{i},f_{t})}{%
\partial \lambda } &=&\frac{B_{\alpha ,1}}{T}+o_{P}\left( 1/T\right) ,
\label{eq: bias 2} \\
\frac{1}{NT}\sum_{it}\frac{\partial g_{0,Y}(\lambda _{i},f_{t})}{\partial
f^{\prime }}(\hat{f}_{t}-f_{t})(\hat{f}_{t}-f_{t})^{\prime }\frac{\partial
g_{0,Y}(\lambda _{i},f_{t})}{\partial f} &=&\frac{B_{\gamma ,1}}{N}%
+o_{P}\left( 1/N\right) ,  \notag
\end{eqnarray}%
c.f. \cite{FernandezValWeidner2016}. We also expect $\hat{g}_{0,Z}\left(
\lambda ,f\right) -\hat{g}_{0,Z}^{\ast }\left( \lambda ,f\right) $ to
contribute; for example, for kernel--based estimators, we expect $\frac{1}{NT%
}\sum_{it}(\hat{g}_{0,Y}(\lambda _{i},f_{t})-\hat{g}_{0,Y}^{\ast }(\lambda
_{i},f_{t}))(\hat{g}_{0,X}(\lambda _{i},f_{t})-\hat{g}_{0,X}^{\ast }(\lambda
_{i},f_{t}))^{\prime }=O_{P}(1/N)+O_{P}(1/T)$ under weak conditions, c.f. 
\cite{Okui2019} and \cite{Barras2021}.

One could now attempt to carry out a more complete characterisation of the
incidental parameter biases and then use the resulting expressions of the
biases to carry out bias adjustment of $\hat{\beta}_{NO}$ in Theorem \ref%
{Th: NO general}. We will refrain from carrying out such an analysis since
this will require us to restrict ourselves to the ideal scenario of regular
estimators satisfying (\ref{eq: lambda and f expan}), and take a stand on
the precise form of $\hat{g}_{0,Z}$.

\subsection{Two--way nonparametric regression}

As explained in the previous subsection, the incidental parameter biases
will arise due to the use of $\hat{g}_{0,Z}(\hat{\lambda}_{i},\hat{f}_{t})$
in place of $\hat{g}_{0,Z}^{\ast }\left( \lambda _{i},f_{t}\right) $. We
here explain how these biases can be adjusted for by employing a generalised
version of the two-way kernel regression estimator of \cite%
{freeman2023linear,freeman2022multidimensional}.

For a given nonparametric regression technique, let $\hat{g}_{0,Z}(\lambda
,f)$ be the full--sample version that regresses $Z_{js}$ onto $(\hat{\lambda}%
_{j},\hat{f}_{s})$ for $j=1,...,N$ and $s=1,...,T$. Next, for each $%
t=1,...,T $, let $\hat{g}_{0,Z}^{\left( 1\right) }(\lambda ,f_{t})$ be the
estimator obtained by regressing $Z_{jt}=g_{0,Z}(\lambda
_{j},f_{t})+\varepsilon _{Z,it}$ onto $\hat{\lambda}_{j}$ for $j=1,...,N$;
that is, we run $T$ nonparametric regressions, each along the
cross-sectional dimension. Finally, for each $i=1,...,N$, let $\hat{g}%
_{0,Z}^{\left( 2\right) }(\lambda _{i},f)$ be the estimator obtained by
regressing $Z_{is}$ onto $\hat{f}_{s}$ for $s=1,...,T$; that is, we run $N$
nonparametric regressions, each along the time dimension. We then combine
these to obtain the following two--way nonparametric regression estimator, 
\begin{equation}
\hat{\Gamma}_{0,Z,it}^{TW}=\hat{g}_{0,Z}^{\left( 1\right) }(\hat{\lambda}%
_{i},f_{t})+\hat{g}_{0,Z}^{\left( 2\right) }(\lambda _{i},\hat{f}_{t})-\hat{g%
}_{0,Z}(\hat{\lambda}_{i},\hat{f}_{t}).  \label{eq: two-way regression}
\end{equation}%
When the kernel regression estimator in (\ref{eq: g_Z kernel reg}) is
employed, it takes the form 
\begin{equation}
\hat{\Gamma}_{0,Z,it}^{TW}=\frac{\sum_{j}Z_{jt}K_{1,h_{\lambda }}(\hat{%
\lambda}_{j}-\hat{\lambda}_{i})}{\sum_{j}K_{1,h_{\lambda }}(\hat{\lambda}%
_{j}-\hat{\lambda}_{i})}+\frac{\sum_{t}Z_{is}K_{2,h_{f}}(\hat{f}_{s}-\hat{f}%
_{t})}{\sum_{t}K_{2,h_{f}}(\hat{f}_{s}-\hat{f}_{t})}-\frac{%
\sum_{js}Z_{js}K_{1,h_{\lambda }}(\hat{\lambda}_{j}-\hat{\lambda}%
_{i})K_{2,h_{f}}(\hat{f}_{s}-\hat{f}_{t})}{\sum_{js}K_{1,h_{\lambda }}(\hat{%
\lambda}_{j}-\hat{\lambda}_{i})K_{2,h_{f}}(\hat{f}_{s}-\hat{f}_{t})}.
\label{eq: two-way kernel regression}
\end{equation}

The two--way estimator has the following two attractive features: First, if (%
\ref{eq: nonpar reg bias}) holds, then under great generality,%
\begin{equation}
\mathbb{E}\left[ \hat{g}_{0,Z}^{\left( 1\right) }(\lambda ,f_{t})-\hat{g}%
_{0,Z}^{\ast }\left( \lambda ,f_{t}\right) \right] =\frac{1}{T}b_{\lambda
}\left( \lambda ,f\right) +o\left( 1/T\right) ,\text{ \ \ }\mathbb{E}\left[ 
\hat{g}_{0,Z}^{\left( 2\right) }(\lambda ,f_{t})-\hat{g}_{0,Z}^{\ast }\left(
\lambda ,f_{t}\right) \right] =\frac{1}{N}b_{f}\left( \lambda ,f\right)
+o\left( 1/N\right)
\end{equation}%
Thus, $\mathbb{E}[\hat{\Gamma}_{0,Z,it}^{TW}|\hat{\lambda}_{i},\hat{f}%
_{t}]=o\left( 1/T\right) $ and so the leading biases due to $\{\hat{\lambda}%
_{j},\hat{f}_{s}\}_{js}$ being used in the nonparametric regression has been
removed. One can think this of this as a type of Jackknifing. Second,
assuming the first order derivatives of $\hat{g}_{0,Z}$ and $\hat{\lambda}%
_{i},\hat{f}_{t}$ are consistent,%
\begin{eqnarray*}
\frac{\partial \hat{\Gamma}_{0,Z,it}^{TW}}{\partial \hat{\lambda}_{i}} &=&%
\frac{\partial \hat{g}_{0,Z}^{\left( 1\right) }(\hat{\lambda}_{i},f_{t})}{%
\partial \hat{\lambda}_{i}}-\frac{\partial \hat{g}_{0,Z}(\hat{\lambda}_{i},%
\hat{f}_{t})}{\partial \hat{\lambda}_{i}}=\frac{\partial g_{0,Z}(\lambda
_{i},f_{t})}{\partial \lambda _{i}}-\frac{\partial g_{0,Z}(\lambda
_{i},f_{t})}{\partial \lambda _{i}}+o_{P}\left( 1\right) =o_{P}\left(
1\right) , \\
\frac{\partial \hat{\Gamma}_{0,Z,it}^{TW}}{\partial \hat{f}_{t}} &=&\frac{%
\partial \hat{g}_{0,Z}^{\left( 2\right) }(\lambda _{i},\hat{f}_{t})}{%
\partial \hat{f}_{t}}-\frac{\partial \hat{g}_{0,Z}(\hat{\lambda}_{i},\hat{f}%
_{t})}{\partial \hat{f}_{t}}=\frac{\partial g_{0,Z}(\lambda _{i},f_{t})}{%
\partial f_{t}}-\frac{\partial g_{0,Z}(\lambda _{i},f_{t})}{\partial f_{t}}%
+o_{P}\left( 1\right) =o_{P}\left( 1\right) ,
\end{eqnarray*}%
and similarly for the second-order derivatives. As such, $\hat{\Gamma}%
_{0,Z,it}^{TW}$ is Neyman orthogonal to $(\hat{\lambda}_{i},\hat{f}_{t})$.
As a consequence, if we carry out the expansions leading to eqs. (\ref{eq:
bias 1}) and (\ref{eq: bias 2}) with $\hat{g}_{0,Z}(\hat{\lambda}_{i},\hat{f}%
_{t})$ replaced by $\hat{\Gamma}_{0,Z,it}^{TW}$, all the partial derivatives
in the final expressions are now $o_{P}\left( 1\right) $ and so the leading
biases due to $(\hat{\lambda}_{i},\hat{f}_{t})$ being used as arguments in $%
\hat{g}_{0,Z}(\lambda ,f)$ also become asymptotically negiglible. All
together, we expect $\hat{\beta}^{NO}$ that takes $\hat{\Gamma}%
_{0,Z,it}^{TW} $ as input not to suffer from any incidental parameter biases.

The two--way estimator could also be employed in other settings. For
example, \cite{deaner2025inferring} feeds the following sample
pseudo--metrics of \cite{zhang2017estimating},%
\begin{equation}
{\hat{d}}_{Y,ij}^{(1)}=\frac{1}{T}\max_{k\notin \{i,j\}}\left\vert
\left\langle Y_{k,1:T},Y_{i,1:T}-Y_{i,1:T}\right\rangle \right\vert ,\text{
\ \ }{\hat{d}}_{Y,st}^{(2)}=\frac{1}{N}\max_{u\notin \{s,t\}}\left\vert
\left\langle Y_{1:N,u},Y_{1:N,s}-Y_{1:N,t}\right\rangle \right\vert 
\label{eqn:ZhangDistance}
\end{equation}%
into a kernel regression estimator. Their kernel regression estimator could
then be combined with above two--way estimator to remove parts of the
estimation error due to the "generated" kernel regressors ${\hat{d}}%
_{Y,ij}^{(1)}$ and ${\hat{d}}_{Y,st}^{(2)}$. Simulation evidence strongly
suggests that this indeed leads to substantial improvements in the resulting
regression estimator. This metric is, however, sensitive to the dimension of 
$\alpha _{i}$ and $\gamma _{t}$.\footnote{%
See the results of our simulation study in Section \ref{sect:AppSims}.}

\subsection{Debiasing by sample splitting}

\label{sect:samplesplit}

In order to formalise the arguments in the past two subsections and obtain
precise expressions of the incidental parameter biases, we would need to
impose restrictions on the precise form of $(\hat{\lambda}_{i},\hat{f}_{t})$
and $\hat{g}_{0,Z}$. We here develop sample--splitting procedures that
avoids us having to do so: First, sample splitting removes the leading
biases caused by the first-order terms under very weak conditions on the
chosen nonparametric regression method. Second, sample--splitting allows us
to show that the two--way regression method removes the leading biases of
the second--order term under weak conditions on $(\hat{\lambda}_{i},\hat{f}%
_{t})$ and $\hat{g}_{0,Z}$.

We develop two sample--splitting procedures. The first is standard in the
literature and works for any, possibly non--linear and possibly
non--regression--based estimators of $\Gamma _{Z,it}$; this version is able
to remove the incidental parameter biases incurred from the two first-order
terms in (\ref{eq: rate cond 1})--(\ref{eq: rate cond 2}), but the biases
arising from the second--order term in (\ref{eq: rate cond 0}) may still be
present. The second version appears to be new to the panel data literature
and applies to linear estimators of $g_{0,Z}$; this version will be able to
remove the biases from both the linear and second-order terms when combined
with the two--way estimator.

\subsubsection{Standard sample splitting}

We follow \cite{freeman2023linear}, amongst others, and here consider the
following sample splitting procedure: Define the following sample splits in
terms of the indices of the data points,%
\begin{equation}
\begin{split}
\mathcal{I}_{1,1}& =\left\{ i=1,....,N/2,t=1,....,T/2\right\} , \\
\mathcal{I}_{1,2}& =\left\{ i=1,....,N/2,t=T/2+1,....,T\right\} , \\
\mathcal{I}_{2,1}& =\left\{ i=N/2+1,....,N,t=1,....,T/2\right\} , \\
\mathcal{I}_{2,2}& =\left\{ i=N/2+1,....,N,t=T/2+1,....,T\right\} ,
\end{split}
\label{eqn:sampleSplit}
\end{equation}%
and $\mathcal{I=}\left\{ i=1,....,N,t=1,....,T\right\} $. We then compute
our first--step estimators so it does not use data from $\mathcal{I}%
_{k_{1},k_{2}}$,%
\begin{equation}
\hat{\Gamma}_{Z,it}^{\left( k_{1},k_{2}\right) }\in \mathcal{F}%
_{k_{1},k_{2}}:=\mathcal{F}\left\{ Z_{js}:\left( j,s\right) \in \mathcal{I}%
\backslash \mathcal{I}_{k_{1},k_{2}}\right\} ,\text{ \ }\left( i,t\right)
\in \mathcal{I}_{k_{1},k_{2}},  \label{eq: Gamma-hat sample splitting}
\end{equation}%
for $k_{1},k_{2}=1,2$, and use these to obtain%
\begin{eqnarray}
\hat{\beta}_{NO}^{SS} &=&\left[ \sum_{k_{1},k_{2}=1}^{2}\sum_{\left(
i,t\right) \in \mathcal{I}_{k_{1},k_{2}}}(X_{it}-\hat{\Gamma}_{X,it}^{\left(
k_{1},k_{2}\right) })(X_{it}-\hat{\Gamma}_{X,it}^{\left( k_{1},k_{2}\right)
})^{\prime }\right] ^{-1}  \notag \\
&&\times \sum_{k_{1},k_{2}=1}^{2}\sum_{\left( i,t\right) \in \mathcal{I}%
_{k_{1},k_{2}}}(X_{it}-\hat{\Gamma}_{X,it}^{\left( k_{1},k_{2}\right)
})(Y_{it}-\hat{\Gamma}_{Y,it}^{\left( k_{1},k_{2}\right) }).
\label{eq: beta-hat sample splitting}
\end{eqnarray}%
We can still apply the same arguments that lead to Theorem \ref{Th: NO
general} to $\hat{\beta}_{NO}^{SS}$ and so the rate conditions stated in (%
\ref{eq: rate cond 0})--(\ref{eq: rate cond 3}) still hold, except that in
the left-hand side expressions $\sum_{it}$ and $\hat{\Gamma}_{Z,it}$ are now
replaced by $\sum_{\left( i,t\right) \in \mathcal{I}_{k_{1},k_{2}}}$ and $%
\hat{\Gamma}_{Z,it}^{\left( k_{1},k_{2}\right) }$, respectively. Eqs. (\ref%
{eq: rate cond 1})--(\ref{eq: rate cond 2}) will hold under the following
weak restrictions on the variances of the two error terms, where $%
\varepsilon ^{\left( k_{1},k_{2}\right) }=\left\{ \varepsilon _{it}:\left(
i,t\right) \in \mathcal{I}_{k_{1},k_{2}}\right\} $ and $\eta ^{\left(
k_{1},k_{2}\right) }=\left\{ \eta _{it}:\left( i,t\right) \in \mathcal{I}%
_{k_{1},k_{2}}\right\} $:

\begin{theorem}
\label{thm:sampleSplit} Suppose Assumption \ref{ass:ID} and eqs. (\ref{eq:
rate cond 0}) and (\ref{eq: rate cond 3}) hold. Suppose further, for $%
k_{1},k_{2}=1,2$, 
\begin{eqnarray}
\frac{1}{NT}\sum_{\left( i,t\right) \in \mathcal{I}_{k_{1},k_{2}}}\left\vert 
\mathbb{E}\left[ \varepsilon _{it}|\mathcal{F}^{\left( k_{1},k_{2}\right) },%
\mathcal{G}\right] \right\vert &=&O_{P}\left( \frac{1}{NT}\right) ,
\label{eq: sample split cond} \\
\frac{1}{NT}\sum_{\left( i,t\right) \in \mathcal{I}_{k_{1},k_{2}}}\left\Vert 
\mathbb{E}\left[ \eta _{it}|\mathcal{F}^{\left( k_{1},k_{2}\right) },%
\mathcal{G}\right] \right\Vert &=&O_{P}\left( \frac{1}{NT}\right) ,  \notag
\\
\left\Vert \mathbb{E}\left[ vec\left( \varepsilon ^{\left(
k_{1},k_{2}\right) }\right) vec\left( \varepsilon ^{\left(
k_{1},k_{2}\right) }\right) ^{\prime }|\mathcal{F}^{\left(
k_{1},k_{2}\right) },\mathcal{G}\right] \right\Vert _{op} &=&O_{P}\left( 
\frac{1}{NT}\right) ,  \notag \\
\left\Vert \mathbb{E}\left[ vec\left( \eta ^{\left( k_{1},k_{2}\right)
}\right) vec\left( \eta ^{\left( k_{1},k_{2}\right) }\right) ^{\prime }|%
\mathcal{F}^{\left( k_{1},k_{2}\right) },\mathcal{G}\right] \right\Vert
_{op} &=&O_{P}\left( \frac{1}{NT}\right)  \notag
\end{eqnarray}%
where $\mathcal{G}=\mathcal{F}\left( \lambda
_{i},f_{t}:i=1,...,N,t=1,...,T\right) $, and that%
\begin{equation}
\frac{1}{NT}\sum_{\left( i,t\right) \in \mathcal{I}_{k_{1},k_{2}}}(\hat{%
\Gamma}_{Y,it}^{\left( k_{1},k_{2}\right) }-\Gamma _{Y,it})^{2}=o_{P}\left(
1\right) ,\text{ \ \ }\frac{1}{NT}\sum_{\left( i,t\right) \in \mathcal{I}%
_{k_{1},k_{2}}}||\hat{\Gamma}_{X,it}^{\left( k_{1},k_{2}\right) }-\Gamma
_{X,it}||^{2}=o_{P}\left( 1\right) .  \label{eq: L_2 convergence}
\end{equation}%
Then, $\sqrt{NT}(\hat{\beta}_{NO}^{SS}-\beta _{0}-B_{\alpha ,1}/T-B_{\gamma
,1}/N)\rightarrow ^{d}\mathcal{N}\left( 0,\Omega _{X}^{-1}\Sigma \Omega
_{X}^{-1}\right) $.
\end{theorem}

\begin{corollary}
If, conditional on $\mathcal{G}$, $\left\{ Z_{it}\right\} $ are mutually
independent with $\sup_{it}\mathbb{E}\left[ \varepsilon _{it}^{2}|\mathcal{G}%
\right] <\infty $ and $\sup_{it}\mathbb{E}\left[ \eta _{it}^{2}|\mathcal{G}%
\right] <\infty $, then (\ref{eq: sample split cond}) holds.
\end{corollary}

This generalises Theorem 1 in \cite{beyhum2025inference} to allow for a very
broad class of first--step estimators $(\hat{\Gamma}_{X,it},\hat{\Gamma}%
_{Y,it})$ and to allow for possible time series and cross-sectional
dependence. In particular, we expect that (\ref{eq: sample split cond}) will
hold under suitable weak dependence conditions as explored in \cite%
{lunde2019sample}.

The theorem shows that sample splitting alone allows removes incidental
parameter biases from the two first-order terms in great generality.
However, the second--order term in (\ref{eq: rate cond 0}) may still
generate incidental parameter biases.

\subsubsection{Sample Splitting with Linear Smoothers}

\label{sect:samplesplitLinear}

We here develop an alternative sample splitting procedure that not only
removes incidental parameter biases arising from the first-order terms but
also from the second--order term when combined with the two--way estimator
where $\hat{g}_{0,Z}^{\left( 1\right) }(\hat{\lambda}_{i},f_{t})$, $\hat{g}%
_{0,Z}^{\left( 2\right) }(\lambda _{i},\hat{f}_{t})$, and $\hat{g}_{0,Z}(%
\hat{\lambda}_{i},\hat{f}_{t})$ are obtained using a linear nonparametric
regression procedure. One example is (\ref{eq: two-way kernel regression}),
but other options are also possible, such as a two--way series estimator.

Take a linear smoother $\mathcal{W}_{it}^{\left( k_{1},k_{2}\right) }$
computed from data in $\mathcal{I}\backslash \mathcal{I}_{k_{1},k_{2}}$ only,%
\begin{equation*}
\mathcal{W}_{it}^{\left( k_{1},k_{2}\right) }\in \mathcal{F}_{k_{1},k_{2}}:=%
\mathcal{F}\left\{ \left( Y_{js},X_{js}\right) :\left( j,s\right) \in 
\mathcal{I}\backslash \mathcal{I}_{k_{1},k_{2}}\right\} ,\text{ \ }\left(
i,t\right) \in \mathcal{I}_{k_{1},k_{2}},
\end{equation*}%
and let 
\begin{equation*}
Z^{\left( k_{1},k_{2}\right) }=\left\{ Z_{it}\right\} _{\left( i,t\right)
\in \mathcal{I}_{k_{1},k_{2}}}\in \mathcal{F}\left\{ Z_{js}:\left(
j,s\right) \in \mathcal{I}_{k_{1},k_{2}}\right\}
\end{equation*}%
denote data from $\mathcal{I}_{k_{1},k_{2}}$. We then define%
\begin{equation*}
\hat{\Gamma}_{Z,it}^{\left( k_{1},k_{2}\right) }=\mathcal{W}_{it}^{\left(
k_{1},k_{2}\right) }Z^{\left( k_{1},k_{2}\right) },\text{ }\left( i,t\right)
\in \mathcal{I}_{k_{1},k_{2}},\text{ \ }\hat{\Gamma}_{Z}=[\hat{\Gamma}%
_{Z}^{\left( 1,1\right) },\hat{\Gamma}_{Z}^{\left( 1,2\right) },\hat{\Gamma}%
_{Z}^{\left( 2,1\right) },\hat{\Gamma}_{Z}^{\left( 2,2\right) }],
\end{equation*}%
where $\hat{\Gamma}_{Z}$ simply concatenates the partitioned matrices of
estimates. This sample split ensures the smoother $\mathcal{W}_{it}^{\left(
k_{1},k_{2}\right) }$ has no stochastic dependence with $\eta _{it}$ and $%
\varepsilon _{it}$ for $\left( i,t\right) \in \mathcal{I}_{k_{1},k_{2}}$
under independence across $\left( i,t\right) $. Note that whilst the \cite%
{freeman2023linear} sample split in \eqref{eqn:sampleSplit} is sufficient
for this, with mutually independent $\eta _{it}$ and $\varepsilon _{it}$ a
leave-one-out split would also be sufficient but requires $N\times T$
estimations, i.e. for each $i$ and $t$, so we do not implement this version.

The sample splitting estimator still takes the form (\ref{eq: beta-hat
sample splitting}) and so we could again apply Theorem \ref{Th: NO general}
as we did in the previous subsection. However, when linear estimators are
employed in the first-stage, a more precise expansion can be obtained that
leads to sharper rate restrictions on the first--stage. For any two matrices 
$A,B\in \mathbb{R}^{N\times T}$, let $\langle A,B\rangle
_{F}:=\sum_{it}A_{it}B_{it}/\left( NT\right) $ denote the scaled entrywise
Frobenius inner product. With $\hat{\Gamma}=\mathcal{W}(\Gamma +\varepsilon
)\in \mathbb{R}^{N\times T}$, where $\Gamma _{it}:=g(\alpha _{i},\gamma
_{t}) $, we have $\hat{\Gamma}_{Y}=\hat{\Gamma}_{X}\beta +\hat{\Gamma}$ and
so for $\hat\Omega_X := (NT)^{-1} \sum_{it} (X_{it}-\hat{\Gamma}%
_{X,it})(X_{it}-\hat{\Gamma}_{X,it})^\prime$, 
\begin{equation*}
\hat\Omega_X^{-1}\langle X-\hat{\Gamma}_{X},Y-\hat{\Gamma}_{Y}\rangle
_{F}=\beta +\hat\Omega_X^{-1}\langle \Gamma _{X}-\hat{\Gamma}_{X},(\Gamma -%
\hat{\Gamma})+\varepsilon \rangle _{F}+\hat\Omega_X^{-1}\langle \eta
,(\Gamma -\hat{\Gamma})+\varepsilon \rangle _{F}.
\end{equation*}%
Hence, for the estimator to be asymptotically normally distributed without
any asymptotic biases, we need (\ref{eq: rate cond 3}) to hold together
with\ 
\begin{align}
\xi _{\Gamma X} &:=\langle \Gamma _{X}-\hat{\Gamma}_{X},\Gamma -\hat{\Gamma}%
\rangle _{F}=o_{p}(1/\sqrt{NT}),  \label{eqn:sampleSplitLinear2} \\
\langle \Gamma _{X}-\hat{\Gamma}_{X},\varepsilon \rangle _{F} &=o_{p}(1/%
\sqrt{NT}),\text{ \ \ }\langle \eta ,\Gamma -\hat{\Gamma}\rangle
_{F}=o_{p}(1/\sqrt{NT}).
\end{align}%
Importantly, compared to the general case, $\hat{\Gamma}_{Y}-\Gamma _{Y}$
has been replaced by $\hat{\Gamma}-\Gamma $. Sample splitting will now
ensure that $\langle \Gamma _{X}-\hat{\Gamma}_{X},\varepsilon \rangle
_{F}=o_{p}(1/\sqrt{NT})$ and $\langle \Gamma _{X}-\hat{\Gamma}%
_{X},\varepsilon \rangle _{F}=o_{p}(1/\sqrt{NT})$. What remains is to show (%
\ref{eqn:sampleSplitLinear2}). We conclude:

\begin{theorem}
\label{thm:SSWWresult} Let $\{\|\mathcal{W}\|\cdot \xi _{\Gamma },\|\mathcal{%
W}\|\cdot \xi _{X}\}=o_{p}(1)$, 
$\langle\mathcal{W}\mathbb{E}[\varepsilon|\mathcal{F}], \mathcal{W}\mathbb{E}%
[\eta|\mathcal{F}] \rangle _{F} = o_p(NT)^{-1/2}$. Suppose, $\Vert \mathbb{E}%
\left[ vec(\varepsilon )vec(\varepsilon )^{\prime }|\mathcal{F}\right] \Vert
_{2}=O(1)$, and $\Vert \mathbb{E}\left[ vec(\eta )vec(\eta )^{\prime }|%
\mathcal{F}\right] \Vert _{2}=O(1)$. Further, let, 
\begin{align}  \label{eqn:WeightedWithinCondtions}
\Vert \mathcal{W}^{\ast }\mathcal{W}\eta_k \Vert _{F}^2 = o_p(NT), & & 
\langle \Gamma _{X}-\mathcal{W}\Gamma _{X},\Gamma -\mathcal{W}\Gamma \rangle
_{F}=o_{p}(NT)^{-1/2}.& & 
\end{align}%
Then, under (\ref{eq: rate cond 3}), 
\begin{equation*}
\sqrt{NT}(\hat{\beta}_{NO}^{SS}-\beta )\xrightarrow[]{d}\mathcal{N}(0,\Omega
_{X}^{-1}\Sigma \Omega _{X}^{-1}).
\end{equation*}
\end{theorem}

We establish $\Vert \mathcal{W}^{\ast }\mathcal{W}\eta_k \Vert _{F}^{2}
=o_p(NT)$, and $\mathbb{E}[\langle\Gamma - \mathcal{W}\Gamma,\Gamma_X - 
\mathcal{W}\Gamma_X\rangle_F] = o_p(NT)^{-1/2}$ for our specific estimators
in Section~\ref{sect:Asymptotics}. Condition $\langle\mathcal{W}\mathbb{E}%
[\varepsilon|\mathcal{F}], \mathcal{W}\mathbb{E}[\eta|\mathcal{F}] \rangle
_{F} = o_p(NT)^{-1/2}$ holds in many weak dependence settings, see Remark~%
\ref{rem:weakDep}. It trivially holds for iid $\varepsilon$, since then $%
\mathbb{E}[\varepsilon|\mathcal{F}] = 0$.

Condition $\Vert \mathcal{W}^{\ast }\mathcal{W}\eta_k \Vert _{F}^{2}
=o_p(NT) $ is satisfied for many $\mathcal{W}$. Take a cross-sectional
smoother for simplicity, such that $\mathcal{W} = w\in\mathbb{R}^{N\times N}$%
. When weights are set to $1/N$, we get the usual $\mathcal{W}\eta_k =
N^{-1} \iota_N \cdot \iota_N^\prime \eta_k = \iota_N \iota_T^\prime \cdot
O_p(1/\sqrt{N})$. When nonparametric smoothers are used with bandwidth $h$, $%
\mathcal{W}\eta_k = \iota_N \iota_T^\prime \cdot O_p(1/\sqrt{Nh})$ under
common regularity conditions. When additive nonparametric smoothers are
used, with bandwidth $h$ and $R$ additive terms, this becomes $\mathcal{W}%
\eta_k = \iota_N \iota_T^\prime \cdot O_p(R/\sqrt{Nh})$. These rates get
slower with more flexible smoothers, but are easily all $o(1)$.\footnote{%
The condition for additive smoother in this example is $R^2/h = o(N)$. Since
we only consider $R \to \infty$ very slowly as a function of $\min\{N,T\}$,
this is not difficult to satisfy.}

\begin{remark}
\label{rem:weakDep} As an example of weak dependence for $\varepsilon_{it}$
and $\eta_{it}$, we show in Appendix~\ref{sect:AppProofs} that Theorem~\ref%
{thm:sampleSplit} and Theorem~\ref{thm:SSWWresult} holds under the following
time series models:%
\begin{equation*}
\varepsilon _{it}=\rho \varepsilon _{i,t-1}+e_{\varepsilon ,it},\text{ \ \ }%
\eta _{it}=\tilde{\rho}\eta _{it-1}+e_{\eta ,it},
\end{equation*}
\begin{equation*}
\varepsilon _{it}=\sum_{s=1}^{\infty }\theta _{s}e_{\varepsilon
,it-s}+e_{\varepsilon ,it},\text{ \ \ }\eta _{it}=\sum_{s=1}^{\infty }\tilde{%
\theta}_{s}e_{\eta ,it-s}+e_{\eta ,it},
\end{equation*}
where $\max \left\{ |\rho |,|\tilde{\rho}|\right\} <1$, $\int_{t}^{\infty
}|\theta _{s}|ds\lesssim t^{-a}$, $\int_{t}^{\infty }|\tilde{\theta}%
_{s}|ds\lesssim t^{-\tilde{a}}$ with $\min \{a,\tilde{a}\}>1/2,$ and $%
e_{\varepsilon ,it},e_{\eta ,it}$ are i.i.d. mean zero with $\mathbb{E}\left[
e_{\varepsilon ,it}^{4}\right] <\infty $ and $\mathbb{E}\left[ e_{\eta
,it}^{4}\right] <\infty $.
\end{remark}

\section{A novel estimator of $g_{Z}$ and fixed effects\label{sect:g_Z ID}}

In this section, we import a novel identification result for $g_{Z}$ and the
fixed effects from the companion paper \cite{FreemanKristensen2026} and use
this to develop a kernel regression-based estimators of these. We then
proceed to verify that this estimator satisfies the high--level conditions
in Theorem \ref{thm:SSWWresult}.

We first introduce a singular value decomposition (SVD) of $g$: Let, for a
given multi-index $\iota =\left( \iota _{1,1},\dots ,\iota _{1,d_{\alpha
}},\iota _{2,1},\dots ,\iota _{2,d_{\gamma }}\right) \in \mathbb{N}%
_{0}^{d_{\alpha }+d_{\gamma }}$, 
\begin{equation*}
g^{(\iota )}\left( \alpha ,\gamma \right) =\frac{\partial ^{|\iota |}g\left(
\alpha ,\gamma \right) }{\partial \alpha _{1}^{\iota _{1,1}},\dots ,\partial
\gamma _{d_{\gamma }}^{\iota _{2,d_{\gamma }}}}
\end{equation*}%
be the mixed partial derivative. Let the norm $\Vert g\Vert
_{L_{f}^{2}(\Omega _{\alpha }\times \Omega _{\gamma })}$ denote the usual $%
L_{2}$-norm over support of $\left( \alpha _{i},\gamma _{t}\right) $,
denoted $\Omega _{\alpha }\times \Omega _{\gamma }$ taken with respect to
joint distribution of $\left( \alpha _{i},\gamma _{t}\right) $, 
\begin{equation*}
\Vert g\Vert _{L_{f}^{2}(\Omega _{\alpha }\times \Omega _{\gamma })}=\mathbb{%
E}\left[ g^{2}(\alpha _{i},\gamma _{t})\right] =\int_{\Omega _{\alpha
}}\int_{\Omega _{\gamma }}g^{2}(a,c)\pi _{\alpha }(a)\pi _{\gamma }(c)dadc,
\end{equation*}%
where $\pi _{\alpha }$ and $\pi _{\gamma }$ denote the densities of $\alpha
_{i}$ and $\gamma _{t}$. We will then assume that:

\begin{assumption}
\label{ass:g_Z}$\left( \alpha _{i},\gamma _{t}\right) $ is finite
dimensional, $d_{\alpha }=\dim \left( \alpha _{i}\right) <\infty $ and $%
d_{\gamma }=\dim \left( \gamma _{t}\right) <\infty $, with time--invariant
joint density $\pi _{\alpha }(\alpha )\pi _{\gamma }(\gamma )$ and compact
support $\Omega _{\alpha }\times \Omega _{\gamma }\subseteq \mathbb{R}%
^{d_{\alpha }}\times \mathbb{R}^{d_{\gamma }}$; 
\begin{equation}
g\in H_{f}^{p}(\Omega _{\alpha }\times \Omega _{\gamma })=\{g\in
L_{f}^{2}(\Omega _{\alpha }\times \Omega _{\gamma }):g^{(\iota )}\in
L_{f}^{2}(\Omega _{\alpha }\times \Omega _{\gamma })\,\,\forall \,\,|\iota
|\leq p\}.  \label{eqn:HilbertCondition}
\end{equation}
\end{assumption}

Under this assumption, we obtain the following SVD of $g$,%
\begin{equation}
g(\alpha _{i},\gamma _{t})=\sum_{r=1}^{\infty }\sigma _{r}u_{r}(\alpha
_{i})v_{r}(\gamma _{t}),  \label{eqn: SV decomposition}
\end{equation}%
where $\sigma _{1}>\sigma _{2}>....$ are the ordered singular values, and $%
u_{r}\in H_{f}^{p}(\Omega _{\alpha }\times \Omega _{\gamma })$ and $v_{r}\in
H_{f}^{p}(\Omega _{\alpha }\times \Omega _{\gamma })$ are eigenfunctions.
These constitute an orthonormal basis of $L_{f}^{2}(\Omega _{\alpha })$ and $%
L_{f}^{2}(\Omega _{\gamma })$, respectively, so that%
\begin{equation*}
\int_{\Omega _{\alpha }}u_{r}(a)u_{r^{\prime }}(a)\pi _{\alpha
}(a)da=I\left\{ r=r^{\prime }\right\} \quad \int_{\Omega _{\gamma
}}v_{r}(c)v_{r^{\prime }}(c)\pi _{\gamma }(c)dc=I\left\{ r=r^{\prime
}\right\} 
\end{equation*}%
We refer to \cite{griebel2014approximation} and \cite{freeman2023linear} for
further details. Importantly, the eigenfunctions are identified, c.f. \cite%
{freeman2023linear}.

Next, we impose the following injectivity condition on $g$:

\begin{assumption}
\label{ass:g_Z inj}Suppose that, for some $p_{\alpha }\geq d_{\alpha }$ and $%
p_{\gamma }\geq d_{\gamma }$, there exits $\left( \alpha _{0,1},....,\alpha
_{0,p_{\gamma }}\right) \in \Omega _{\alpha }^{p_{\gamma }}$ and $\left(
\gamma _{0,1},....,\gamma _{0,p_{\alpha }}\right) \in \Omega _{\gamma
}^{p_{\alpha }}$so that%
\begin{align*}
\frac{\partial \left( g\left( \alpha ,\gamma _{0,1}\right) ,....,g\left(
\alpha ,\gamma _{0,p_{\alpha }}\right) \right) }{\partial \alpha }& \in 
\mathbb{R}^{d_{\alpha }\times p_{\alpha }}\text{ is rank $d_{\alpha }$ for
all }\alpha , \\
\frac{\partial \left( g\left( \alpha _{0,1},\gamma \right) ,....,g\left(
\alpha _{0,p_{\gamma }},\gamma \right) \right) }{\partial \gamma }& \in 
\mathbb{R}^{d_{\gamma }\times p_{\gamma }}\text{ is rank $d_{\gamma }$ for
all }\gamma .
\end{align*}
\end{assumption}

In discussing Assumption \ref{ass:g_Z inj}, it is important to note that
there exist many observational equivalent representations of $g\left( \alpha
_{i},\gamma _{t}\right) $. Suppose, for example, that $g\left( \alpha
_{i},\gamma _{t}\right) =\tilde{g}\left( A\alpha _{i},B\gamma _{t}\right) $,
where $A\in \mathbb{R}^{\tilde{d}_{\alpha }\times d_{\alpha }}$, $B\in 
\mathbb{R}^{\tilde{d}_{\gamma }\times d_{\gamma }}$ and $\tilde{g}_{Z}:%
\mathbb{R}^{\tilde{d}_{\alpha }}\times \mathbb{R}^{\tilde{d}_{\gamma }}$
with $\tilde{d}_{\alpha }<d_{\alpha }$ and $\tilde{d}_{\gamma }<d_{\gamma }$%
. In this case, a lower--dimensional observational equivalent representation
is $\tilde{g}\left( \tilde{\alpha}_{i},\tilde{\gamma}_{t}\right) $, where $%
\tilde{\alpha}_{i}=A\alpha _{i}$ and $\tilde{\gamma}_{t}=B\gamma _{t}$. More
generally, if, for any given tuple $(\alpha _{1},...,\alpha _{d_{\gamma }})$%
, the mapping $\gamma \mapsto \left( g\left( \alpha _{1},\gamma \right)
,....,g\left( \alpha _{d_{\gamma }},\gamma \right) \right) $ does not have
full rank or, for any given tuple $(\gamma _{1},...,\gamma _{d_{\alpha
}})\in R^{d_{\gamma }\times d_{\gamma }}$, the mapping \newline
$\alpha \mapsto \left( g\left( \alpha ,\gamma _{1}\right) ,....,g\left(
\alpha ,\gamma _{d_{\alpha }}\right) \right) $ does not have full rank, we
say that $g\left( \alpha _{i},\gamma _{t}\right) $ is reducible. In either
case the finite--dimensional distribution of $\left\{ Z_{it}\right\} _{1\leq
i\leq d_{\gamma },1\leq t\leq d_{\alpha }}$ can be represented by some
lower--dimensional mapping $\tilde{g}\left( \tilde{\alpha}_{i},\tilde{\gamma}%
_{t}\right) $, where $\tilde{d}_{\alpha }<d_{\alpha }$ or $\tilde{d}_{\gamma
}<d_{\gamma }$.

In the light of these observations, Assumption \ref{ass:g_Z inj} is quite
weak. If Assumption \ref{ass:g_Z inj} is violated, and so $g\left( \alpha
_{i},\gamma _{t}\right) $ is reducible, then we can work with the
lower--dimensional observational equivalent representation of $g$ that will
satisfy above assumption. As such, we find that the assumptions \ref{ass:g_Z
inj} imposes very weak restrictions on $g\left( \alpha _{i},\gamma
_{t}\right) $.

The following theorem shows that under Assumption \ref{ass:g_Z inj} a linear
combination of a finite number of the leading eigenfunctions are valid
proxies. Moreover, the associated regression function inherits the
smoothness properties of its mother.

\begin{theorem}
\label{thm:eigen inver}Suppose Assumptions \ref{ass:g_Z} and \ref{ass:g_Z
inj} hold. Then there exists $R_{0}\geq \max \left\{ p_{\alpha },p_{\gamma
}\right\} $, $A\in \mathbb{R}^{p_{\alpha }\times R_{0}}$, $B\in \mathbb{R}%
^{p_{\gamma }\times R_{0}}$ with $\{p_\alpha,p_\gamma\} \geq
\{d_\alpha,d_\gamma\}$ so that, with $U\left( \alpha \right) :=\left(
u_{1}\left( \alpha _{i}\right) ,...,u_{R_{0}}\left( \alpha _{i}\right)
\right) ^{\prime }$ and $V\left( \gamma \right) :=\left( v_{1}\left( \gamma
_{t}\right) ,...,v_{R_{0}}\left( \gamma _{t}\right) \right) ^{\prime }$,%
\begin{equation}
\alpha \mapsto AU\left( \alpha \right) ,\text{ \ \ }\gamma \mapsto BV\left(
\gamma \right) \text{ are injective.}  \label{eq: one-to-one cond}
\end{equation}%
Moreover, $U\left( \alpha _{i}\right) $ and $V\left( \gamma _{t}\right) $
are identified so that%
\begin{equation}
g_{0,Z}\left( \lambda _{i},f_{t}\right) :=E\left[ Z_{it}|\lambda _{i},f_{t}%
\right] ,\text{ \ \ }\lambda _{i}=AU\left( \alpha _{i}\right) ,\text{ \ \ }%
f_{t}=BV\left( \gamma _{t}\right)  \label{eq: g_0 index}
\end{equation}%
is identified and satisfies $g_{0,Z}\left( \lambda _{i},f_{t}\right)
=g_{Z}\left( \alpha _{i},\gamma _{t}\right) $. The function $g_{0,Z}\left(
\lambda _{i},f_{t}\right) $ has the same degree of smoothness as $g_{Z}$.
\end{theorem}

Estimation of $g_{0,Z}$ is a high--dimensional nonparametric regression
problem if $p_{\alpha }$ and/or $p_{\gamma }$ are large, and so may suffer
from the well-known curse--of--dimensionality which leads to large errors in
the nonparametric estimation. We therefore now introduce restrictions under
which this curse is less of a concern. To simplify notation, we here assume
that $d_{\alpha }=d_{\gamma }=d$.

\begin{theorem}
\label{thm:additive model}Suppose that $g_{Z}\left( \alpha _{i},\gamma
_{t}\right) $ is additive,%
\begin{equation*}
g_{Z}\left( \alpha _{i},\gamma _{t}\right) =\sum_{k=1}^{d}h_{k}\left( \alpha
_{i,k},\gamma _{t,k}\right) ,
\end{equation*}%
and Assumptions \ref{ass:g_Z} and \ref{ass:g_Z inj} hold. Then $%
g_{0,Z}\left( \lambda _{i},f_{t}\right) $ defined in Theorem \ref{thm:eigen
inver} is also additive, 
\begin{equation}
g_{0,Z}\left( \lambda _{i},f_{t}\right) =\sum_{k=1}^{d}h_{k}\left( \lambda
_{i,k},f_{t,k}\right) ,\text{ \ \ }\lambda _{i,k}=a_{k}^{\prime }U\left(
\alpha _{i}\right) ,\text{ \ \ }f_{t,k}=b_{k}^{\prime }V\left( \gamma
_{t}\right) ,  \label{eq: additive model}
\end{equation}%
where $A=\left[ a_{1}^{\prime },...,a_{d}^{\prime }\right] ^{\prime }$ and $%
B=\left[ b_{1}^{\prime },...,b_{d}^{\prime }\right] ^{\prime }$ were defined
in Theorem \ref{thm:eigen inver}.
\end{theorem}

Next, we develop two--step regression estimators of $g_{Z}$ based on above
identification result: In the first step, we estimate the leading
eigenfunctions of $g$ that in the second step are used as proxies for $%
\alpha _{i}$ and $\gamma _{t}$ in a nonparametric regression procedure.

\subsection{First--step estimation of eigenfunction proxies \label{sect:FM}}

This section presents our first-step estimators of the leading
eigenfunctions of the SVD\ representation of $g$ in (\ref{eqn: SV
decomposition}). Substituting (\ref{eqn: SV decomposition}) into (\ref%
{eqn:model}) and truncating the singular value decomposition at some $%
R_{1}\geq 1$ chosen by the econometrician yields%
\begin{equation*}
Y_{it}=\beta ^{\prime }X_{it}+\sum_{r=1}^{R_{1}}\lambda
_{ir}f_{tr}+e_{R_{1},it}+\varepsilon _{it},
\end{equation*}%
where $\lambda _{ir}=\sigma _{r}u_{r}(\alpha _{i})$, $f_{tr}=v_{r}(\gamma
_{t})$, and $e_{R_{1},it}=\sum_{r=R_{1}+1}^{\infty }\lambda _{ir}f_{tr}$ is
the truncation error. We follow \cite{freeman2023linear} and obtain
first-step estimates of $\lambda \in \mathbb{R}^{N\times R_{1}}$ and $f\in 
\mathbb{R}^{T\times R_{2}}$ by applying the estimator of \cite{Bai2009} to
above approximate factor model, 
\begin{equation}
(\hat{\beta}_{LS},\hat{\lambda}_{1:R_{1}},\hat{f}_{1:R_{1}})=\arg \min 
_{\substack{ \lambda _{1:R_{1}}\in \mathbb{R}^{N\times R_{1}},  \\ %
f_{1:R_{1}}\in \mathbb{R}^{T\times R_{1}}}}\min_{\beta \in \mathbb{R}%
^{K}}\sum_{i,t}\left( Y_{it}-\beta ^{\prime
}X_{it}-\sum_{r=1}^{R_{1}}\lambda _{ir}f_{tr}\right) ^{2}
\label{eq: beta-hat-LS}
\end{equation}%
where we impose the normalisations from \cite{Bai2009,bai2023approximate}, $%
N^{-1}\sum_{i}\lambda _{1:R_{1},i}\lambda _{1:R_{1},i}^{\prime }$ is
diagonal and $T^{-1}\sum_{t}f_{1:R_{1},t}f_{1:R_{1},t}^{\prime }=\mathbb{I}%
_{R_{1}}$.

Above algorithm delivers estimators of the leading $R_{1}$ eigenfunctions, $%
\hat{\lambda}_{1:R_{1}}$ and $\hat{f}_{1:R_{1}}$. Importantly, compared to
the alternative estimation procedure of \cite{beyhum2025inference}, the
algorithm effectively reduces the dimension of the fixed effects to be
controlled for since it takes into account the presence of $\beta ^{\prime
}X_{it}$ in the model. If the DGP for $X_{it}$ takes the form (\ref{eq: X
DGP Beyhum}), then the algorithm of \cite{beyhum2025inference} will
generally take into account not only the fixed effects $(\alpha _{i},\gamma
_{t})$ that enter the model for $Y_{it}$ but also the fixed effects $(\alpha
_{i}^{\left( 2\right) },\gamma _{t}^{\left( 2\right) })$ that are specific
to $X_{it}$. In contrast, above algorithm does not suffer from such
shortcomings since it controls for $\beta ^{\prime }X_{it}$ and so directly
targets $(\alpha _{i},\gamma _{t})$.

The algorithm also delivers estimators $\hat{\beta}_{LS}$ and $\hat{\Gamma}%
_{it}=\sum_{r=1}^{R_{1}}\hat{\lambda}_{ir}\hat{f}_{tr}$ of $\beta $ and $%
\Gamma _{it}$, respectively. However, these estimators suffer from large
errors due to the truncation error $e_{R,it}$ and so $\hat{\beta}_{LS}$ will
not enjoy $\sqrt{NT}$-asymptotic normality: The factor model approach to
Neyman Orthogonal estimator uses $\Gamma -\hat{\Gamma}=(\mathbb{I}-P_{\hat{%
\lambda}})g(\alpha ,\gamma )(\mathbb{I}-P_{\hat{f}})$, which leads to, 
\begin{equation*}
\frac{1}{\sqrt{NT}}\Vert \Gamma -\hat{\Gamma}\Vert \approx
O_{p}(R_{1}^{1-\rho })+O_{p}(R_{1}^{2+2\rho }\min \{N,T\}^{-1}).
\end{equation*}%
c.f. \cite{bai2023approximate} and Section \ref{sect:FirstAsymptotics}. The
first term $O_{p}(R_{1}^{1-\rho })$ is due to the truncation error and
shrinks as $R_{1}$ grows, and is also decreasing in smoothness $\rho $, i.e.
smoother functions lead to smaller bias. Variance term $O_{p}(R_{1}^{2+2\rho
}\min \{N,T\}^{-1})$, however, is increasing in model complexity $R_{1}$,
and also increasing in smoothness $\rho $. Setting $R_{1}=c\cdot \min
\{N,T\}^{\frac{1}{1+3\rho }}$ to balance bias and variance leads to $\frac{1%
}{\sqrt{NT}}\Vert \Gamma -\hat{\Gamma}\Vert =\min \{N,T\}^{\frac{1-\rho }{%
1+3\rho }}$. Hence, we can at best obtain $||\hat{\beta}_{LS}-\beta
||=O_{P}\left( \min \{N,T\}^{\frac{2-2\rho }{1+3\rho }}\right) $ which is
too slow for $\sqrt{NT}$-inference.

Figure \ref{fig:signalNoise} shows the two limiting components to the
first-step estimator. Either singular values decay too slowly, and this
leaves a large bias from $e_{it}$, or singular values decay too fast and
variation from $\sum_{r=1}^{R_{1}}\lambda _{ir}f_{tr}$ is indiscernible from 
$\varepsilon _{it}$, leading to higher variance. Figure \ref{fig:signalNoise}
shows a signal to noise comparison for the distribution of singular values
generated from a function and from Gaussian noise. When signal drops below
noise, factors from the function are no longer estimable, or estimated with
noise. We see in the left panel that for non-smooth functions with slowly
decaying singular values, we can potentially estimate and control for many
factors, however, there is a large error that persists in the tail of the
approximation. This leads to large bias. In the right panel, whilst the
approximation error in the tail is small, variation from the function
quickly becomes indiscernible from noise, hence estimates are noisy. As as a
consequence, in either scenario, the over--all estimation error of $\hat{g}$
is too large and do not vanish at the rate required in Theorem \ref{Th: NO
general}.

\begin{figure}[tbp]
\centering
\includegraphics[width=0.45\linewidth]{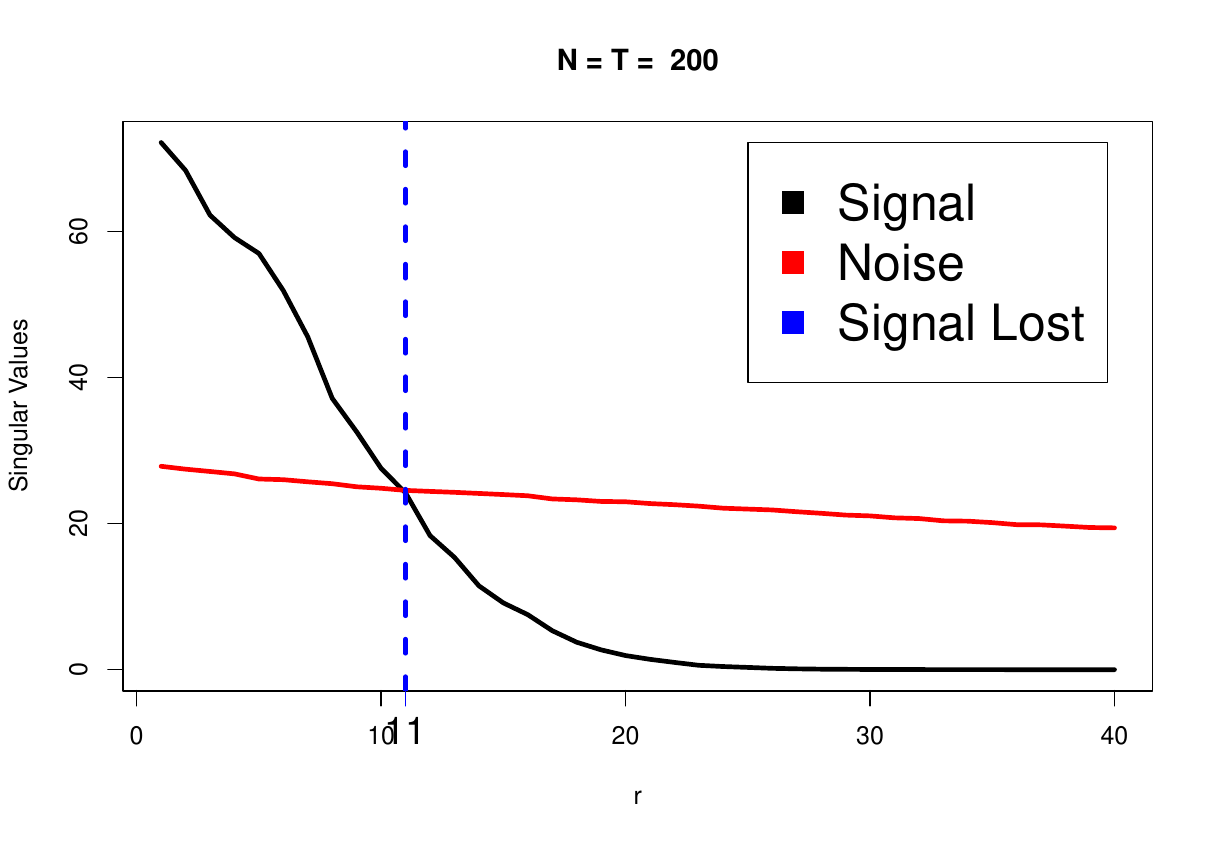} %
\includegraphics[width=0.45\linewidth]{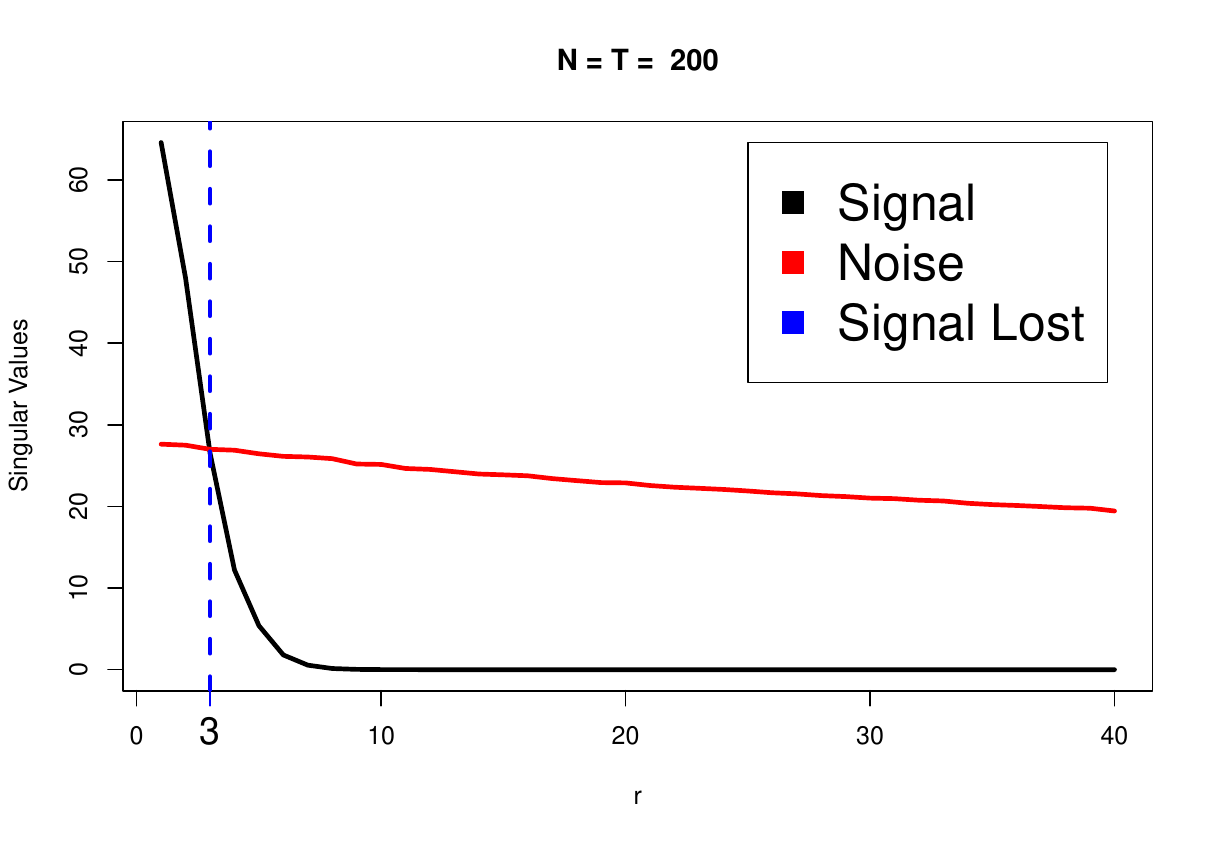}
\caption{Singular value decay: Function Signal vs. Gaussian Noise \newline
{\protect\footnotesize Primitive function: $g(a,b)=(\protect\theta \protect%
\sqrt{2\protect\pi })^{-1}\exp (-(a-b)^{2}/\protect\theta ^{2})$ with
different $\protect\theta $: higher $\protect\theta $ implies less smooth
function. $a,b\sim U(-1,1)$.} }
\label{fig:signalNoise}
\end{figure}

\subsection{Second--Step Estimation of $g_{0,Z}$}

We here develop two nonparametric regression algorithms that both take as
input the subset of the first $R_{2}$ of the $\ R_{1}$ estimated leading
eigenfunctions in the first step, where again $R_{2}$ is chosen by the
econometrician. With some abuse of notation, we let $\hat{\lambda}_{i}=(\hat{%
\lambda}_{1,i},.....,\hat{\lambda}_{R_{2},i})^{\prime }$ and $\hat{f}_{t}=(%
\hat{f}_{1,t},.....,\hat{f}_{R_{2},t})^{\prime }$ denote these final $R_{2}$
estimated eigenfunctions. We will require $R_{2}\geq R_{0}$ so that we can
apply Theorems \ref{thm:eigen inver} and \ref{thm:additive model} and obtain
consistent estimators based on the representation results in equations (\ref%
{eq: g_0 index}) and (\ref{eq: additive model}), respectively.

\subsubsection{Multi--index eigenfunction regression\label{sect:MI}}

The following kernel regression estimator is a consistent estimator of $%
g_{0,Z,k}\left( \lambda ,f\right) $ defined in \eqref{eq: g_Z kernel reg}, 
\begin{equation*}
\hat{g}_{0,Z,k}\left( \lambda ,f;\hat{A}_{k},\hat{B}_{k}\right) =\frac{%
\sum_{i=1}^{n}\sum_{t=1}^{T}Z_{k,it}K_{\hat{A}_{k},h_{1}}(\hat{\lambda}%
_{i}-\lambda )K_{\hat{B}_{k},h_{2}}(\hat{f}_{t}-f)}{\sum_{i=1}^{n}%
\sum_{t=1}^{T}K_{\hat{A}_{k},h_{1}}(\hat{\lambda}_{i}-\lambda )K_{\hat{B}%
_{k},h_{2}}(\hat{f}_{t}-f)},
\end{equation*}%
where $K_{\hat{A}_{k},h_{1}}(\hat{\lambda}_{i}-\lambda ):=K_{1}(\hat{A}_{k}\{%
\hat{\lambda}_{i}-\hat{\lambda}_{i_{0}}\}/h_{1})/h_{1}^{d_{\lambda }}$, $K_{%
\hat{B}_{k},h_{2}}(\hat{f}_{t}-f)=K_{2}(\hat{B}_{k}\{\hat{f}%
_{t}-f\}/h_{2})/h_{2}^{d_{f}}$, $h_{1}$ and $h_{2}$ are bandwidths, $K_{1}$
and $K_{2}$ are kernels, and%
\begin{equation*}
(\hat{A}_{k},\hat{B}_{k})=\arg \min_{A,B}\sum_{i=1}^{n}\sum_{t=1}^{T}\left(
Z_{k,it}-\hat{g}_{Z,k}(\hat{\lambda}_{i},\hat{f}_{t};A,B)\right) ^{2}.
\end{equation*}

Above estimator is a so--called multi--index regression estimator with
generated regressors $\hat{\lambda}_{i}$ and $\hat{f}_{t}$. When the
regressors are observed without errors, this has estimator has been analyzed
in, among others, \cite{Ma2012} and \cite{Ma2013}. This estimator will
suffer from a curse-of--dimensionality of order $\max \left\{ p_{\alpha
},p_{\gamma }\right\} $.

\subsubsection{Additive eigenfunction regression\label{sect:WW}}

Theorem \ref{thm:additive model} allows us to use additive nonparametric
estimation techniques to remove the curse--of--dimensionality that above
multi--index estimator suffers from. We here propose to employ the
kernel--based backfitting algorithm for nonparametric additive models to
remove this curse of dimensionality. We refer to \cite{opsomer1999root} for
details on this algorithm in a cross--sectional setting. We here extend it
to a panel data setting.%

Figure~\ref{fig:svDecay} shows a comparison of the linear projection method
from a factor model versus a nonparametric difference, where true $%
g(\alpha_i,\gamma_t)$ is observed. We see that whilst error for the first
three directly estimated terms is zero for linear projection methods, and
positive for nonparametric methods, there is a huge benefit in the tail of
residual singular values. By definition, factor components are exactly
orthogonal to tail eigenfunctions, however the nonparametric estimator can
still difference these out as a result of the additive model. In contrast,
nonparametric methods allow for infinite linear dimension, albeit with some
additional smoothness conditions.

\begin{figure}[tbp]
\centering
\includegraphics[width=0.7\linewidth]{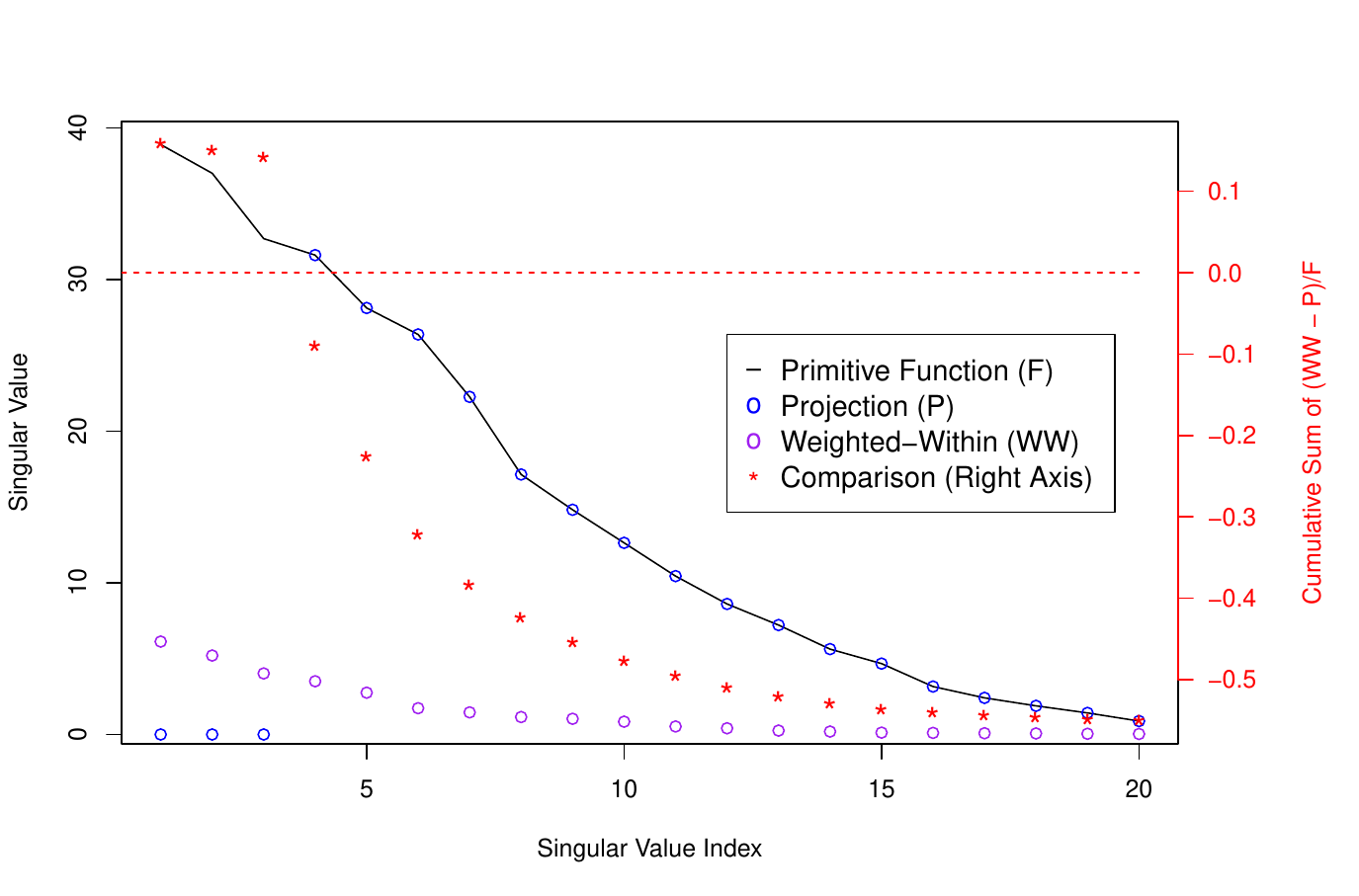}
\caption{Singular value decay: Function vs. Projection vs. Weighted-within 
\newline
{\protect\footnotesize Primitive function: $g(a,b)=(\protect\theta \protect%
\sqrt{2\protect\pi })^{-1}\exp (-(a-b)^{2}/\protect\theta ^{2})$ with $%
\protect\theta =1/8$. $a,b\sim U(-1,1)$. Comparison measures the cumulative
comparative performance of WW with respect to P. Cumulative sum on right
axis measures cumulative sum of the difference of singular values divided by
cumulative sum of the primitive function singular values. } }
\label{fig:svDecay}
\end{figure}

Here we combine the first-step estimators with a backfitting algorithm that
iterates over weighted-within transformations from \cite%
{freeman2022multidimensional} to obtain our final estimates. The
weighted-within transformation performs dimension specific differencing of
the fixed-effects. Take the set of estimates $\{\hat{\lambda}_{ir},\hat{f}%
_{tr}\}_{r=1}^{R_{1}}$. Smoother weights are formed for the $i$,
respectively $t$ direction as, 
\begin{equation*}
S_{ij,r}^{(1)}=\frac{K_{1,h_{1}}(\hat{\lambda}_{ir}-\hat{\lambda}_{jr})}{%
\sum_{j}K_{1,h_{1}}(\hat{\lambda}_{ir}-\hat{\lambda}_{jr})},\text{ \ \ }%
S_{ts,r}^{(2)}=\frac{K_{2,h_{2}}(\hat{f}_{tr}-\hat{f}_{sr})}{%
\sum_{s}K_{2,h_{2}}(\hat{f}_{tr}-\hat{f}_{sr})}.
\end{equation*}%
Dependent and independent variables are sequentially residualised with
respect to weighted-differences according to the sequence of weights for $%
r=1,\dots ,R_{2}$. For our asymptotic theory, we require a backfitting
update. With $\check{Y}^{(0)}=Y-\bar{Y}$, step $\ell $ in the backfitting
iteration can be written, 
\begin{equation*}
\check{Y}^{(\ell )}=\prod_{r=1}^{R_{2}}\big(\mathbb{I}_{N}-S_{r}^{(1)}\big)%
\check{Y}^{(\ell -1)}\prod_{r^{\prime }=1}^{R_{2}}\big(\mathbb{I}%
_{T}-S_{r^{\prime }}^{(2)}\big).
\end{equation*}

The algorithm works as follows. Take a generic smoothing function $s:\mathbb{%
R}^{n}\times \mathbb{R}^{n}\rightarrow \mathbb{R}^{n\times n}$, and $\hat{%
\lambda}_{ir}$ and $\hat{f}_{tr}$, $r=1,...,R_2$, We use backfitting
iteration from Algorithm~\ref{alg:backfitting}:

\begin{algorithm}
\caption{Additive Eigenfunction Backfitting}\label{alg:backfitting}
\begin{algorithmic}[1]
  \STATE Initialise $\tilde Z_0 = \bar Z$, residual $\check Z_0 = Z - \bar Z$, $\hat{%
\lambda}_{ir}$ and $\hat{f}_{tr}$. Smoothing function $s:\mathbb{%
R}^n\times\mathbb{R}^n\to \mathbb{R}^{n\times n}$, 
  \STATE For $r = 1,\dots, R_2$,
\begin{enumerate}[(i).]
\item Set $S_r^{(1)} = s(\hat\lambda_{r},\hat\lambda_{r})$, $S_r^{(2)} = s(\hat f_{r},\hat f_{r})$
\item Set $\tilde Z_r = S_r^{(1)}\check Z_{r-1} + \check Z_{r-1}S_r^{(2)} -
S_r^{(1)}\check Z_{r-1}S_r^{(2)} $ and residual $\check Z_{r} = \check
Z_{r-1} - \tilde Z_r$
\end{enumerate}
\STATE Do pooled OLS of $\check{Y}$ on $\check X$ to obtain $\hat\beta$. Take 
$\hat{\lambda}_{i},\hat{f}_{t}$ as leading $R_2$ eigenvectors of $Y - X\cdot\hat\beta$ 
\STATE Iterate steps 2,3 to convergence. 
\end{algorithmic}
\end{algorithm}

If convergence occurs at $L$ iterations, the procedure can be written, 
\begin{equation*}
\check{Y}=\prod_{\ell }^{L}\left[ \prod_{r=1}^{R_{2}}\big(\mathbb{I}%
_{N}-S_{r}^{(1)}\big)\right] \cdot (Y-\bar{Y})\cdot \prod_{\ell }^{L}\left[
\prod_{r^{\prime }=1}^{R_{2}}\big(\mathbb{I}_{T}-S_{r^{\prime }}^{(2)}\big)%
\right] .
\end{equation*}%
We define the estimator under these weighted differences as $\hat{\beta}_{W}$%
: 
\begin{equation*}
\hat{\beta}_{NO}=\left( \sum_{it}\check{X}_{it}\check{X}_{it}^{\prime
}\right) ^{-1}\sum_{it}\check{X}_{it}\check{Y}_{it}
\end{equation*}

\section{Asymptotic Theory\label{sect:Asymptotics}}

In this section we formally analyse asymptotic rates for the two--step
estimators of $g_{Z}$.

\subsection{First--step estimators of eigen proxies}

\label{sect:FirstAsymptotics} \cite{freeman2023linear} derive convergence
rates of the least squares estimator in Section~\ref{sect:FM}. These,
however, can be improved by using refined estimates from \cite%
{griebel2019singular}, which gives faster decay in the singular values for
the tail term, $\sum_{r=R_{1}+1}^{\infty }\lambda _{ir}f_{tr}$. Let $\rho
:=p/\min \{d_{\alpha },d_{\gamma }\}$, where $p$ is introduced in %
\eqref{eqn:HilbertCondition}, and $\{d_{\alpha },d_{\gamma }\}$ are
dimensions of $\alpha $, and $\gamma $. This comprises Lemma~\ref%
{lemma:frewei2023}.

\begin{assumption}
\label{ass:singularvalues} For $\rho >3/2$, as $N,T\rightarrow \infty $, $%
\sigma _{r}^{2}(\Gamma )$ the singular values of $\Gamma $ descending in $r$%
, then, 
\begin{equation*}
\sigma _{r}^{2}(\Gamma )=cNTr^{-2\rho -1}\text{ as }r\rightarrow \infty
,\quad \text{ such that, }\quad \frac{1}{NT}\sum_{r=R_{1}+1}^{\min
\{N,T\}}\sigma _{r}^{2}(\Gamma )=O_{p}(R_{1}^{-2\rho })\text{ as }%
R_{1}\rightarrow \infty .
\end{equation*}
\end{assumption}

Assumption~\ref{ass:singularvalues} refines Lemma 1 in \cite%
{freeman2023linear} by applying Corollary 3.4 and Proposition 3.5 from \cite%
{griebel2019singular}. Here we state regularity conditions from \cite%
{freeman2023linear,MoonWeidner2015} on covariates and $\varepsilon$.

\begin{assumption}[Bounded norms covariates and $\protect\varepsilon$]
\label{ass:boundedNorms} For $k = 1,\dots,dim(X_{it})$, and $e: =
vec(\varepsilon)$, 
\begin{align*}
\frac{1}{NT}\sum_{it} X_{it,k}^2 = O_p(1), & & \|\varepsilon\|_2 = O_p(
\max\{N,T\}^{1/2}), & & \| \mathbb{E}[ee^{\prime }|X]\|_2 = O_p(1) .
\end{align*}
\end{assumption}

Condition $\| \mathbb{E}[ee^{\prime }|X]\|_2 = O_p(1)$ bounds cross and
serial correlation. For example, i.i.d. $\varepsilon_{it}$ implies $\mathbb{E%
}[ee^{\prime }|X] = \mathbb{E}\varepsilon_{it}^2 \mathbb{I}_{NT}$ such that $%
\| \mathbb{E}[ee^{\prime }|X]\|_2 = \mathbb{E}\varepsilon_{it}^2 = O(1)$.
With weak correlation in $i$ and $t$ such that $\sum_{j\neq i}\sum_{s\neq t}|%
\mathbb{E}[\varepsilon_{it}\varepsilon_{js}|X]| = O(1)$, then Ostrowski
theorem implies $\| \mathbb{E}[ee^{\prime }|X]\|_2 = O_p(1)$.

\begin{assumption}[Weak exogeneity]
\label{ass:weakExogeneity} 
\begin{equation*}
\sum_{it}X_{it,k}\varepsilon _{it}=O_{p}(\sqrt{NT})\,\,\text{ for }%
\,\,k=1,\dots ,\dim (X_{it}).
\end{equation*}
\end{assumption}

\begin{assumption}[Non-collinearity in $X$]
\label{ass:nonCollinearity} For linear combinations $\delta \cdot
X:=\sum_{k}\delta _{k}X_{k}$ such that $\Vert \delta \Vert =1$, assume
exists $b>0$: 
\begin{equation*}
\min_{\delta \in \mathbb{R}^{K}:\Vert \delta \Vert =1}\sum_{r\geq
2R_{1}+1}\sigma _{r}\left( \frac{\delta \cdot X}{NT}\right) \geq
b\,\,\,\,\,\,wpa1.
\end{equation*}
\end{assumption}

We first analyse the first-step estimators in (\ref{eq: beta-hat-LS}):

\begin{lemma}
\label{lemma:frewei2023} 
Impose Assumption \ref{ass:singularvalues}, \ref{ass:boundedNorms}, \ref%
{ass:weakExogeneity}, and \ref{ass:nonCollinearity}. 
For ${R_1}=\min \{N,T\}^{1/2\rho }$, the estimator in (\ref{eq: beta-hat-LS}%
) satisfies%
\begin{equation*}
\widehat{\beta }_{LS}-\beta _{0}=O_{p}(R_1^{1-\rho })+O_{p}(R_1\min
\{N,T\}^{-1/2})=O_{p}\left( \min \{N,T\}^{\frac{1-\rho }{2\rho }}\right) .
\end{equation*}
\end{lemma}

For $N\sim T$, this is at best $O_{p}(NT)^{-1/4}$ when $\rho \rightarrow
\infty $; this is too slow for standard inference tools to be valid.

In Appendix~\ref{sect:AppProofs} we state regularity conditions on $\{\eta
_{it},\varepsilon _{it}\}$ and their correlation with $\{\lambda
_{i},f_{t}\} $ to apply results from \cite{Bai2009} and \cite%
{bai2023approximate} for factor model estimates. Results in \cite%
{choi2025high} may also apply to this setting. Define, 
\begin{equation}  \label{eqn:factorRates}
\xi _{f}^{2}:=\frac{1}{T}\sum_{t=1}^{T}\Vert \hat{f}_{t}-H^{\prime
}f_{t}\Vert ^{2},\text{ \ \ }\xi _{\lambda }^{2}:=\frac{1}{N}%
\sum_{i=1}^{N}\Vert \hat{\lambda}_{i}-(H^{\prime })^{-1}{\lambda }_{i}\Vert
^{2},
\end{equation}%
where $H$ are rotation matrices defined in the appendix, see \cite%
{bai2023approximate}.\footnote{%
These are not important to our analysis so we refer to discussion found in
existing literature.}

\begin{lemma}
\label{lemma:bai2009fe} Impose Assumption \ref{ass:singularvalues}, \ref%
{ass:boundedNorms}, \ref{ass:weakExogeneity}, \ref{ass:nonCollinearity}, \ref%
{ass:AppNorm}, \ref{ass:AppError}, and \ref{ass:AppCorrelation}. When $%
R_2\lesssim \min \{N,T\}^{\frac{1}{2\rho +1}}$ with $\rho >2$, and a
preliminary estimator $\widehat{\beta }$, 
\begin{equation*}
\xi _{f}^{2}=\xi _{\lambda }^{2}=O_{p}\left( \Vert \widehat{\beta }-\beta
^{0}\Vert ^{2}+R_2^{2\rho +1}\min \{N,T\}^{-1}+R_2^{1-4\rho }\right) ,
\end{equation*}%
Further, for $N\sim T$, set $R_2=\min \{N,T\}^{1/6\rho }$, such that, 
\begin{align*}
& {R_2^{2}}\cdot \xi _{f}^{2}\xi _{\lambda }^{2}=O_{p}\left( \min \{N,T\}^{%
\frac{1}{3\rho }}\Vert \widehat{\beta }-\beta ^{0}\Vert ^{4}+\min \{N,T\}^{%
\frac{2-4\rho }{3\rho }}\right) . \\
& \frac{R_2^{6}}{NT}\cdot \xi _{f}^{2}\xi _{\lambda }^{2}=O_{p}\left( \min
\{N,T\}^{\frac{1-2\rho }{\rho }}\Vert \widehat{\beta }-\beta ^{0}\Vert
^{4}+\min \{N,T\}^{\frac{4-10\rho }{3\rho }}\right) .
\end{align*}
\end{lemma}

In Corollary~\ref{cor:bai2009fe} we show $R_2^{2}\cdot \xi _{f}^{2}\xi
_{\lambda }^{2}=o_{p}(\min \{N,T\}^{-1})$ and $\frac{R_2^{6}}{NT}\cdot \xi
_{f}^{2}\xi _{\lambda }^{2}=o_{p}(\min \{N,T\}^{-2})$, which is sufficient
for our result stated in Corollary~\ref{cor:Final} below. Since $\hat{\beta}$
is arbitrary in the statement of Lemma~\ref{lemma:bai2009fe}, we can set a
different number of factors for $\beta $ estimate and factor component
estimates.

\begin{corollary}
\label{cor:bai2009fe} Impose Lemma~\ref{lemma:frewei2023} assumptions and
set $R_1=c\cdot \min \{N,T\}^{1/2\rho }$ for $\hat{\beta}_{LS}$. For $\rho
>7/3 $: 
\begin{equation*}
\min \{N,T\}^{\frac{1}{3\rho }}\Vert \widehat{\beta }_{LS}-\beta _{0}\Vert
^{4}=o_{p}(\min \{N,T\}^{-1}), \min \{N,T\}^{\frac{1-2\rho }{\rho }}\Vert 
\widehat{\beta }_{LS}-\beta _{0}\Vert ^{4}=o_{p}(\min \{N,T\}^{-2}).
\end{equation*}
\end{corollary}

It may be possible to attain similar, or potentially faster rates, by
iterating between factor estimation and our final $\hat\beta_{NO}$
estimator. However, the above rate is sufficient, and easy to verify, so we
do not confirm this.

\subsection{Second--step estimators of $g_{0,Z}$}

\label{sect:EstSecondStep}

We here only analyze the additive version of the two proposed kernel
regression estimators of $g_{0,Z}$. Since both the additive and the
multi-index version are kernel regressions, our analysis of the additive
version will carry over the multi-index estimator with obvious
modifications. To handle the estimation of the multi-index parameters, we
can apply the techniques developed in \cite{Ma2012} and \cite{Ma2013}.

We add some regularity to the kernel functions and distribution of
eigenfunctions.

\begin{assumption}[Kernels]
\label{ass:Kernels} Kernel function, $k(\cdot)$, 
\begin{enumerate*}[%
(i).,series = tobecont, itemjoin = \quad]

\item Bounded with compact support, \newline

\item $\int u^jk(u)du = 0$ for odd $j$

\item $\int k(u)du = 1$

\item $0<\int u^2 k(u)du <\infty$. 
\end{enumerate*}
\end{assumption}

Assumption~\ref{ass:Kernels} are standard restrictions. 
We regularise the distribution of proxies.

\begin{assumption}[Eigenfunctions]
\label{ass:Proxies} Marginals $f_{u_r}({u_{ir}})$, and $f_{v_r}({v_{ir}})$
are bounded with compact support and admit for all $r=1\dots, R$: 
$f_{u_r}({u_{ir}})>0\,\,\forall {u}_{ir}$, $f_{v_r}({v_{ir}}) >0 \,\,\forall 
{v}_{ir}$. 
\end{assumption}

Assumptions~\ref{ass:Kernels},~\ref{ass:Proxies} are standard nonparametric
estimation assumptions, see Assumptions 1 and $2^{\prime }$-$4^{\prime }$
from \cite{opsomer1999root}. Assumptions~\ref{ass:Proxies} can be justified
with estimated eigenfunctions of functions that adhere to Assumption~\ref%
{ass:singularvalues}, which pointwise converge to the true eigenfunctions,
which are in a compact set.

We state convergence results for $\widehat\beta_W^{SS}$, the estimator from
Section~\ref{sect:WW} using the sample splitting described in Section~\ref%
{sect:samplesplitLinear} to estimate smoothing weight matrices $S^{(1)}$ and 
$S^{(2)}$. Lemma~\ref{lemma:withinEstimatorBias} applies results from \cite%
{opsomer1999root} for the semiparametric additive model. Define $\xi_{\Gamma
X}$ as: 
\begin{align*}
\xi_{\Gamma X}: = (NT)^{-1}\sum_{it}({\Gamma }_{X,it}-\widehat{\Gamma }%
_{X,it})\big(\Gamma_{it}- \hat \Gamma_{it}\big).
\end{align*}
We verify $RMSE(\xi_{\Gamma X}) = o_p(NT)^{-1/2}$, which is sufficient for
Theorem~\ref{Th: NO general}. Assume without loss that $X_{it}$ is scalar,
and expand ${\Gamma }_{X}-\widehat{\Gamma }_{X} = \widetilde{\Gamma }_{X} -
\hat\eta$, where 
\begin{align*}
\widetilde{\Gamma }_{X} = (\mathbb{I} - W^{(1)}) {\Gamma }_{X}(\mathbb{I} -
W^{(2)})^{\prime }, & & \hat\eta = W^{(1)}\eta - \eta W^{(2)\,^{\prime }} +
W^{(1)}\eta W^{(2)\,^{\prime }}
\end{align*}
Define $\widetilde{\Gamma } $ and $\hat\varepsilon$ similarly. In turn the
definition of $\xi_{\Gamma X}$ implies, 
\begin{align*}
\xi_{\Gamma X} = (NT)^{-1}\left[ tr\{\widetilde{\Gamma}\widetilde{\Gamma }%
_{X}\} - tr\{\widetilde{\Gamma}\hat\eta\} - tr\{\hat\varepsilon\widetilde{%
\Gamma }_{X}\} + tr\{\hat\varepsilon\hat\eta\} \right].
\end{align*}
To show any variance added to the estimates of $\hat\beta_{NO}^{SS}$ from $%
\hat\eta$ and $\hat\varepsilon$ is sufficiently small order, i.e. $%
o_p(NT)^{-1/2}$, it is convenient to analyse bias and variance of each term
in this expression directly.

\begin{lemma}
\label{lemma:withinEstimatorBias} Impose \eqref{eq: rate cond 2},
Assumptions in Theorem~\ref{thm:SSWWresult}, Assumptions, \ref{ass:ID},~\ref%
{thm:additive model}, ~\ref{ass:singularvalues}-\ref{ass:Proxies}, and
Assumptions~\ref{ass:AppNorm}, \ref{ass:AppError}, and \ref%
{ass:AppCorrelation} from the appendix. 
Let $\xi_f^2$, $\xi_\lambda^2$ be from \eqref{eqn:factorRates}. Then, 
\begin{align*}
\mathbb{E}[\xi_{\Gamma X} ] &= R_2^2\cdot O_p(h_\lambda^2 h_f^2
+\xi_\lambda^2\xi_f^2)
\end{align*}
\end{lemma}

In the following Lemma we show the results hold under mutually independent
proxies, which leads to much simpler theoretical properties. We present
Lemma~\ref{lemma:withinEstimatorBiasIndependent} as it may be of practical
interest.

\begin{lemma}
\label{lemma:withinEstimatorBiasIndependent} Make Lemma~\ref%
{lemma:withinEstimatorBias} assumptions. If proxies are independent over $%
r=1,\dots, R_2$. Then, 
\begin{align*}
\mathbb{E}[\xi_{\Gamma X} ] &= R_2^2\cdot O_p(h_\lambda^2 h_f^2
+\xi_\lambda^2\xi_f^2)
\end{align*}
\end{lemma}

Next we state a bound on the variance for the multiplicative error:

\begin{lemma}
\label{lemma:withinEstimatorVar} Impose Lemma~\ref{lemma:withinEstimatorBias}
assumptions. Then, 
\begin{align*}
Var[\xi_{\Gamma X}] &= \frac{R_2^6}{NT}\cdot O_p(\xi_\lambda^2\xi_f^2) + 
\frac{R_2^8}{NT}\cdot O_p(h_\lambda^2h_f^2 + (NT h_\lambda h_f)^{-1})
\end{align*}
\end{lemma}

\subsection{Estimation of $\protect\beta$}

\label{sect:EstBeta}

Finally, we use the rate results derived in the previous subsection to
verify the general conditions of Theorem \ref{thm:SSWWresult} when our novel
estimator of $g_Z$ is employed:

\begin{corollary}
\label{cor:withinEstimator} Impose Lemma~\ref{lemma:withinEstimatorBias}
assumptions. Let $h_\lambda = cN^{-\tau}, h_f = cT^{-\tau}$ and $\rho >8/3$.
For $\tau\in (1/4,(3\rho - 2)/3\rho)$, $R_2 \lesssim
c\cdot\min\{N,T\}^{1/6\rho}$; $N\sim T$, and Theorem \ref{lemma:bai2009fe}
implies, 
\begin{align*}
\mathbb{E}[\hat\beta_{NO}^{SS} - \beta_0] = o_p(NT)^{-1/2}, & & 
Var[\hat\beta_{NO}^{SS}|X]= O_p(NT)^{-1} + o_p(NT)^{-1}.
\end{align*}
\end{corollary}

In Corollary~\ref{cor:bai2009fe}, we show for $R_2 = O(\min\{N,T\}^{\frac{1}{%
6\rho}})$ that $R_2^2\cdot\xi_\lambda^2\xi_f^2 = o_p(\min\{N,T\}^{-1})$, and 
$R_2^6 (NT)^{-1}\xi_\lambda^2\xi_f^2 = o_p(\min\{N,T\}^{-2})$ for $N\sim T$.
For the range of $h_\lambda$, and $h_f$ stated in Corollary~\ref%
{cor:withinEstimator} all other terms in bias are $o_p(\min\{N,T\}^{-1})$,
and in variance are $o_p(\min\{N,T\}^{-2})$.

\begin{corollary}
\label{cor:Final} Impose Lemma~\ref{lemma:withinEstimatorBias} conditions.
Then, 
\begin{align*}
\sqrt{NT}(\widehat{\beta}_{NO}^{SS}-\beta ) \xrightarrow[]{d} \mathcal{N}%
\big(0, \Omega_X^{-1} \Sigma \Omega_X^{-1}\big).
\end{align*}
\end{corollary}

Recall in \eqref{eq: rate cond 2} that $\Omega_X := \text{plim} (NT)^{-1}
\sum_{it} \eta_{it}\eta_{it}^{\prime }$, and $\Sigma$ from $(NT)^{-1/2}
\sum_{it} \eta_{it}\varepsilon_{it} \xrightarrow[]{d}N(0,\Sigma)$.

\subsection{Implementation}

Estimation of the preliminary model, along with final estimation of the
multi step debiasing estimators involved the following hyperparameters: $%
\{R_1, R_2, h_\lambda, h_f\}$.

We motivate that $R_2$, i.e. the number of eigenfunctions used in the final
estimation step, should be relatively small with respect to sample size. In
practice, we advocate these being moderate, but fixed, in large enough
samples. For example, if $\alpha,\gamma \in \mathbb{R}^d$ then $R_2 \geq d$
should in most cases control all latent fixed-effect variation. Hence, as
long as eigenfunctions estimated in the first stage $R_1 \geq R_2 \geq d$,
then our model should identify all latent variation in $g(\alpha, \gamma)$.

Our first-step estimation of eigenfunction proxies, however, showed a clear
bias variance trade-off in $R_1$, regardless of the ambient dimension of $%
\alpha,\gamma \in \mathbb{R}^d$. This is because variation in the tail terms
of $g(\alpha, \gamma) = \sum_{r =1}^\infty \sigma_r u_r(\alpha) v_r(\gamma)$
may still influence estimation of leading terms, even as $\sigma_r \to 0$.
We conjecture here, without formal proof, that tests of a matrix rank
observed with noise should estimate the $R_1$ that optimally trades-off bias
and variance in this setting. That is, since in finite samples we would
optimally set $R_1$ equal to the highest $r$ such that $\sigma_r > \left\|
\varepsilon\right\|_2$, information criterion tests from e.g. \cite%
{BaiNg2002} or eigenvalue ratio test in \cite{ahn2013eigenvalue} and Section
5 of \cite{ke2024robust} for linear regression models specifically, should
work well. These tests are naturally designed to choose the point at which
the distribution of eigenvalues are related to noise, hence should do well
to establish that optimal $R_1$ for estimation.

Lastly, for bandwidths $\{h_\lambda, h_f\}$ we again conjecture without
proof that usual split sample cross-validation arguments should pick the
optimal bandwidths. Since bandwidths from Section~\ref{sect:EstSecondStep}
required to perform inference are the bandwidths that optimise mean squared
error, the objective function of these cross-validation hyperparameter
optimisers are aligned with our purposes. We leave formalising this for
future work.

Numerical implementations in simulation Section~\ref{sect:Simulations} and
Appendix~\ref{sect:AppSims} simply use the asymptotically optimal rates for $%
\{R_1, R_2, h_\lambda, h_f\}$ supposed by the theory. We make no strong
claim to estimating optimal hyperparameters, just that there exist some that
conform to these rates which produce good numerical results.

Lastly, in finite samples, there are degrees of freedom corrections to
consider. For the factor model, we use \cite{freeman2023linear} adjustments
to rescale variance estimators by, 
\begin{align*}
dfc = \frac{NT}{(N-R_1)(T - R_1)}
\end{align*}
Likewise, for the nonparametric estimators, 
\begin{align*}
dfc = \frac{NT}{(N-tr\{W^{(1)}\})(T - tr\{W^{(2)}\})}
\end{align*}
where $tr\{W^{(1)}\}$ and $tr\{W^{(2)}\}$ are effective degrees of freedom
from nonparametric estimation, see \cite{hastie2009introduction} Section
7.5.2, and $W^{(1)}, W^{(2)}$ are the nonparametric kernel weights derived
in Section~\ref{sect:MI} and \ref{sect:WW}. We also implement the correction
akin to \cite{rueda2013degrees}: 
\begin{align*}
dfc = \frac{NT}{\left(N-\min\{(1.5\cdot
tr\{W^{(1)}\},N/2\}\right)\cdot\left(T - \min\{1.5\cdot
tr\{W^{(2)}\},T/2\}\right)}.
\end{align*}
This correction produces wider confidence intervals, with better coverage
under dependence structures in the noise term. We make no formal claim to
the validity of these corrections, but note in simulations that calculated
standard errors approximate simulated standard deviations well.

We implement the weighted-within by using the multi-index weights in Section~%
\ref{sect:MI} to initialise estimates before using the backfitting in
Section~\ref{sect:WW}. Estimates under multi-index weights in Section~\ref%
{sect:MI} perform well with nominal coverage, but are dominated by the
backfitting in Section~\ref{sect:WW}, so we report only those results. For
the higher dimensional simulations in Appendix~\ref{sect:AppSims} we also
find it helps to initialise the backfitting with the standard factor model.

\section{Simulations\label{sect:Simulations}}

Data is generated according to 
\begin{align}  \label{eqn:sims}
Y_{it}=X_{it}\beta +g(\alpha _{i},\gamma _{t})+\varepsilon _{it}, & & 
X_{it}=g_{X}(\alpha _{i},\gamma _{t})+\eta _{it},
\end{align}%
where $\beta = 2$, $\alpha_i, \gamma_t \sim U(-1,1)$. For $\mu_{it}\sim
N(0,1), \mu_{it}^* \sim N(0,4)$, $\nu_{it} \sim N(0,\eta_{it}^2)$, noise
terms $\varepsilon _{it}$, and $\eta_{it}$ are generated, 
\begin{align*}
\eta_{it} = \mu_{it} + (\mu_{it-1} + \mu_{i-1t})/\sqrt{2} + \mu_{i-1t-1}/2 +
\mu_{it}^*, & & \varepsilon _{it} = \nu_{it} + (\nu_{it-1} + \nu_{i-1t})/%
\sqrt{2} + \nu_{i-1t-1}/2
\end{align*}
then both normalised to have variance 1. Term $\mu_{it}^*$ normalises $%
\eta_{it}$ to not have too much correlation. In this way they admit
conditional heteroskedasticity and correlation over $i$, and $t$. To
simulate that the cross-sectional ordering is unknown to the econometrician,
we randomise the order of $i$, but maintain the order of $t$. Finally, 
\begin{equation*}
g(\alpha _{i},\gamma _{t})=\frac{1}{\theta \sqrt{2\pi }}\exp \left( -\frac{%
(\alpha _{i}-\gamma _{t})^{2}}{\theta ^{2}}\right) ,\quad \quad g_{X}(\alpha
_{i},\gamma _{t})=\frac{1}{(|\alpha _{i}-\gamma _{t}|+1)^{\theta }},\quad
\quad \theta =1/2.
\end{equation*}%
Terms $g(\alpha _{i},\gamma _{t})$ and $g_{X}(\alpha _{i},\gamma _{t})$ are
normalised to variance four.\footnote{%
Rescaling variance of these terms does not impact asymptotic results, but
elucidates differences across estimators with much smaller sample sizes in
simulations.} Figure~\ref{fig:SimEmpiricalBounds} displays the results
graphically, Table~\ref{tbl:Simulation} tables the bias, root mean squared
error, and coverage for 95\% nominal test.

Standard errors across all estimators use the partial sum estimator from 
\cite{BaiNg2006} in conjunction with the \cite{newey1987hac} kernel
estimator in the time dimension.

Tuning parameters are chosen as follows. For the factor model estimators,
denote $R_1=\min \{N,T\}^{1/3}$. For the nonparametric weighted-within
estimator these are set to $R_2= 4$, i.e. constant, since the optimal rate $%
R_2 = \min\{N,T\}^{1/6\rho}$ implies such slow rate of growth for reasonably
smooth functions, that in practice optimal $R_2$ would not change over the
sample sizes we consider. Notice here that $R_2 < R_1$, which follows our
theoretical arguments: we find for first-stage preliminary estimators the
bias and variance is better traded off with higher number of estimated
components than in the second stage. Bandwidths for the nonparametric
weighted-within estimator are set to $h=(25/\min \{N,T\})^{1/2}$ and for the
Oracle estimator are set to $h=(1/\min \{N,T\})^{1/2}$. Both weighted-within
and Oracle estimators use the Gaussian kernel.

In addition to the estimators compared in Table~\ref{tbl:Simulation} and
Figure~\ref{fig:SimEmpiricalBounds} we also implement our estimator with the 
\cite{zhang2017estimating} pseudo-metric from \eqref{eqn:ZhangDistance}, and
also with cross-sectional/time-serial first moments, e.g. those used in \cite%
{bonhomme2021discretizing}. We report these results in Appendix~\ref%
{sect:AppSims}. As conjectured, smoothers using the \cite%
{zhang2017estimating} pseudo-metric performs well when $\alpha _{i}$ and $%
\gamma _{t}$ are scalars, but scales poorly as their dimension increases.%
\footnote{%
The \cite{zhang2017estimating} pseudo-metric is computationally very costly,
so we restrict Monte Carlo rounds to 1,000 and limit sample size to at most $%
N=T=300$.} Our eigenfunction based method performs well for
higher-dimensional $\alpha _{i}$ and $\gamma _{t}$, but does require
smoother functions. This is predicted in our theory, where error decreases
in $\rho =p/\min \{d_{\alpha },d_{\gamma }\}$. The moment based estimator
performs poorly regardless of dimension for $\alpha _{i}$ and $\gamma _{t}$,
since no cross-sectional or time-serial moments are injective for radial
type functions proposed in \eqref{eqn:sims}.\footnote{%
In Appendix~\ref{sect:AppSims} we implement DGPs that produce nominal
coverage for the moment based estimator, and compare performance in those
settings.}

\begin{figure}[tbp]
\centering
\includegraphics[width=0.45\linewidth]{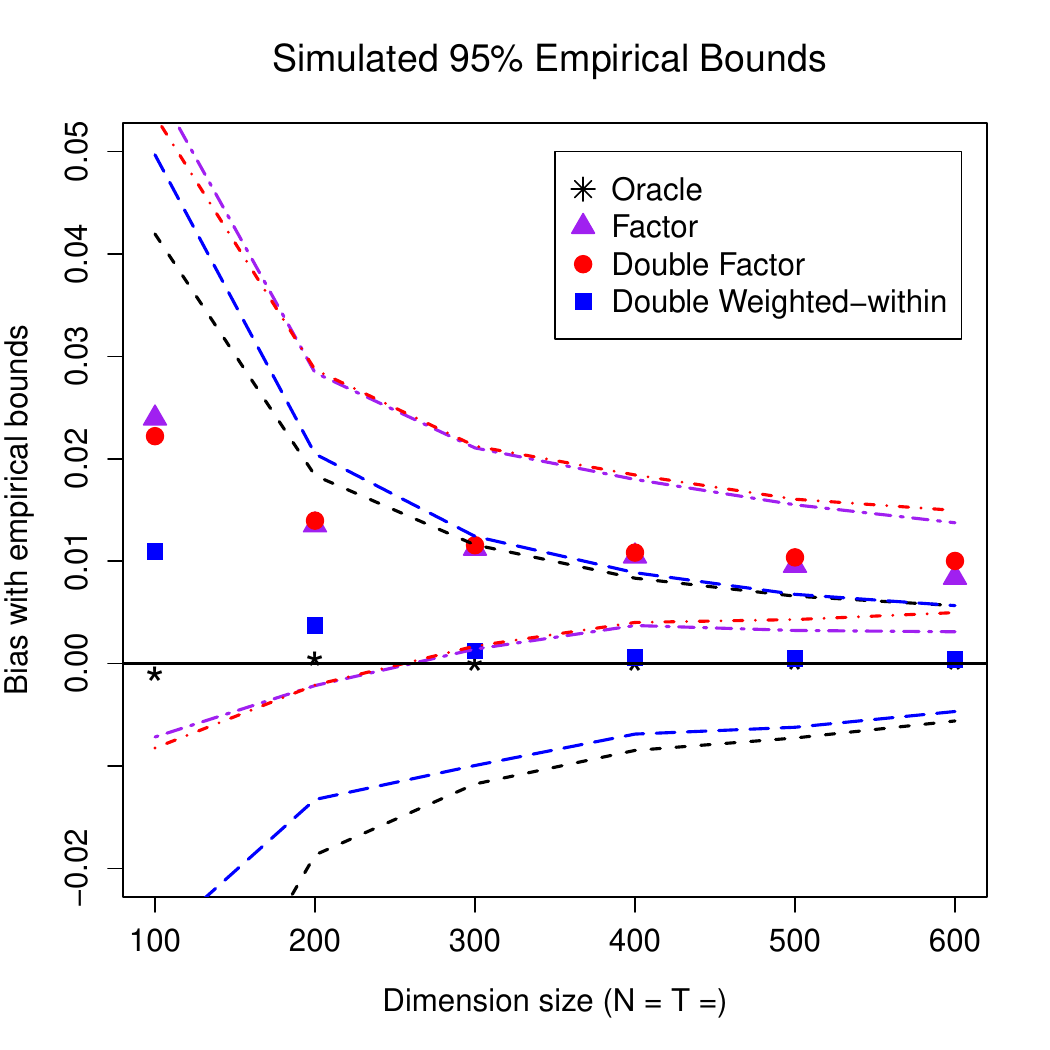} %
\includegraphics[width=0.45\linewidth]{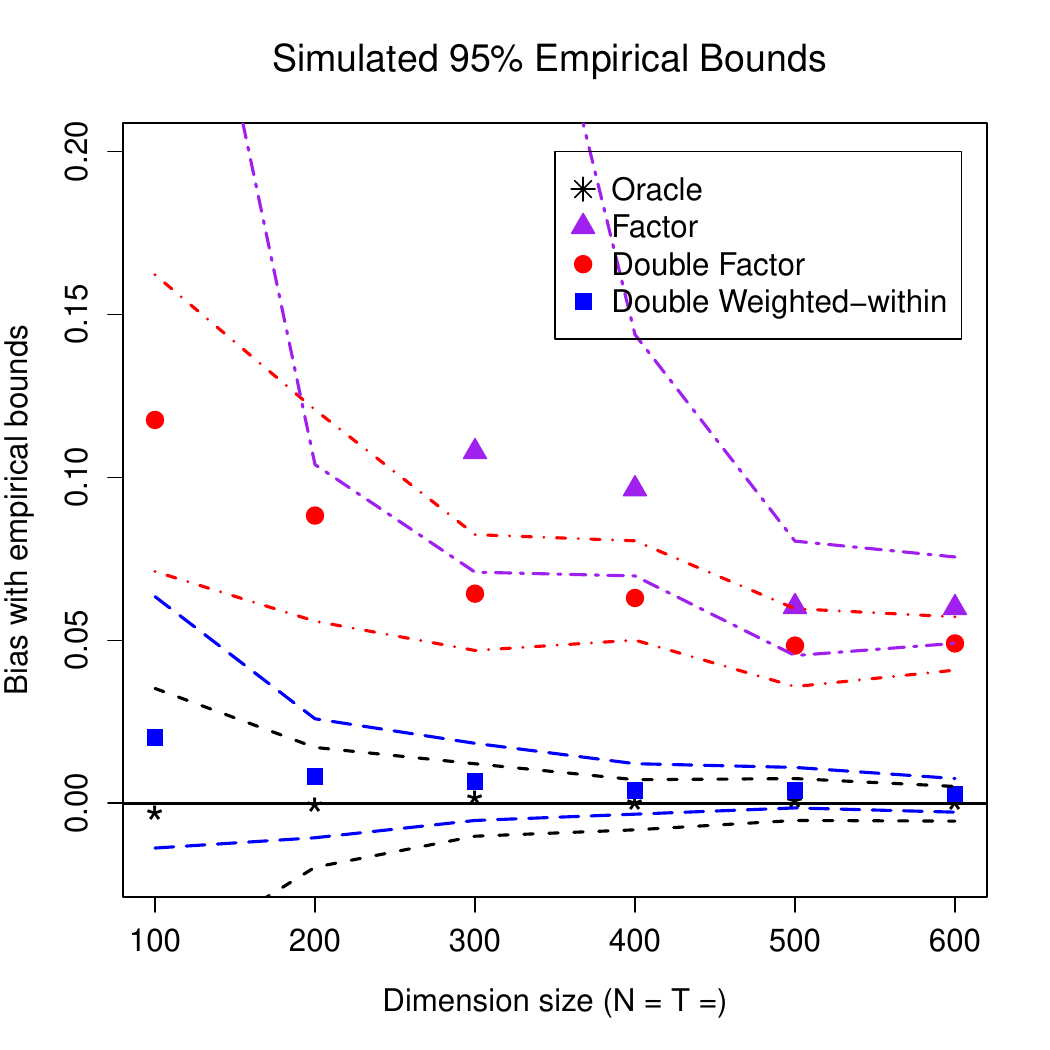}
\caption{Bias with 95\% empirical bounds. \newline
{\protect\footnotesize Points are bias. Lines are 95\% empirical bounds
across simulations. Oracle is the Neyman orthogonal estimator using the
weight-within transformation with known $(\protect\alpha_i, \protect\gamma%
_t) $. Factor is the standard factor model. Double factor estimates a factor
model for both $Y$ and $X$, and projects both out for $Y$ and $X$.
Weighted-within is the Neyman orthogonal estimator with weighted-within
transformation. Left panel is \eqref{eqn:sims} with $\protect\theta = 1/2$,
right panel sets $\protect\theta = 1/8$. $R_2 = 4$ in left panel and $R_2 =6$
in right panel.} }
\label{fig:SimEmpiricalBounds}
\end{figure}

\begin{table}[tbp]
\centering
\begin{tabular}{r|rrrr|}
& \multicolumn{4}{c}{Bias} \\ \hline
Dim. & Oracle & Factor & D. Factor & WW \\ \hline
100 & -0.0011 & 0.0240 & 0.0222 & 0.0109 \\ 
200 & 0.0004 & 0.0135 & 0.0140 & 0.0037 \\ 
300 & -0.0001 & 0.0112 & 0.0116 & 0.0012 \\ 
400 & -0.0001 & 0.0105 & 0.0108 & 0.0006 \\ 
500 & -0.0000 & 0.0096 & 0.0104 & 0.0005 \\ 
600 & 0.0000 & 0.0084 & 0.0100 & 0.0004 \\ \hline
\end{tabular}
\centering
\begin{tabular}{|rrrr}
\multicolumn{4}{c}{RMSE} \\ \hline
Oracle & Factor & D. Factor & WW \\ \hline
0.022 & 0.029 & 0.027 & 0.022 \\ 
0.010 & 0.016 & 0.016 & 0.009 \\ 
0.006 & 0.012 & 0.013 & 0.006 \\ 
0.004 & 0.011 & 0.011 & 0.004 \\ 
0.004 & 0.010 & 0.011 & 0.003 \\ 
0.003 & 0.009 & 0.010 & 0.003 \\ \hline
\end{tabular}
\begin{tabular}{r|rrrr}
& \multicolumn{4}{c}{Coverage} \\ \hline
Dim. & Oracle & Factor & D. Factor & WW \\ \hline
100 & 0.95 & 0.53 & 0.61 & 0.98 \\ 
200 & 0.93 & 0.47 & 0.45 & 0.99 \\ 
300 & 0.93 & 0.31 & 0.31 & 0.97 \\ 
400 & 0.93 & 0.14 & 0.12 & 0.96 \\ 
500 & 0.92 & 0.08 & 0.05 & 0.96 \\ 
600 & 0.91 & 0.08 & 0.01 & 0.95 \\ \hline
\end{tabular}%
\par
{\footnotesize {Dim. is the size of each dimension, with $N=T$. Oracle uses
the weight-within transformation with known $(\alpha_i, \gamma_t)$. Factor
is a factor model approximation. D. Factor for double factor uses a factor
model approximation for $Y$ and $X$. WW is the Neyman Orthogonal Weighted
Within estimator. }}
\caption{Bias, RMSE and coverage (5\% nominal test) for 1,000 Monte Carlo
rounds}
\label{tbl:Simulation}
\end{table}

\section{Conclusion\label{sect:Conclusion}}

In this paper we present novel theory for the linear panel setting with a
general function specification of unobserved heterogeneity. Using the
double-debias approach proposed in \cite{freeman2022multidimensional},
deriving from \cite{chernozhukov2022locally} methods, along with preliminary
estimators from \cite{freeman2023linear,freeman2022multidimensional}, we
show that under low level regularity conditions on the function of
unobserved heterogeneity, statistical inference on estimates follows. In
particular, we show properties of eigenfunctions that offer an incidental
debias in the tail terms in the functional singular value decomposition of
functions in a Hilbert space, when the weighted-within transformation from 
\cite{freeman2022multidimensional} is implemented in this setting. The
second orthogonalisation over covariates in the Neyman orthogonal estimator
allows this incidental debias to be asymptotically weaker than just applying
the weighted-within transformation to the equation for dependent variables.
The Neyman orthogonal estimator in general allows for weaker asymptotic
convergence for estimates of fixed-effects in either the dependent or
independent variables.

\appendix

\renewcommand{\thetheorem}{A.\arabic{theorem}} \setcounter{theorem}{0} %
\renewcommand{\thelemma}{A.\arabic{lemma}} \setcounter{lemma}{0} %
\renewcommand{\theproposition}{A.\arabic{proposition}} %
\setcounter{proposition}{0} \renewcommand{\theequation}{A.\arabic{equation}} %
\setcounter{equation}{0} \renewcommand{\theassumption}{A.\arabic{assumption}}
\setcounter{assumption}{0}

\section{Proof of technical results\label{sect:AppProofs}}

\subsection{Proof of Section~\protect\ref{sect:NO} results}

\begin{proof}[Proof of Lemma~\protect\ref{Lem: Stoch Eq}]
For any given $\epsilon >0$, there exists $\delta _{1},\delta _{2}>0$ so
that, for $A=X,Y$,%
\begin{eqnarray*}
&&\lim_{n,T\rightarrow \infty }\sup \Pr \left( \sqrt{NT}\left\Vert \hat{\nu}%
_{A}^{\ast }\left( \hat{g}_{0,A}\right) -\hat{\nu}_{A}^{\ast }\left(
g_{0,A}\right) \right\Vert >\epsilon \right) \\
&\leq &\lim_{n,T\rightarrow \infty }\sup \Pr \left( \sqrt{NT}\left\Vert \hat{%
\nu}_{A}^{\ast }\left( \hat{g}_{0,A}\right) -\hat{\nu}_{A}^{\ast }\left(
g_{0,A}\right) \right\Vert >\epsilon ,\left\Vert \hat{g}_{0,A}-g_{0,A}\right%
\Vert _{\mathcal{G}}<\delta _{1}\right) \\
&\leq &\lim_{n,T\rightarrow \infty }\sup \Pr \left( \sup_{\left\Vert \hat{g}%
_{0,A}-g_{0,A}\right\Vert _{\mathcal{G}}<\delta _{1}}\left\Vert \sqrt{NT}%
\hat{\nu}_{A}^{\ast }\left( \hat{g}_{0,A}\right) -\sqrt{NT}\hat{\nu}%
_{A}^{\ast }\left( g_{0,A}\right) \right\Vert >\epsilon \right) <\delta _{2},
\end{eqnarray*}%
where the last inequality follows by stochastic equicontinuity of $%
g_{A}\mapsto \sqrt{NT}\hat{\nu}_{A}^{\ast }(g_{A})$ at the data--generating
values $g_{0,A}$.

The second part follows from the Proposition on p. 50 \cite{Andrews1994}.
\end{proof}

\begin{proof}[Proof of Theorem \protect\ref{thm:sampleSplit}]
First note that the expansion of $\hat{\beta}_{NO}$ also applies to $%
\widehat{\beta }_{NO}^{SS}$, except that $\sum_{it}$ has to be replaced by $%
\sum_{k_{1},k_{2}=1}^{2}\sum_{\left( i,t\right) \in \mathcal{I}%
_{k_{1},k_{2}}}$ everywhere. Thus, the result of Theorem \ref{Th: NO general}
follows under the same conditions (\ref{eq: rate cond 0}) and (\ref{eq: rate
cond 2}) with $\sum_{it}$ replaced by $\sum_{k_{1},k_{2}=1}^{2}\sum_{\left(
i,t\right) \in \mathcal{I}_{k_{1},k_{2}}}$. Next,\ choose any partition $%
\left( k_{1},k_{2}\right) $ and, for any two matrices $A,B\in \mathbb{R}%
^{m\times n}$, let $\langle A,B\rangle
_{F}:=\sum_{i=1}^{m}\sum_{t=1}^{n}A_{it}B_{it}/\left( mn\right) =tr\left\{
A^{\prime }B\right\} /\left( mn\right) $ denote the scaled entrywise
Frobenius inner product. We then with to show that $\langle \Gamma _{X,k}-%
\hat{\Gamma}_{X.k}^{\left( k_{1},k_{2}\right) },\varepsilon ^{\left(
k_{1},k_{2}\right) }\rangle _{F}=o_{p}(1/\sqrt{NT})$ and \ \ $\langle \eta
_{k}^{\left( k_{1},k_{2}\right) },\Gamma _{Y}-\hat{\Gamma}_{Y}^{\left(
k_{1},k_{2}\right) }\rangle _{F}=o_{p}(1/\sqrt{NT})$, $k=1,...,K$. We have%
\begin{eqnarray*}
\left\vert \mathbb{E}\left[ \langle \Gamma _{X,k}-\hat{\Gamma}_{X.k}^{\left(
k_{1},k_{2}\right) },\varepsilon ^{\left( k_{1},k_{2}\right) }\rangle _{F}|%
\mathcal{F}^{\left( k_{1},k_{2}\right) }\right] \right\vert &=&\left\vert
\langle \Gamma _{X,k}-\hat{\Gamma}_{X.k}^{\left( k_{1},k_{2}\right) },%
\mathbb{E}\left[ \varepsilon ^{\left( k_{1},k_{2}\right) }|\mathcal{F}%
^{\left( k_{1},k_{2}\right) }\right] \rangle _{F}\right\vert \\
&\leq &\sqrt{\left\Vert \Gamma _{X,k}-\hat{\Gamma}_{X.k}^{\left(
k_{1},k_{2}\right) }\right\Vert _{F}\left\Vert \mathbb{E}\left[ \varepsilon
^{\left( k_{1},k_{2}\right) }|\mathcal{F}^{\left( k_{1},k_{2}\right) }\right]
\right\Vert _{F}}.
\end{eqnarray*}%
Next, recall that $tr\left\{ A^{\prime }B\right\} =vec\left( A\right)
^{\prime }vec\left( B\right) $ so that%
\begin{eqnarray*}
&&\mathbb{E}\left[ \langle \Gamma _{X,k}-\hat{\Gamma}_{X.k}^{\left(
k_{1},k_{2}\right) },\varepsilon ^{\left( k_{1},k_{2}\right) }\rangle
_{F}^{2}|\mathcal{F}^{\left( k_{1},k_{2}\right) }\right] \\
&=&\mathbb{E}\left[ vec\left( \Gamma _{X,k}-\hat{\Gamma}_{X.k}^{\left(
k_{1},k_{2}\right) }\right) ^{\prime }vec\left( \varepsilon ^{\left(
k_{1},k_{2}\right) }\right) vec\left( \varepsilon ^{\left(
k_{1},k_{2}\right) }\right) ^{\prime }vec\left( \Gamma _{X,k}-\hat{\Gamma}%
_{X.k}^{\left( k_{1},k_{2}\right) }\right) |\mathcal{F}^{\left(
k_{1},k_{2}\right) }\right] \\
&=&vec\left( \Gamma _{X,k}-\hat{\Gamma}_{X.k}^{\left( k_{1},k_{2}\right)
}\right) ^{\prime }\mathbb{E}\left[ vec\left( \varepsilon ^{\left(
k_{1},k_{2}\right) }\right) vec\left( \varepsilon \right) ^{\prime }|%
\mathcal{F}^{\left( k_{1},k_{2}\right) }\right] vec\left( \Gamma _{X,k}-\hat{%
\Gamma}_{X.k}^{\left( k_{1},k_{2}\right) }\right) \\
&\leq &\left\Vert \Gamma _{X,k}-\hat{\Gamma}_{X.k}^{\left(
k_{1},k_{2}\right) }\right\Vert _{F}^{2}\left\Vert \mathbb{E}\left[
vec\left( \varepsilon ^{\left( k_{1},k_{2}\right) }\right) vec\left(
\varepsilon ^{\left( k_{1},k_{2}\right) }\right) ^{\prime }|\mathcal{F}%
^{\left( k_{1},k_{2}\right) }\right] \right\Vert _{op}
\end{eqnarray*}

The result now follows from Lemma 6.1 in \cite{Chernozhukov2018}. \hfill
\end{proof}

\begin{proof}[Proof of Theorem~\protect\ref{thm:SSWWresult}]
We consider a given split $\left( k_{1},k_{2}\right) \in \left\{ \left(
1,1\right) ,\left( 1,2\right) ,\left( 2,1\right) ,\left( 2,2\right) \right\} 
$ and, for ease of notation, suppress dependence of $Z^{\left(
k_{1},k_{2}\right) }$, $\mathcal{W}^{\left( k_{1},k_{2}\right) }$, $\mathcal{%
F}^{\left( k_{1},k_{2}\right) }$, etc., on $\left( k_{1},k_{2}\right) $ in
the following, and so, for example, write $Z$, $\mathcal{W}$ and $\mathcal{F}
$. 
Moreover, conditional on $\mathcal{G}=\mathcal{F}\left( \lambda
_{i},f_{t}:i=1,...,N,t=1,...,T\right) $, we again make Theorem~\ref%
{thm:sampleSplit} assumptions including \eqref{eq: sample split cond}
conditions.  
With $\Gamma -\hat{\Gamma}=\Gamma -\mathcal{W}(\Gamma +\varepsilon )$ and $%
\Gamma _{X}-\hat{\Gamma}_{X}=\Gamma _{X}-\mathcal{W}(\Gamma _{X}+\eta )$, 
\begin{align*}
\mathbb{E}[\xi _{\Gamma X}|\mathcal{F}]& =\mathbb{E}[\langle \Gamma -%
\mathcal{W}\Gamma -\mathcal{W}\varepsilon ,\Gamma _{X}-\mathcal{W}\Gamma
_{X}-\mathcal{W}\eta \rangle _{F}|\mathcal{F}] \\
& =\mathbb{E}[\langle \Gamma -\mathcal{W}\Gamma ,\Gamma _{X}-\mathcal{W}%
\Gamma _{X}\rangle _{F}|\mathcal{F}]-\mathbb{E}[\langle \Gamma -\mathcal{W}%
\Gamma ,\mathcal{W}\eta \rangle _{F}|\mathcal{F}] \\
&\quad -\mathbb{E}[\langle \mathcal{W}\varepsilon ,\Gamma _{X}-\mathcal{W}%
\Gamma _{X}\rangle _{F}|\mathcal{F]}+\mathbb{E}[\langle \mathcal{W}%
\varepsilon ,\mathcal{W}\eta \rangle _{F}|\mathcal{F}]
\end{align*}%
With $\mathcal{F}$ being the filtration that defines $\mathcal{W}$. In
general, we have $\mathbb{E}\left[ \mathcal{W}\varepsilon |\mathcal{F},%
\mathcal{G}\right] =\mathcal{W}\mathbb{E}\left[ \varepsilon |\mathcal{F},%
\mathcal{G}\right]$. Full independence is sufficient, but not strictly
necessary for our result. Consider, 
\begin{align*}
\mathbb{E}[\langle \mathcal{W}\varepsilon ,\mathcal{W}\eta \rangle _{F}|%
\mathcal{F}] = \langle\mathcal{W}\mathbb{E}[\varepsilon|\mathcal{F}], 
\mathcal{W}\mathbb{E}[\eta|\mathcal{F}] \rangle _{F}.
\end{align*}
This can be shown to be $o_p(NT)^{-1/2}$ under mild stationarity conditions,
e.g., for AR(1) processes, see Remark~\ref{rem:weakDep}. Likewise, 
\begin{equation*}
\mathbb{E}[\langle \mathcal{W}\varepsilon ,\Gamma _{X}-\mathcal{W}\Gamma
_{X}\rangle _{F}|\mathcal{F},\mathcal{G}]=\langle \mathcal{W}\mathbb{E}\left[
\varepsilon |\mathcal{F},\mathcal{G}\right] ,\Gamma _{X}-\mathcal{W}\Gamma
_{X}\rangle _{F} \leq O_p((NT)^{-1/2}\xi_X) \|\mathcal{W}\|.
\end{equation*}%
thus, $\xi_X\|\mathcal{W}\| = o_p(1)$ is sufficient for the term to be $%
o_p(NT)^{-1/2}$.

Next, we analyse the variances of the three last terms. In the $N\times T$
matrix case, the operator $\mathcal{W}$ nontrivially, but linearly,
transforms rows and columns. Generically, the vectorised term $vec(\mathcal{W%
}\varepsilon )=\mathbb{W}vec(\varepsilon )$ where $\mathbb{W}\in \mathbb{R}%
^{NT\times NT}$, which has an adjoint $\mathbb{W}^{\ast }\in \mathbb{R}%
^{NT\times NT}$. Hence, we can write, 
\begin{equation*}
\mathbb{E}[\langle \mathcal{W}\varepsilon ,\mathcal{W}\eta \rangle _{F}^{2}]=%
\frac{1}{(NT)^{2}}\mathbb{E}\left[ \langle \mathbb{W}vec(\varepsilon ),%
\mathbb{W}vec(\eta )\rangle ^{2}\right] =\frac{1}{(NT)^{2}}\mathbb{E}\left[
\langle vec(\varepsilon ),\mathbb{W}^{\ast }\mathbb{W}vec(\eta )\rangle ^{2}%
\right]
\end{equation*}%
where the adjoint $\mathbb{W}^{\ast }$ exists since we only consider bounded
operators. Define $x_k:= \mathbb{W}^{\ast }\mathbb{W}vec(\eta_k )$ and $%
\mathcal{E}:= \mathbb{E}\left[ vec(\varepsilon )vec(\varepsilon )^{\prime
}|X ,\mathcal{F}\right]$. Then, denoting $\|\cdot\|_2$ as a matrix's largest
singular value, 
\begin{align*}
\frac{1}{(NT)^{2}}\mathbb{E}& \left[ [\mathbb{W}^{\ast }\mathbb{W}vec(\eta_k
)]^{\prime } vec(\varepsilon )vec(\varepsilon )^{\prime } \mathbb{W}^{\ast }%
\mathbb{W}vec(\eta_k )|X ,\mathcal{F}\right] = \frac{1}{(NT)^{2}} x_k^\prime 
\mathcal{E} x_k \leq \frac{1}{(NT)^{2}} \|x_k\|^2 \|\mathcal{E}\|_2
\end{align*}%
The class of linear operators we consider are collections of row and column
transformation. Hence, 
\begin{equation*}
\|x_k\|^2 = \Vert \mathcal{W}^{\ast }\mathcal{W}\eta_k \Vert _{F}^{2}\leq
\Vert \mathcal{W}\Vert ^{2}\cdot \Vert \mathcal{W}\eta_k \Vert _{F}^{2}
\end{equation*}%
where $\Vert \mathcal{W}\Vert $ is the operator norm for the row and column
transformations embedded in $\mathcal{W}$.

Likewise we can also bound the term, 
\begin{align*}
\mathbb{E}\left[ \langle \mathcal{W}\varepsilon ,\Gamma _{X_k}-\mathcal{W}%
\Gamma _{X_k}\rangle_F^2 |X ,\mathcal{F}\right] &= \frac{1}{(NT)^2} \mathbb{E%
}\left[ \langle vec(\varepsilon), \mathbb{W}^* vec(\Gamma _{X_k}-\mathcal{W}%
\Gamma _{X_k}) \rangle^2 |X ,\mathcal{F}\right] \\
&\leq \frac{1}{(NT)^2} \|\mathcal{W}\|^2 \|\mathcal{E}\|_2 \|\Gamma _{X_k}-%
\mathcal{W}\Gamma _{X_k}\|_F^2.
\end{align*}

\hfill
\end{proof}

The $\mathcal{E}$ term turns out to be a Toeplitz matrix, with singular
values bounded as follows:

\begin{lemma}
Assume there exists a permutation of $\{1, \dots, N-1\}$, called $\pi$ such
that $\gamma(n,m) := \{\mathbb{E}[\varepsilon_{\pi_i t} \varepsilon_{\pi_js}
| \mathcal{F}] : |\pi_i-\pi_j| = n, |t-s| = m\}$ is identical for all $%
\pi_i,t,\pi_j,s$, i.e. $\varepsilon$ is covariance stationary. Restrict $%
\gamma(n,m) \to 0$ for all $n$ as $m\to\infty$, which insists that time is
ordinal, and that distance in time necessarily reduces correlation. Also
assume weak dependence in the cross-section such that $\sum_{n=1}^{N-1}
\gamma(n,m) \to C\cdot \gamma(0,m)$ as $N\to\infty$. Then $\|\mathcal{E}\|_2
= O_p(1)$.
\end{lemma}

\begin{proof}
Without loss assume $N,T$ are even. Define the matrix $\mathcal{E}$ as
before with permuted $i$ indices, 
\begin{align*}
\mathcal{E} := \mathbb{E}[vec(\varepsilon) vec(\varepsilon)^{\prime }| 
\mathcal{F}] = 
\begin{bmatrix}
\gamma(0,0) & \gamma(1,0) & \dots & \gamma(N-1,T -1) \\ 
\gamma(1,0) & \gamma(0,0) & \dots & \gamma(N-2,T -1) \\ 
\vdots &  &  &  \\ 
\gamma(N-1,T -1) & \gamma(N-2,T -1) & \dots & \gamma(0,0)%
\end{bmatrix}%
\end{align*}
Then by Ostrowski's disc theorem, all singular values are bounded in the
disc, 
\begin{align*}
\sigma_r(\mathcal{E}) \subseteq \bigcup_{\ell} \Big\{z: |z - \gamma(0,0)|
\leq \sum_{k\neq \ell} |\mathcal{E}_{k \ell}|\Big\}.
\end{align*}
From $\gamma(n,m) \to 0$, these discs are widest for $\ell$ such that $%
m=T/2, \dots, 0, \dots T/2-1$ in $\gamma(n,m)$, i.e. columns of $\mathcal{E}$
that contains $\gamma(n,m)$ furthest from $0$. Call $M = \{T/2, \dots, 0,
\dots T/2-1\}$. Call $n^{\prime }$ possible sequences for $n$ in $%
\gamma(n,m) $ through the columns of $\mathcal{E}$. Then, 
\begin{align*}
\max_{r}\sigma_r(\mathcal{E}) = \max_{\ell}( \gamma(0,0) + \sum_{k\neq \ell}
|\mathcal{E}_{k \ell}|) &= \max_{n^{\prime }} \sum_{n \in n^{\prime }}
\sum_{m\in M} |\gamma(n,m)| \to C\gamma(0,0) + 2C\sum_{m = 1}^{\infty}
|\gamma(0,m)|.
\end{align*}
This is bounded under standard long run variance conditions such that $%
\sum_{m = 1}^{\infty} |\gamma(0,m)| = O(1)$. Hence, our assumption that $\|%
\mathcal{E}\|_2$ is bounded amounts to a bounded long run variance
condition, with uniform weak dependence over $i$. It is trivially satisfied
for e.g., i.i.d. data since $\|\mathbb{E}\left[vec(\varepsilon)
vec(\varepsilon)^\prime | \mathcal{F}\right]\|_2 = O_p(\sigma_\varepsilon^2)$%
, and i.n.i.d leads to $\|\mathbb{E}\left[vec(\varepsilon)
vec(\varepsilon)^\prime |\mathcal{F}\right]\|_2 = \max_{it}O_p(\mathbb{E}%
[\varepsilon_{it}^2])$.

\hfill
\end{proof}

\begin{proof}[Proof of Remark~\protect\ref{rem:weakDep}]
Take the AR(1) process for covariate noise term, $\eta_{it}$, 
\begin{align*}
\eta_{it} = \rho \eta_{it-1} + e_{\eta,it}, & & |\rho|<1, & & e_{\eta,it}
\sim iid, \mathbb{E}e_{\eta,it} = 0, & & \mathbb{E}e_{\eta,it}^4 <\infty,
\end{align*}
where $e_{\eta,it}$ is absolutely integrable. Denote $\sigma^2:= \mathbb{E}%
e_{\eta,it}^2$. Without loss, assume the partition $\mathcal{I}%
_{k_{1},k_{2}} $ we consider are for $t\geq1$, and the conditioning
filtration $\mathcal{F}_{k_{1},k_{2}}$ are for $t\leq0$. Then the term 
\begin{eqnarray*}
\mathbb{E}&\left[ \frac{1}{NT}\sum_{\left( i,t\right) \in \mathcal{I}%
_{k_{1},k_{2}}}\eta _{it}(\hat{\Gamma}_{Y,it}^{\left( k_{1},k_{2}\right)
}-\Gamma _{Y,it})|\mathcal{F}_{k_{1},k_{2}}\right] =\frac{1}{NT}\sum_{\left(
i,t\right) \in \mathcal{I}_{k_{1},k_{2}}}\mathbb{E}\left[ \eta _{it}|%
\mathcal{F}_{k_{1},k_{2}}\right] (\hat{\Gamma}_{Y,it}^{\left(
k_{1},k_{2}\right) }-\Gamma _{Y,it}) \\
&= \frac{1}{NT}\sum_{\left( i,t\right) \in \mathcal{I}_{k_{1},k_{2}}}\rho^t
\eta_{i0}(\hat{\Gamma}_{Y,it}^{\left( k_{1},k_{2}\right) }-\Gamma _{Y,it})
\leq \left(\frac{1}{NT}\sum_{\left( i,t\right) \in \mathcal{I}%
_{k_{1},k_{2}}}\rho^{2t}\eta_{i0}^2 \right)^{1/2}O_p(\xi_{Y})
\end{eqnarray*}%
This is then bound by $O_p(T^{-1/2}\xi_{Y})$, which is $o_p(NT)^{-1/2}$ for $%
N\sim T$ and $\xi_{Y} = o_p(NT)^{-1/4}$.

To bound the second order term, begin with the term, 
\begin{align*}
\frac{1}{\left( NT\right) ^{2}}\sum_{\left( i,s\right) \in \mathcal{I}%
_{k_{1},k_{2}}}\sum_{\left( j,t\right) \in \mathcal{I}_{k_{1},k_{2}}}\mathbb{%
E}\left[ \eta _{is}^{\prime }\eta _{jt}|\mathcal{F}_{k_{1},k_{2}}\right] {(%
\hat{\Gamma}_{Y,is}^{\left( k_{1},k_{2}\right) }-\Gamma _{Y,is})(\hat{\Gamma}%
_{Y,jt}^{\left( k_{1},k_{2}\right) }-\Gamma _{Y,jt}) } \\
= \frac{1}{\left( NT\right) ^{2}}\sum_{\left( i,s\right) \in \mathcal{I}%
_{k_{1},k_{2}}}\sum_{\left( i,t\right) \in \mathcal{I}_{k_{1},k_{2}}}\mathbb{%
E}\left[ \eta _{is}^{\prime }\eta _{it}|\mathcal{F}_{k_{1},k_{2}}\right] {(%
\hat{\Gamma}_{Y,is}^{\left( k_{1},k_{2}\right) }-\Gamma _{Y,is})(\hat{\Gamma}%
_{Y,it}^{\left( k_{1},k_{2}\right) }-\Gamma _{Y,it}) }
\end{align*}
where the second line simply uses independence over $i$. By standard AR(1)
arguments there is, 
\begin{align*}
\mathbb{E}\left[ \eta _{is}^{\prime }\eta _{it}|\mathcal{F}_{k_{1},k_{2}}%
\right] = \rho^{t+s}\eta_{i0}^2 + \sum_{m =
0}^{\min\{t,s\}-1}\rho^{2m}\sigma^2
\end{align*}
Label 
$A_{it} := (\hat{\Gamma}_{Y,it}^{\left( k_{1},k_{2}\right) }-\Gamma _{Y,it})$%
. Then, 
\begin{align*}
\frac{1}{\left( NT\right) ^{2}}\sum_{its} \mathbb{E}\left[ \eta
_{is}^{\prime }\eta _{it}|\mathcal{F}_{k_{1},k_{2}}\right]A_{it}A_{is} = 
\frac{1}{\left( NT\right) ^{2}}\sum_{its} \rho^{t+s}\eta_{i0}^2A_{it}A_{is}
+ \frac{\sigma^2}{\left( NT\right) ^{2}}\sum_{its} \sum_{m =
0}^{\min\{t,s\}}\rho^{2m}A_{it}A_{is}.
\end{align*}
Use geometric series, $\sum_{m = 0}^{\min\{t,s\}-1}\rho^{2m} =
(1-\rho^{2\min\{t,s\}})/(1-\rho^2) = O(1)$, to bound second term, 
\begin{align*}
\frac{\sigma^2}{\left( NT\right) ^{2}}&\sum_{its} \sum_{m =
0}^{\min\{t,s\}-1}\rho^{2m}A_{it}A_{is} = \frac{\sigma^2}{\left( NT\right)
^{2}}\sum_{its} A_{it}A_{is}\cdot O(1) \\
&\leq O(1)\cdot\frac{\sigma^2}{\left( NT\right) ^{2}}\left(\sum_{its}
A_{it}^2\right)^{1/2}\left(\sum_{its} A_{is}^2\right)^{1/2} \leq \frac{%
\sigma^2}{\left( NT\right)}O_p(T\cdot \xi_Y^2)
\end{align*}
which is $o_p((NT)^{-1})$ for $N\sim T$ and $\xi_{Y} = o_p(NT)^{-1/4}$.
Next, two Cauchy-Schwarz bounds give, 
\begin{align*}
\frac{1}{\left( NT\right) ^{2}}\sum_{its} \rho^{t+s}\eta_{i0}^2A_{it}A_{is}
\leq \frac{1}{\left( NT\right) ^{2}} \left(\sum_{its}A_{it}^4\right)^{1/2}
\left(\sum_{ts}\rho^{2(t+s)}\sum_{i}\eta_{i0}^4\right)^{1/2} \\
\leq \frac{\sqrt{NT}}{\left( NT\right) ^{2}} \sum_{it}A_{it}^2 \left(\frac{1%
}{(1-\rho^2)^2}\frac{1}{N}\sum_{i}\eta_{i0}^4\right)^{1/2} = \frac{1}{\sqrt{%
NT}}O_p(\xi_Y^2)\cdot O_p(1).
\end{align*}
This is again $o_p((NT)^{-1})$ for $N\sim T$ and $\xi_{Y} = o_p(NT)^{-1/4}$.
The MA($\infty$) result follows similarly.

Condition $\langle\mathcal{W}\mathbb{E}[\varepsilon|\mathcal{F}], \mathcal{W}%
\mathbb{E}[\eta|\mathcal{F}] \rangle _{F} = o_p(NT)^{-1/2}$ holds in the
Remark~\ref{rem:weakDep} AR(1) process, considering again the additive
univariate weight example from above: 
\begin{align*}
\frac{1}{NT}\sum_{it} \sum_j w_{ij} \mathbb{E}[\varepsilon_{jt}|\mathcal{F}]
\sum_{j^\prime} w_{ij^\prime} \mathbb{E}[\eta_{j^\prime t}|\mathcal{F}] = 
\frac{1}{NT} \sum_j\varepsilon_{j0} \sum_{j^\prime} \eta_{j^\prime 0}
\sum_{i} w_{ij} w_{ij^\prime} \sum_t \rho^{t} \tilde\rho^{t} = \frac{R^2}{T}%
O_p({Nh})^{-1}
\end{align*}
where, without loss, we rearrange the conditioning set $\mathcal{F}$ to
contain any $t\leq 0$. By mean zero and finite variance in $%
\varepsilon_{i0},\eta_{i0}$ this is $o_p(NT)^{-1/2}$ for $R^2h^{-1} =
o(NT)^{-1/2}$, which is easy to satisfy.

\hfill \hfill
\end{proof}

\subsection{Proof of Section~\protect\ref{sect:Asymptotics} Results}

\begin{proof}[Proof of Lemma~\protect\ref{lemma:frewei2023}]
Refer to the online supplement to \cite{freeman2023linear} for Proof of
Theorem 1. The term $\|S\|_F^2$ there can be bound by $O_p(NTR^{-2\rho})$.
Then, again referencing terminology used in that paper, $R\eta/\sqrt{NT} =
O_p\big(R\min\{N,T\}^{-1/2}\big) + O_p(R^{1-\rho})$, which is the bounding
factor for $\widehat{\beta}_{LS} - \beta^0$. \hfill
\end{proof}

Here we state regularity conditions to apply theory from \cite%
{bai2023approximate}. We define $\delta \in (0,1]$ to be the constant such
that $N^{ 1-\delta}\Sigma_R^2$ converges in probability to pd, where $%
\Sigma_R$ is the diagonal matrix of singular values, $\sigma_r$ for $r
=1,\dots, R$. From Assumption~\ref{ass:singularvalues}, this constrains $%
R\lesssim N^{\frac{1-\delta}{2\rho +1}}$. For $R = R_X =
\min\{N,T\}^{1/6\rho}$ rate chosen in Lemma~\ref{lemma:bai2009fe} below,
this translates to $\delta = (4\rho - 1)/6\rho$ which is in $(0,1)$ for all $%
\rho\geq1$, hence admissable by the theory in \cite{bai2014theory}.

\begin{assumption}
\label{ass:AppNorm} 
\begin{enumerate*}[(i).,series = tobecont, itemjoin = \quad]
    \item $\mathbb{E}X_{it}^4 = O(1)$,
    \item $N^{-\delta}\lambda'\lambda \xrightarrow[]{p} \Sigma_\lambda >0$ for $\delta \in (0,1]$,
    \item  $T^{-1}f'f \xrightarrow[]{p} \Sigma_f >0$. 
        \item $\mathbb{E}[u_{ir}^4] <\infty$, $\mathbb{E}[v_{ir}^4] <\infty$
    \end{enumerate*}
\end{assumption}

\begin{assumption}
\label{ass:AppError}

\begin{enumerate}
\item $\mathbb{E}\left[ N^{-1/2} \sum_i (\varepsilon_{it}\varepsilon_{is} - 
\mathbb{E}[\varepsilon_{it}\varepsilon_{is}]) \right]^2 = O(1)$, 

\item For all $i$: $(\sqrt{N}T)^{-1}\|\varepsilon_i^{\prime }\varepsilon\|_F
= O_p(\min\{N,T\}^{-1})$, for all $t$: $(\sqrt{T}N)^{-1}\|\varepsilon_t^{%
\prime }\varepsilon\|_F = O_p(\min\{N,T\}^{-1})$.
\end{enumerate}
\end{assumption}

\begin{assumption}
\label{ass:AppCorrelation} For each $t$, 
\begin{enumerate*}[(i).,series = tobecont, itemjoin = \,\,\,]
    \item $\mathbb{E}\|\frac{1}{N^{\delta/2}}\sum_i\lambda_i\varepsilon_{it}\|^2 = O(1)$,
    \item $\frac{1}{NT}\varepsilon_t'\varepsilon'f = O_p(\min\{N,T\}^{-2}$;
    For each $i$,
    \item $\mathbb{E}\|\frac{1}{\sqrt{T}}\sum_tf_t\varepsilon_{it}\|^2 = O(1)$,
    \item $\frac{1}{NT}\varepsilon_i'\varepsilon'\lambda = O_p(N^{-\delta})$;
    \item $\lambda'\varepsilon'f = O_p(\sqrt{N^\delta T})$
    \end{enumerate*}
\end{assumption}

\begin{proof}[Proof of Lemma~\protect\ref{lemma:bai2009fe}]
Refer to Proposition A.1 in \cite{Bai2009}. 
The proof is mostly the same, with the addition of the approximation error$%
\sum_{r = R + 1}^{\min\{N,T\}} \lambda_{ir}f_{tr}$. Define $\Sigma_R$ is the
diagonal matrix of the $R$ first singular values $\sigma_r$. Take, 
\begin{align*}
\left[ \frac{1}{NT}\sum_i (Y_i - X_i\widehat{\beta})(Y_i - X_i\widehat{\beta}%
)^\prime \right] \widehat{f} = \widehat{f} \Sigma_R
\end{align*}
with $Y_i - X_i\widehat{\beta} = X_i(\beta^0 - \widehat{\beta}) + \sum_{r =
1}^R \lambda_{ir}f_{tr} + \sum_{r = R + 1}^{\min\{N,T\}} \lambda_{ir}f_{tr}
+ \varepsilon_{it}$. The additional term, $\sum_{r = R + 1}^{\min\{N,T\}}
\lambda_{ir}f_{tr}$ adds the following term to the proof of Proposition A.1
in \cite{Bai2009}, 
\begin{align*}
\frac{1}{\sqrt{T}} &\left\| \frac{1}{NT}\sum_{r=R+1}^{\min\{N,T\}}\sum_{r^%
\prime=R+1}^{\min\{N,T\}} \sum_i \lambda_{ir}\lambda_{ir^\prime} f_r
f_{r^\prime}^\prime \widehat{f}\, \right\|_F \leq \sqrt{R}\left\| \frac{1}{NT%
}\sum_{r=R+1}^{\min\{N,T\}} \sum_i \lambda_{ir}^2 f_r f_{r}^\prime \right\|_F
\\
&= \sqrt{R}\left\| \frac{1}{NT}\sum_{r=R+1}^{\min\{N,T\}} \sigma_r^2(\Gamma)
v_r v_r^\prime \right\|_F \leq \sqrt{R}\frac{1}{NT}\sum_{r=R+1}^{\min\{N,T%
\}} \sigma_r^2(\Gamma) \left\|v_r v_r^\prime\right\|_F = O_p\left(R^{1/2 -
2\rho}\right).
\end{align*}
The term $v_r$ is an orthonormal right singular vector. Square this for the
final result. There are also cross product terms between the approximation
error and $(\beta^0 - \widehat{\beta})$, which are also smaller order.

Consider now only the approximation error interacting with $\varepsilon$.
Take $\lambda_{ir} = u_{ir} {\sigma}_r$, $f_{tr} = v_{tr}$ for $r = 1,\dots,
R$. Define $H = (\lambda^{\prime }\lambda/N)(f^{\prime }\hat f/T)
\Sigma_R^{-2}$, where $\Sigma_R$ is the diagonal matrix of the $R$ first
singular values $\sigma_r$. \cite{bai2023approximate} show that, 
\begin{align*}
\frac{1}{{T}}\left\|\hat f - fH\right\|^2 = O_p(\min\{N,T\}^{-1})\cdot
\|\Sigma_R^{-2}\|.
\end{align*}
$\|\Sigma_R^{-2}\| = \sum_{r=1}^R \sigma_r^{-2} \lesssim R^{2\rho +1}$ from
Assumption~\ref{ass:singularvalues} which gives the final expression in the
lemma. 

\hfill
\end{proof}

\begin{proof}[Proof of Corollary~\protect\ref{cor:bai2009fe}]
By Lemma~\ref{lemma:frewei2023} $\|\widehat{\beta}_{LS} - \beta^0\|^4 =
\min\{N,T\}^\frac{2 - 2\rho}{\rho}$. Then $\min\{N,T\}^\frac{1}{3\rho}\|%
\widehat{\beta}_{LS} - \beta^0\|^4 = \min\{N,T\}^{\frac{1}{3\rho} + \frac{2
- 2\rho}{\rho}} = \min\{N,T\}^{\frac{7-6\rho }{3\rho}}$. For $\rho\geq7/3$
this is $o_p(\min\{N,T\}^{-1})$.

\hfill
\end{proof}

The fact factor loadings and factors are only identified up-to-rotations is
ignored here, since analysis can be performed after multiplying factor
loadings and factors by $HH^{-1}$, the rotation matrix from Lemma~\ref%
{lemma:bai2009fe}. As shown in the proof of Lemma~\ref{lemma:bai2009fe},
number of estimated $R$ is upper bounded in order to ensure $H$ is
asymptotically non-singular, hence invertible at all points in the sequence
for which $R\to\infty$. Hence, we analyse the problem without considering
these rotations.

\begin{proof}[Proof of Lemma~\protect\ref{lemma:withinEstimatorBias}]
For $(i,t)\in \mathcal{I}_{k_1,k_2}$, weights $W_{r,jj^{\prime }}^{(1)}$ and 
$W_{r,ss^{\prime }}^{(2)}$ are estimated from $(j,s) \in \mathcal{I}%
\backslash \mathcal{I}_{k_1,k_2}$, and importantly $(i,t)\in \mathcal{I}%
_{k_1,k_2}$ only has non-zero weights for other $(i,t)\in \mathcal{I}%
_{k_1,k_2}$. Hence, we study asymptotic variance by partition without any
stochastic dependence between weights and idiosyncratic terms. In the
following, consider one partition $(i,t)\in \mathcal{I}_{k_1,k_2}$.

For $\xi_{\Gamma X} := (NT)^{-1}\sum_{it}({\Gamma }_{X,it}-\widehat{\Gamma }%
_{X,it})\big({\Gamma }_{it}-\widehat{\Gamma }_{it}\big)$ we verify, 
\begin{align*}
\mathbb{E}[\xi_{\Gamma X} ] &= R_2^2\cdot O_p(h_\lambda^2 h_f^2
+\xi_\lambda^2\xi_f^2).
\end{align*}
Assume without loss that $dim(X_{it}) = 1$ to avoid heavy notation.\footnote{%
The following analysis can be done separately for each $k=1\dots,
dim(X_{it}) $, so it does not matter.} Expand $\Gamma_{X,it} -
\hat\Gamma_{X,it}$, 
\begin{align*}
\Gamma_{X} - \hat\Gamma_{X} = (\mathbb{I}_N - W^{(1)})\Gamma_{X}(\mathbb{I}%
_N - W^{(2)}) + W^{(1)\prime } \eta - \eta W^{(2)\,^{\prime }} + W^{(1)}\eta
W^{(2)\,^{\prime }}.
\end{align*}
Define $\tilde\Gamma_X := (\mathbb{I}_N - W_1)\Gamma_{X}(\mathbb{I}_N -
W_2)^{\prime }$, $\hat\eta := W^{(1)}\eta + \eta W^{(2)\,^{\prime }} -
W^{(1)}\eta W^{(2)\,^{\prime }}$, likewise for $\tilde\Gamma $ and $%
\hat\varepsilon$. 
\begin{align*}
\xi_{\Gamma X} = \frac{1}{NT} tr\left\{\tilde\Gamma_X^{\prime
}\tilde\Gamma\right\} - \frac{1}{NT} tr\left\{\tilde\Gamma_X^{\prime
}\hat\varepsilon\right\} - \frac{1}{NT} tr\left\{\hat\eta^{\prime
}\tilde\Gamma\right\} + \frac{1}{NT} tr\left\{\hat\eta^{\prime
}\hat\varepsilon\right\}.
\end{align*}

The sample split ensures $\mathbb{E}[W^{(1)}\eta] = \mathbb{E}[\eta
W^{(2)\,^{\prime }}] = \mathbb{E}[W^{(1)}\eta W^{(2)\,^{\prime }}] = 0$,
likewise for $\varepsilon$. Mean independence between $\{\varepsilon,\eta\}$
and $\{\Gamma_X,\Gamma\}$ thus ensures the last three terms are mean zero.
Hence we must show $\mathbb{E}\frac{1}{NT} tr\left\{\tilde\Gamma_X^{\prime
}\tilde\Gamma\right\} = R^2\cdot O_p(h_\lambda^2 h_f^2
+\xi_\lambda^2\xi_f^2) $. By Cauchy-Schwarz, 
\begin{align*}
\left|\mathbb{E}\left[\frac{1}{NT} tr\left\{\tilde\Gamma_X^{\prime
}\tilde\Gamma\right\}\right]\right| \leq \mathbb{E}\left[\frac{1}{NT}
\|\tilde\Gamma_X\|_F \|\tilde\Gamma\|_F\right].
\end{align*}

Here we bound the following, 
\begin{align*}
\frac{1}{\sqrt{NT}}\left\| \tilde\Gamma\right\|_F = \frac{1}{\sqrt{NT}}%
\left\|(\mathbb{I}_N - W^{(1)}) \Gamma(\mathbb{I}_T - W^{(2)})^\prime
\right\|_F = R_2^2\cdot O_p(h_\lambda h_f + \xi_\lambda\xi_f).
\end{align*}

For each $k = 1,\dots, R_0$ in $\tilde\Gamma= \left(\mathbb{I}_N - W^{(1)}
\right) \sum_{k = 1}^{R_0} h_k(\lambda_k, f_k) \left(\mathbb{I}_T - W^{(2)}
\right)$ take row and column wise product $\odot$, and mean value theorem,
for $\tilde\lambda_k$ entrywise between $\lambda_k$ and $\hat\lambda_k$,
likewise $\tilde f_k$: 
\begin{align}  \label{eqn: additive Taylor Expansion}
h_k(\lambda_k, f_k) = h_k(\hat\lambda_k, \hat f_k) + (\lambda_k -
\hat\lambda_k)\odot \nabla_\lambda^{(1)}h_k(\tilde\lambda_k, \tilde f_k) +
\nabla_f^{(1)}h_k(\tilde\lambda_k, \tilde f_k) \odot (f_k - \hat f_k)^\prime.
\end{align}
Two more mean value theorem expansions on the second and third terms, using $%
\check\lambda_k$ entrywise between $\hat\lambda_k$ and $\tilde\lambda_k$,
likewise again for $\check f_k$: 
\begin{align*}
(\lambda_k - \hat\lambda_k)\odot \nabla_\lambda^{(1)}h_k(\tilde\lambda_k,
\tilde f_k) &= (\lambda_k - \hat\lambda_k)\odot
\nabla_\lambda^{(1)}h_k(\tilde\lambda_k, \hat f_k) + (\lambda_k -
\hat\lambda_k)\odot \nabla_{\lambda f}^{(2)}h_k(\tilde\lambda_k, \check
f_k)\odot (\tilde f_k - \hat f_k)^\prime \\
\nabla_f^{(1)}h_k(\tilde\lambda_k, \tilde f_k) \odot (f_k - \hat f_k)^\prime
&= \nabla_f^{(1)}h_k(\hat\lambda_k, \tilde f_k) \odot (f_k - \hat
f_k)^\prime + (\tilde\lambda_k - \hat\lambda_k)\odot \nabla_{f \lambda
}^{(2)}h_k(\check\lambda_k, \tilde f_k)\odot (\tilde f_k - \hat f_k)^\prime.
\end{align*}

The last terms in these expansions can be bounded, 
\begin{align*}
\frac{1}{\sqrt{NT}}&\left\|(\mathbb{I}_N - W^{(1)}) (\lambda_k -
\hat\lambda_k)\odot \nabla_{\lambda f}^{(2)}h_k(\tilde\lambda_k, \check
f_k)\odot (\tilde f_k - \hat f_k)^\prime(\mathbb{I}_T - W^{(2)})^\prime
\right\|_F \\
&\leq \frac{1}{\sqrt{NT}} \left\| \mathbb{I}_N - W^{(1)}\right\|_2 \left\| 
\mathbb{I}_T - W^{(2)}\right\|_2 \left\| (\lambda_k - \hat\lambda_k)\odot
\nabla_{\lambda f}^{(2)}h_k(\tilde\lambda_k, \check f_k)\odot (\tilde f_k -
\hat f_k)^\prime\right\|_F.
\end{align*}
Terms $\left\| \mathbb{I}_N - W^{(1)}\right\|_2 = O(R_2)$, and $\left\| 
\mathbb{I}_T - W^{(2)}\right\|_2 = O(R_2)$ since additive weights are
absolutely summable. Term $\nabla_{\lambda f}^{(2)}h_k(\tilde\lambda_k,
\check f_k) = \nabla_{\lambda f}^{(2)}h_k(\lambda_k, f_k) + o_p(1)$ by
consistency of $\hat\lambda_k$ and $\hat f_k$, and continuity since $h_k\in
H_f^p(\Omega_\lambda, \Omega_f)$ with $p\geq 3$. Also, by Markov inequality
and $h_k\in H_f^p(\Omega_\lambda, \Omega_f)$ with $p\geq 3$, 
\begin{align*}
Pr\left[\nabla_{\lambda f}^{(2)}h_k( \lambda_k, f_k)^2> M \right] \leq 
\mathbb{E}\left[ \nabla_{\lambda f}^{(2)}h_k( \lambda_k, f_k)^2 \right]/M
=O(1/M).
\end{align*}
Hence, the last term is, noting that $\tilde f_k$ is between $\hat f_k$ and $%
f_k$, hence also consistent for $f_k$, 
\begin{align*}
\frac{1}{\sqrt{NT}}\left\| (\lambda_k - \hat\lambda_k)\odot \nabla_{\lambda
f}^{(2)}h_k(\tilde\lambda_k, \check f_k)\odot (\tilde f_k - \hat
f_k)^\prime\right\|_F &= \frac{O_p(1)}{\sqrt{NT}}\left\| \lambda_k -
\hat\lambda_k\right\|\left\| \tilde f_k - \hat f_k\right\| + o_p(\xi_\lambda
\xi_f) \\
&= O_p(\xi_\lambda \xi_f) + o_p(\xi_\lambda \xi_f)
\end{align*}

Term $(\lambda_k - \hat\lambda_k)\odot
\nabla_\lambda^{(1)}h_k(\tilde\lambda_k, \hat f_k)$ can be bound using \cite%
{opsomer2000asymptotic} over $\hat f_k$ argument, 
\begin{align*}
\frac{1}{\sqrt{NT}}&\left\|(\mathbb{I}_N - W^{(1)}) (\lambda_k -
\hat\lambda_k)\odot \nabla_\lambda^{(1)}h_k(\tilde\lambda_k, \hat f_k)(%
\mathbb{I}_T - W^{(2)})^\prime \right\|_F \\
&\leq \frac{1}{\sqrt{NT}} \left\| \mathbb{I}_N - W^{(1)}\right\|_2
\left\|(\lambda_k - \hat\lambda_k)\odot
\nabla_\lambda^{(1)}h_k(\tilde\lambda_k, \hat f_k)(\mathbb{I}_T -
W^{(2)})^\prime \right\|_F \\
&= \frac{O_p(R_2)}{\sqrt{NT}} \left\|(\lambda_k - \hat\lambda_k)\odot
\nabla_{\lambda f}^{(2)}h_k(\tilde\lambda_k, \hat f_k)\cdot h_f\right\|_F
\end{align*}
Then by similar arguments as above, this is $R_2\cdot O_p(\xi_\lambda h_f)$.

Last term to consider is $h_k(\hat\lambda_k, \hat f_k) $. Additive model
arguments in \cite{opsomer2000asymptotic} give: 
\begin{align*}
\sum_{k=1}^{R_0}\left(\mathbb{I}_N - W^{(1)} \right)h_k(\hat\lambda_k, \hat
f_k) \left(\mathbb{I}_T - W^{(2)} \right)^\prime &=
\sum_{k=1}^{R_0}O_p(h_\lambda ) \cdot \nabla^{(1)}_\lambda
h_k(\tilde\lambda_k \hat f_k)\left(\mathbb{I}_T - W^{(2)} \right)^\prime \\
&= \sum_{k=1}^{R_0}O_p(h_\lambda)\cdot \nabla^{(2)}_{\lambda f}
h_k(\tilde\lambda_k \tilde f_k)\cdot O_p(h_f)
\end{align*}
Thus, again using consistency and Hilbert properties to bound $%
\nabla^{(2)}_{\lambda f} h_k(\tilde\lambda_k \tilde f_k)$, 
\begin{align*}
\frac{1}{\sqrt{NT}}\left\| \sum_{k=1}^{R_0}\left(\mathbb{I}_N - W^{(1)}
\right)h_k(\hat\lambda_k, \hat f_k) \left(\mathbb{I}_T - W^{(2)}
\right)^\prime\right\|_f \leq R_0\cdot R_2^2\cdot O_p(h_\lambda h_f)
\end{align*}

Combine these results for the final statement of the result.

\hfill
\end{proof}

\begin{proof}[Proof of Lemma~\protect\ref{lemma:withinEstimatorVar}]
Here we verify, 
\begin{align*}
Var[\xi_{\Gamma X} ] &= \frac{R_2^6}{NT}\cdot O_p(\xi_\lambda^2\xi_f^2) + 
\frac{R_2^8}{NT}\cdot O_p(h_\lambda^2h_f^2 + (NT h_\lambda h_f)^{-1}).
\end{align*}

Consider $\frac{1}{NT}tr\{\hat\eta^{\prime }\hat\varepsilon\}$. Use $e:=
vec(\varepsilon)$ and $x := vec(\eta)$. $vec (\hat\eta)$ can be written $(%
\mathbb{I} \otimes W^{(1)})x + (W^{(2)} \otimes \mathbb{I})x -
(W^{(1)}\otimes W^{(2)}) x$. Then, 
\begin{align*}
\frac{1}{NT}tr\{\hat\eta^{\prime }\hat\varepsilon\} &= \frac{1}{NT}x^{\prime
}(\mathbb{I} \otimes W^{(1)})^{\prime }(\mathbb{I} \otimes W^{(1)})e + \frac{%
1}{NT}x^\prime (W^{(2)} \otimes \mathbb{I})^{\prime }(\mathbb{I} \otimes
W^{(1)})e \\
&- \frac{1}{NT}x^{\prime}(W^{(1)}\otimes W^{(2)})^{\prime }(\mathbb{I}
\otimes W^{(1)})e + \frac{1}{NT}x^{\prime }(\mathbb{I} \otimes
W^{(1)})^{\prime}(W^{(1)}\otimes\mathbb{I} )e \\
& + \frac{1}{NT}x^{\prime}(W^{(2)}\otimes\mathbb{I} )^{\prime}(W^{(2)}\otimes%
\mathbb{I} )e- \frac{1}{NT}x^{\prime }(W^{(1)}\otimes W^{(2)})^{\prime
}(W^{(2)}\otimes\mathbb{I} )e \\
&+ \frac{1}{NT}x^{\prime }(\mathbb{I} \otimes W^{(1)})^{\prime}(W^{(1)}
\otimes W^{(2)})e + \frac{1}{NT}x^{\prime}(W^{(2)}\otimes\mathbb{I}
)^{\prime}(W^{(1)}\otimes W^{(2)})e \\
& - \frac{1}{NT}x^{\prime }(W^{(1)}\otimes W^{(2)})^{\prime}(W^{(1)}\otimes
W^{(2)} )e
\end{align*}
All terms are mean zero. Variance of the first term is bounded in norm by, 
\begin{align*}
\frac{1}{(NT)^2}\mathbb{E}&\left[ \left\|x^{\prime }(\mathbb{I} \otimes
W^{(1)})^{\prime }(\mathbb{I} \otimes W^{(1)})e e^\prime (\mathbb{I} \otimes
W^{(1)})^{\prime }(\mathbb{I} \otimes W^{(1)})x\right\|_F\right] \\
&\leq \frac{1}{(NT)^2}\mathbb{E}\left[\|x^{\prime}( W^{(1)}\otimes \mathbb{I}%
)^\prime (W^{(1) }\otimes \mathbb{I} )\|_F\|ee^{\prime }(W^{(1)}\otimes 
\mathbb{I})^\prime(W^{(1) }\otimes \mathbb{I} )x\|_F \right] \\
&\leq \frac{1}{(NT)^2}\mathbb{E}\left[\| W^{(1)\,^{\prime }}W^{(1)}\eta
\|_F^2\| ee^{\prime }\|_2\right]
\end{align*}
We use $(W^{(1)}\otimes \mathbb{I})^{\prime }(W^{(1)} \otimes \mathbb{I} )x
= vec( W^{(1)\,^{\prime }}W^{(1)}\eta )$. Term $W^{(1)}\eta = R_2\cdot
O_p((Nh_\lambda)^{-1/2})$ by mean zero $\eta$. From approximate symmetry in
weights, see \cite{opsomer1999root}, $W^{(1)\,^{\prime }}W^{(1)}\eta =
R_2^2\cdot O_p((Nh_\lambda)^{-1/2})$, since $W^{(1)} =
\sum_{\ell=1}^{R_2}W^{(1)}_\ell$ sums over rows, and columns, are $O(R_2)$.
Hence, 
\begin{align*}
\frac{1}{(NT)^2}\| W^{(1)\,^{\prime }}W^{(1)}\eta \|_F^2\| ee^{\prime }\|_2
= \frac{R_2^4}{NT}\cdot O_p((Nh_\lambda)^{-1})
\end{align*}
Similarly, $\frac{1}{(NT)^2}\| W^{(1)\,^{\prime }}W^{(1)}\eta
W^{(2)\,^{\prime }}W^{(2)} \|_F^2\|\| ee^{\prime }\|_2 = \frac{R_2^8}{NT}%
\cdot O_p((NT h_\lambda h_f)^{-1})$. All other terms in the variance of $%
\frac{1}{NT}tr\{\hat\eta^{\prime }\hat\varepsilon\} $ are similarly bounded
such that, 
\begin{align*}
Var\left[\frac{1}{NT}tr\{\hat\eta^{\prime }\hat\varepsilon\}|\eta\right]
\leq \frac{R^4}{NT}\cdot O_p((Nh_\lambda)^{-1}+ (Th_f)^{-1}) + \frac{R^8}{NT}%
\cdot O_p((NT h_\lambda h_f)^{-1}) .
\end{align*}

The two remaining terms in $\xi_{\Gamma X} $ are bound now. Use $%
\hat\varepsilon = W^{(1)}\varepsilon + W^{(1)}\varepsilon W^{(2)} -
W^{(1)}\varepsilon W^{(2)}$, 
\begin{align*}
\frac{1}{(NT)^2}Var\left[ tr\left\{\tilde\Gamma_X^{\prime}
W^{(1)}\varepsilon\right\} |\eta \right] &= \frac{1}{(NT)^2} vec\left[W^{
(1)\,^{\prime }}\tilde\Gamma_X \right]^{\prime }\mathbb{E}\left[ee^{\prime
}|\eta \right] vec\left[W^{ (1)\,^{\prime }}\tilde\Gamma_X \right] \\
&\leq \frac{1}{(NT)^2} \|\tilde\Gamma_X\|_F^2 \|W^{(1)}\|_2^2 \left\| 
\mathbb{E}\left[ee^{\prime }|\eta \right]\right\|_2
\end{align*}
From Ostrowski theorem $\|W^{(1)}\|_2 = O(R_2)$ since each row and column of 
$W_\ell^{(1)}$ in $W^{(1)} = \sum_{\ell =1}^{R_2}W_\ell^{(1)}$ sums are
bounded. By Assumption~\ref{ass:boundedNorms} $\|\mathbb{E}\left[ee^{\prime
}|\eta \right]\|_2 = O_p(1)$. Hence the term is bounded, 
\begin{align*}
\frac{1}{(NT)^2}Var\left[ tr\left\{\tilde\Gamma_X^{\prime}
W^{(1)}\varepsilon\right\} |\eta \right] = \frac{R_2^2}{(NT)^2}
\|\tilde\Gamma_X\|_F^2
\end{align*}
By similar arguments, $\frac{1}{(NT)^2}Var\left[tr\left\{W^{(2)\,^{\prime
}}\tilde\Gamma_X^{\prime} W^{(1)}\varepsilon\right\}|\eta \right] = \frac{%
R_2^4}{(NT)^2} \|\tilde\Gamma_X\|_F^2$, the bounding factor for variance
terms in $\frac{1}{NT}tr\{\tilde\Gamma_X\hat\varepsilon\}$. Variance from $%
\frac{1}{NT}tr\{\tilde\Gamma\hat\eta\}$ is bound similarly.

Hence, variance of $\xi_{\Gamma X} $ is bounded, 
\begin{align*}
Var(\xi_{\Gamma X} ) &= \frac{R_2^6}{NT}\cdot O_p(\xi_\lambda^2\xi_f^2) + 
\frac{R_2^4}{NT}\cdot O_p((Nh_\lambda)^{-1}+ (Th_f)^{-1}) + \frac{R_2^8}{NT}%
\cdot O_p(h_\lambda^2h_f^2 + (NT h_\lambda h_f)^{-1}) .
\end{align*}
\hfill
\end{proof}

\begin{proof}[Proof of Corollary~\protect\ref{cor:Final}]
We must bound ${\hat\Omega_X}^{-1}$. Let $vec_K\in \mathbb{R}^{NT\times K}$
vectorise $X$ over its third dimension. Denote $\times_n$ the multilinear
product over dimension $n$, e.g. for $X\in \mathbb{R}^{N\times T\times K}$,
and $M\in \mathbb{R}^{N\times N}$, then $M\times_1 X = \sum_i M_{\cdot
i}X_{i\cdot\cdot}\in \mathbb{R}^{N\times T\times K}$, and for $\tilde M \in 
\mathbb{R}^{T\times T}$, $(\tilde M\times_2 X)_{it} = \sum_sM_{t s}\tilde
X_{i s \cdot}\in \mathbb{R}^{K}$. 
\begin{align*}
\hat\Omega_X &= \frac{1}{NT} \left\{ vec_K[ (\mathbb{I} - W_1)\times_1 (%
\mathbb{I} - W_2)^{\prime }\times_2 (\Gamma_X + \eta) ]^{\prime }vec_K[ (%
\mathbb{I} - W_1)\times_1 (\mathbb{I} - W_2)^{\prime }\times_2 (\Gamma_X +
\eta) ] \right\} \\
&= \frac{1}{NT} \left\{ vec_K[ (\mathbb{I} - W_1)\times_1 (\mathbb{I} -
W_2)^{\prime }\times_2 \eta ]^{\prime }vec_K[ (\mathbb{I} - W_1)\times_1 (%
\mathbb{I} - W_2)^{\prime }\times_2 \eta ] \right\} + \dots.
\end{align*}
Denote this first term $\hat\Omega_{\eta}\in \mathbb{R}^{K\times K}$. We
need to show $\|\hat\Omega_X^{-1}\|_2 < \infty$ wpa1.. For the $k,\ell$
entry, 
\begin{align*}
\hat\Omega_{\eta, k\ell} = \frac{1}{NT}tr \left\{(\mathbb{I} - W_2)^{\prime
}\eta_k^{\prime }(\mathbb{I} - W_1)(\mathbb{I} - W_1)^{\prime }\eta_\ell (%
\mathbb{I} - W_2)\right\}
\end{align*}
The term $(\mathbb{I} - W_1)^{\prime }\eta_k (\mathbb{I} - W_2) $ simplifies  
to, 
\begin{align*}
(\mathbb{I} - W_1)^{\prime }\eta_k (\mathbb{I} - W_2) &= \eta_k +
O_p(h_\lambda^2 + (Nh_\lambda)^{-1/2} + h_f^2 + (Th_f)^{-1/2}).
\end{align*}
From $h_\lambda = cN^{-\tau}$, $h_f = cT^{-\tau}$ with $\tau\in(1/4, 1)$, 
\cite{opsomer1999root} Theorem 2 implies, 
\begin{align*}
\hat\Omega_{\eta}^{-1} = \left[\frac{1}{NT}vec_K[\eta]^{\prime }vec_K[\eta]%
\right]^{-1} + o_p(NT)^{-1/2}
\end{align*}
Thus, $\hat\Omega_{\eta}^{-1} = \mathbb{E}[\eta_{it}^{\prime
}\eta_{it}]^{-1} + O_p(NT)^{-1/2} + o_p(NT)^{-1/2}$, by cofactor expansion,
see \cite{opsomer1999root}. Other terms in $\hat\Omega_X$ are $O_p(\xi_X)$
by arguments in Section~\ref{sect:NO}. 
\begin{align*}
\hat\Omega_X^{-1} &= \hat\Omega_{\eta}^{-1} + O_p(\xi_X) = \mathbb{E}%
[\eta_{it}^{\prime }\eta_{it}]^{-1} + O_p(NT)^{-1/2} + O_p(\xi_X).
\end{align*}
Above we show $\xi_X = R_2^2\cdot O_p(h_\lambda h_f + \xi_\lambda \xi_f) =
o_p(1) $, hence $\hat\Omega_X^{-1} = \mathbb{E}[\eta_{it}^{\prime
}\eta_{it}]^{-1} + o_p(1)$. This implies singular values $\hat\Omega_\eta $
are asymptotically lower bounded by singular values of $\mathbb{E}%
\eta_{it}^{\prime }\eta_{it}$ from Weyl's theorem, which are all non-zero by
Assumption~\ref{ass:ID}, such that $\|\hat\Omega_X^{-1}\|_2 = O_p(1)$.

The variance of $\hat{\beta}_{NO}^{SS}$ can thus be characterised, 
\begin{align*}
Var[\hat{\beta}_{NO}^{SS}|\eta] &= \mathbb{E}[\eta_{it}^{\prime
}\eta_{it}]^{-1}\frac{1}{NT} \left\{vec_K(\eta)^{\prime }\mathbb{E}%
[ee^{\prime }|\eta] vec_K(\eta)\right\} \mathbb{E}[\eta_{it}^{\prime
}\eta_{it}]^{-1} \\
&+ \frac{R_2^6}{NT}\cdot O_p(\xi_\lambda^2\xi_f^2) + \frac{R_2^8}{NT}\cdot
O_p(h_\lambda^2h_f^2 + (NT h_\lambda h_f)^{-1})
\end{align*}

\hfill
\end{proof}

\begin{proof}[Proof of Lemma~\protect\ref%
{lemma:withinEstimatorBiasIndependent}]
When proxies are independent proof of Lemma~\ref{lemma:withinEstimatorBias}
simplifies as follows. 
\begin{align*}
\left(\mathbb{I}_N - W^{(1)} \right) h_k(\lambda_k, f_k) \left(\mathbb{I}_T
- W^{(2)} \right)^\prime = \left(\mathbb{I}_N - W_k ^{(1)} - \sum_{k^\prime
\neq k} W_{k^\prime}^{(1)} \right) h_k(\lambda_k, f_k) \left(\mathbb{I}_T -
W_k ^{(2)} - \sum_{k^\prime \neq k} W_{k^\prime}^{(2)} \right)^\prime
\end{align*}
This expands to four terms, $A_1 - A_2 - A_3 + A_4$: 
\begin{align*}
A_1 := \left(\mathbb{I}_N - W_k ^{(1)}\right) h_k(\lambda_k, f_k) \left(%
\mathbb{I}_T - W_k ^{(2)}\right)^\prime & & A_2 := \left(\mathbb{I}_N - W_k
^{(1)}\right) h_k(\lambda_k, f_k) \sum_{k^\prime \neq k}
W_{k^\prime}^{(2)\,\prime} \\
A_3 := \sum_{k^\prime \neq k} W_{k^\prime}^{(1)}h_k(\lambda_k, f_k) \left(%
\mathbb{I}_T - W_k ^{(2)}\right)^\prime & & A_4 := \sum_{k^\prime \neq k}
W_{k^\prime}^{(1)}h_k(\lambda_k, f_k) \sum_{\ell^\prime \neq k}
W_{\ell^\prime}^{(2)\,\prime}
\end{align*}

Term $(NT)^{-1/2}\|A_4\|_F$ can be bound using that $h_k(\lambda_k, f_k)$
and smoothers are mean zero, 
\begin{align*}
\sum_{k^\prime \neq k} \sum_{\ell^\prime \neq k} \sum_{js} W_{ij,
k^\prime}^{(1)} W_{ts, k^\prime}^{(2)}h_k(\lambda_{jk}, f_{sk}) =
\sum_{k^\prime \neq k} \sum_{\ell^\prime \neq k} O_p(NT)^{-1/2} = (R_2 -1)^2
O_p(NT)^{-1/2}
\end{align*}
where $R_2\geq R_0$ are estimated eigenfunctions. Hence, $%
(NT)^{-1/2}\|A_4\|_F = O_p(R_2^2\cdot (NT)^{-1/2})$.

Next to bound $(NT)^{-1/2}\|A_1\|$. First note that by \cite%
{opsomer2000asymptotic}, additive regression weights are $W^{(1)}_k =
S_k^{(1)} + O_p(\iota_N \iota_N^{\prime }/N)$, where recall that $S_k^{(1)}$
are the smoothers over $\lambda_k$, respectively $S_k^{(2)}$ over $f_k$.
Hence, $A_1$ can be further expanded as $A_1 = A_{11} - A_{12} - A_{13}+
A_{14}$, 
\begin{align*}
A_{11} - A_{12} - A_{13}+ A_{14} &= \left(\mathbb{I}_N - S_k ^{(1)}\right)
h_k(\lambda_k, f_k) \left(\mathbb{I}_T - S_k ^{(2)}\right) - O_p(\iota_N
\iota_N^{\prime }/N) h_k(\lambda_k, f_k) \left(\mathbb{I}_T - S_k
^{(2)}\right) \\
&- \left(\mathbb{I}_N - S_k ^{(1)}\right) h_k(\lambda_k, f_k) O_p(\iota_T
\iota_T^{\prime }/T) + O_p(\iota_N \iota_N^{\prime }/N) h_k(\lambda_k, f_k)
O_p(\iota_T \iota_T^{\prime }/T)
\end{align*}
By mean zero $h_k(\lambda_k, f_k)$, term $(NT)^{-1/2}\|A_{14}\| = 0$.

Terms $A_{12}$ and $A_{13}$ can be bounded smaller order than $A_{11}$, so
we focus on this term. Using the Taylor expansions, entries from the term $%
A_{11}$ can be expanded as, 
\begin{align*}
h_{it} - \sum_j S_{ij, k}^{(1)} h_{jt} - \sum_s S_{ts, k}^{(2)} h_{is} +
\sum_{js} S_{ij, k}^{(1)} S_{ts, k}^{(2)} h_{js} = \sum_{js} S_{ij, k}^{(1)}
S_{ts, k}^{(2)} [h_{it} - h_{jt}- h_{is} + h_{js}].
\end{align*}
From the Taylor expansions the term $h_{it} - h_{jt}- h_{is} + h_{js}$
simplifies to, 
\begin{align*}
h_{it} - h_{jt}- h_{is} + h_{js} = \frac{1}{2}(\lambda_{ik} -
\lambda_{jk})\cdot (f_{tk} - f_{sk}) \nabla_{\lambda f}^{(2)}h_{it} +
o(|\lambda_{ik} - \lambda_{jk}|\cdot |f_{tk} - f_{sk}|).
\end{align*}
Expand $(\lambda_i - \lambda_j)$, respectively $(f_t - f_s)$ around their
estimated quantities 
\begin{align*}
(\lambda_i - \lambda_j) = (\lambda_i - \hat\lambda_i) + (\lambda_j -
\hat\lambda_j) + (\hat\lambda_i - \hat\lambda_j)
\end{align*}
and take the squared norm term, 
\begin{align*}
\frac{1}{NT} \sum_{it} \left( \sum_{js} S_{ij, k}^{(1)} S_{ts, k}^{(2)} [
h_{it} - h_{jt}- h_{is} + h_{js} ]\right)^2
\end{align*}
Thus, by similar arguments as the proof of Lemma~\ref%
{lemma:withinEstimatorBias}, and standard nonparametric arguments, 
\begin{align*}
\frac{1}{\sqrt{NT}}\|A_{11}\| = O_p(\xi_\lambda \xi_f) + O_p(h_\lambda h_f)
+ o_p(h_\lambda h_f).
\end{align*}
Terms $(NT)^{-1/2}\|A_{2}\|$ and $(NT)^{-1/2}\|A_{3}\|$ are bound as $%
R_2\cdot O_p(h_\lambda (T)^{-1/2})$ and $R_2\cdot O_p(h_f (N)^{-1/2})$
respectively by similar arguments.

Define $A_n^{(k)}$ as the $A_n$ objects above for each $k=1,\dots,R_0$.
Then, 
\begin{align*}
\frac{1}{\sqrt{NT}}\|\tilde\Gamma\|_F &\leq \frac{1}{\sqrt{NT}}
\sum_{k=1}^{R_0} \left[\|A_1^{(k)}\| + \|A_2^{(k)}\| + \|A_3^{(k)}\| +
\|A_4^{(k)}\|\right] \\
&= R_0\cdot \left\{O_p(\xi_\lambda \xi_f) + O_p(h_\lambda h_f) +
R_2^2O_p(NT)^{-1/2} \right\}
\end{align*}
This is $o_p(NT)^{-1/4}$ for $h_\lambda = c\cdot N^{-\rho}$ and $h_f =
c\cdot T^{-\rho}$ with $\rho \in (1/4,1)$. \hfill
\end{proof}

\section{Additional Simulations}

\label{sect:AppSims}

Here we study more simulated data generating processes, and further compare
our advocated additive eigenfunction approach to other proxies for $\alpha
_{i}$ and $\gamma _{t}$ standard in the literature. In particular, we
implement our nonparametric smoothing estimator with our proposed fixed
effects estimators $\hat{\lambda}_{i}$ and $\hat{f}_{t}$ replaced by the 
\cite{zhang2017estimating} pseudo distance in \eqref{eqn:ZhangDistance} over 
$Y$ and $X$, which we call \textquotedblleft Pseudo\textquotedblright . We
also implement our estimator with our $\hat{\lambda}_{i}$ and $\hat{f}_{t}$
replaced by the cross-sectional and time-serial averages of $Y$ and $X$ as
proxies employed by \cite{beyhum2025inference} , which we call
\textquotedblleft Moment\textquotedblright . We find that our nonparametric
smoothing estimator is robust to higher dimensional fixed-effects and more
complicated functions under our additive eigenfunction approach, but not
under the alternatives Pseudo and Moment. However, we make no claim to the
relative performance of our estimator to those of \cite{deaner2025inferring,
beyhum2025inference}, since they use these distance metrics in different
ways.

The Moment proxies are defined as follows: 
\begin{align*}
\tilde{\lambda}_{i,Y}=\frac{1}{T}\sum_{t}Y_{it},&  & \tilde{\lambda}_{i,X}=%
\frac{1}{T}\sum_{t}X_{it},&  & \tilde{\lambda}& =(\tilde{\lambda}_{Y},\tilde{%
\lambda}_{X})^{\prime }, \\
\tilde{f}_{t,Y}=\frac{1}{N}\sum_{i}Y_{it},&  & \tilde{f}_{t,X}=\frac{1}{N}%
\sum_{i}X_{it},&  & \tilde{f}& =(\tilde{f}_{Y},\tilde{f}_{X})^{\prime }.
\end{align*}%
Variables $\{Y,X\}$ are both residualised according to 
\eqref{eq: two-way
kernel regression} over $\tilde{\lambda}$ and $\tilde{f}$. The
Pseudo--distance proxies w.r.t. $Y$ are defined in \eqref{eqn:ZhangDistance}%
, and we construct similar ones using $X$, denoted  $\hat{d}_{X}^{(1)}$, $%
\hat{d}_{X}^{(2)}$. We then set $\tilde{\lambda}=(\hat{d}_{Y}^{\left(
1\right) },\hat{d}_{X}^{\left( 1\right) })^{\prime }$ and $\tilde{f}=(\hat{d}%
_{Y}^{\left( 2\right) },\hat{d}_{X}^{\left( 2\right) })^{\prime }$. The
Pseudo proxies do not update, hence no backfitting is required.

Consider the data generating process: 
\begin{align*}
Y_{it}=X_{it}\beta +g(\alpha _{i},\gamma _{t})+\varepsilon _{it},
&&
X_{it}=g_{X}(\alpha _{i},\gamma _{t})+\eta _{it}.
\end{align*}%
where $\varepsilon _{it},\eta _{it}$ are $i.i.d.N(0,1)$ for purposes of
Section~\ref{sect:AppSims} simulations.\footnote{%
In the main text, simulations in Section~\ref{sect:Simulations} allows
correlations and heteroskedasticity in these noise terms. However, for
cleaner comparisons to other methods, we use i.i.d. assumptions for this
section's set of simulations.}

Output is reported in Tables~\ref{tbl:AppSimulationD=2} and~\ref%
{tbl:AppSimulationD=3} for the following functions $g$ and $g_X$: 
\begin{align}  \label{eqn:simsApp}
g(\alpha_i,\gamma_t) = \frac{1}{\theta\sqrt{2\pi}}\sum_{\ell = 1}^d
\exp\left(- \frac{(\alpha_{i\ell} - \gamma_{t\ell})^2}{\theta^2}%
\right), \quad\quad g_X(\alpha_i,\gamma_t) = \sum_{\ell = 1}^d\frac{1}{%
(|\alpha_{i\ell} - \gamma_{t\ell}| + 1)^{\theta}}.
\end{align}
Variation from these functions is then normalised to variance 4.
Heterogeneity $\alpha_i,\gamma_t\in \mathbb{R}^d$. For \eqref{eqn:simsApp}
we consider $d =2$ and $d = 3$. Smoothness parameter, $\theta = 1/2$ when $d
= 2$, and $\theta = 1$ when $d = 3$. This ensures numerical stability in the
DGP.\footnote{%
Higher $\theta$ implies numerically smoother functions, and must increase as 
$d$ increases to sufficiently bound variation in the function for small
perturbations around $\alpha$, $\gamma$. This is predicted by our theory,
cf. $\rho$ in Lemmas~\ref{lemma:frewei2023} and ~\ref{lemma:bai2009fe}.}

Tables~\ref{tbl:AppSimulationD=1Simple},~\ref{tbl:AppSimulationD=2Simple}
and~\ref{tbl:AppSimulationD=3Simple} report simulations for $%
\alpha_i,\gamma_t\in \mathbb{R}^d$ with $d\in\{1,2,3\}$ for functions
replicated as multivariate extensions to \cite{beyhum2025inference}
simulations: 
\begin{align}  \label{eqn:simsAppSimple}
g(\alpha_i,\gamma_t) = \|\alpha_i\|^2 + \alpha_i^\prime \gamma_t +
\sum_{\ell=1}^dsin(\alpha_{i\ell} \gamma_{t\ell} ), & & g_X(\alpha_i,%
\gamma_t) = \|\gamma_t\|^2 + \alpha_i^\prime \gamma_t +
\sum_{\ell=1}^dsin(\alpha_{i\ell} \gamma_{t\ell} )
\end{align}
In Table~\ref{tbl:AppSimulationD=1SimpleRE} we generate $\tilde g_X(a,b) =
g_X(a,b) + a_{x,i} + b_{x,t}$ with random effects $a_{x,i},b_{x,t}$. All
simulations for DGP~\eqref{eqn:simsAppSimple} use $\alpha_i,\gamma_t \sim
U(0,1)$ such that $\alpha_i,\gamma_t$ have positive mean.

Hyperparameters are set as follows. For the factor model estimators, $R_1 =
2\cdot\min\{N,T\}^{1/3}$, and for the nonparametric weighted-within
estimator these are set to $R_2 = 2\cdot\min\{N,T\}^{1/5}$.\footnote{%
More factors are required with higher dimension fixed-effects, $d_\alpha,
d_\gamma$, which is predicted by our theory.} Bandwidths for the
nonparametric weighted-within estimators are set to $h =
(25/\min\{N,T\})^{1/2}$. For the Moment approach, $h =
(10/\min\{N,T\})^{1/2} $, and the psuedo distance estimator uses $h =
(25/\min\{N,T\})^{1/4}$.\footnote{%
The pseudo distance estimator is very sensitive to bandwidth selection, and
can produce estimates with increasing standard errors. This is likely due to
the greedy nature for that pseudo distance.} All nonparametric
weighted-within estimators use the Gaussian kernel.

Tables~\ref{tbl:AppSimulationD=2} and~\ref{tbl:AppSimulationD=3} show
consistency under the Psuedo metric, but poor coverage since the bias is
still present, and large here. As the dimension, $d$, of $%
\alpha_i,\gamma_t\in \mathbb{R}^d$ increases, this bias becomes worse across
all estimators. However, the Factor, and Weighted-Within estimators still
perform well, with nominal coverage for our Weighted-Within estimator.

Tables~\ref{tbl:AppSimulationD=1Simple},~\ref{tbl:AppSimulationD=2Simple}
and~\ref{tbl:AppSimulationD=3Simple} for the simpler DGP in %
\eqref{eqn:simsAppSimple} show that whilst the moment based distance metric
performs well for scalar valued $\alpha_i,\gamma_t$, the estimator performs
poorly when $d=2,3$. Hence, the moment based distance, under simple DGPs
that admit injectivity over the effect space works for scalar valued
effects, but evidently struggles under even moderate dimensionality in $%
\alpha_i,\gamma_t$. The Weighted-Within estimator performs well regardless
of dimensions considered. The factor model also performs well, albeit with
under coverage likely only from standard errors being too small.

Table~\ref{tbl:AppSimulationD=1SimpleRE} shows a concerning aspect of the
moment distance metric. Under scalar effects and DGP %
\eqref{eqn:simsAppSimple}, the inclusion of even very simple random effects
breaks the injective moments condition seemingly necessary for this distant
metric to be useful.

\begin{table}[tbp]
\centering
\begin{tabular}{r|rrrr|}
& \multicolumn{4}{c}{Bias} \\ \hline
Dim. & Factor & WW & Pseudo & Moment \\ \hline
50 & 0.0484 & 0.0645 & 0.2467 & 0.7182 \\ 
100 & 0.0228 & 0.0207 & 0.1002 & 0.7255 \\ 
150 & 0.0159 & 0.0125 & 0.0608 & 0.7310 \\ 
200 & 0.0129 & 0.0092 & 0.0427 & 0.7347 \\ 
250 & 0.0118 & 0.0074 & 0.0327 & 0.7356 \\ \hline
\end{tabular}
\begin{tabular}{|rrrr|}
\multicolumn{4}{c}{Coverage} \\ \hline
Factor & WW & Pseudo & Moment \\ \hline
0.44 & 0.92 & 0.00 & 0.00 \\ 
0.45 & 0.97 & 0.00 & 0.00 \\ 
0.38 & 0.96 & 0.00 & 0.00 \\ 
0.30 & 0.96 & 0.00 & 0.00 \\ 
0.18 & 0.97 & 0.00 & 0.00 \\ \hline
\end{tabular}%
\caption{$d = 2$ DGP~\eqref{eqn:simsApp} Bias and coverage (5\% nominal
test) for 1,000 Monte Carlo rounds}
\label{tbl:AppSimulationD=2}
\centering
\begin{tabular}{r|rrrr}
& \multicolumn{4}{c}{Bias} \\ \hline
Dim. & Factor & WW & Pseudo & Moment \\ \hline
50 & 0.0454 & 0.0836 & 0.1854 & 0.6827 \\ 
100 & 0.0261 & 0.0276 & 0.0799 & 0.6938 \\ 
150 & 0.0239 & 0.0158 & 0.0506 & 0.6983 \\ 
200 & 0.0227 & 0.0109 & 0.0377 & 0.7024 \\ 
250 & 0.0220 & 0.0087 & 0.0304 & 0.7038 \\ \hline
\end{tabular}
\begin{tabular}{|rrrr|}
\multicolumn{4}{c}{Coverage} \\ \hline
Factor & WW & Pseudo & Moment \\ \hline
0.41 & 0.85 & 0.00 & 0.00 \\ 
0.26 & 0.93 & 0.00 & 0.00 \\ 
0.04 & 0.94 & 0.01 & 0.00 \\ 
0.00 & 0.95 & 0.01 & 0.00 \\ 
0.00 & 0.95 & 0.01 & 0.00 \\ \hline
\end{tabular}%
\par
{\footnotesize {Dim. is the size of each dimension, with $N=T$. Factor is a
factor model approximation. WW is our Weighted Within estimator. Pseudo
utilises the Weighted Within estimator with pseudo-distance in \cite%
{zhang2017estimating}. Moment is the Weighted Within estimator with
cross-sectional and time-serial means used to form distance metrics.} }
\caption{$d = 3$ DGP~\eqref{eqn:simsApp} Bias and coverage (5\% nominal
test) for 1,000 Monte Carlo rounds}
\label{tbl:AppSimulationD=3}
\end{table}

\begin{table}[tbp]\centering
\begin{tabular}{r|rrrr|}
& \multicolumn{4}{c}{Bias} \\ \hline
Dim. & Factor & WW & Pseudo & Moment \\ \hline
50 & 0.0055 & 0.0154 & 0.0122 & -0.0002 \\ 
100 & 0.0023 & 0.0065 & 0.0203 & 0.0004 \\ 
150 & 0.0008 & 0.0028 & 0.0334 & 0.0001 \\ 
200 & 0.0002 & 0.0017 & 0.0365 & -0.0001 \\ 
250 & 0.0005 & 0.0016 & 0.0353 & 0.0001 \\ \hline
\end{tabular}
\begin{tabular}{|rrrr|}
\multicolumn{4}{c}{Coverage} \\ \hline
Factor & WW & Pseudo* & Moment \\ \hline
0.86 & 1.00 & 0.93 & 0.96 \\ 
0.90 & 0.99 & 0.84 & 0.96 \\ 
0.92 & 1.00 & 0.78 & 0.95 \\ 
0.92 & 0.99 & 0.75 & 0.96 \\ 
0.93 & 1.00 & 0.70 & 0.95 \\ \hline
\end{tabular}%
\caption{$d = 1$ DGP~\eqref{eqn:simsAppSimple} (5\% nominal test) for 1,000
Monte Carlo rounds}
\label{tbl:AppSimulationD=1Simple}
\centering
\begin{tabular}{r|rrrr|}
& \multicolumn{4}{c}{Bias} \\ \hline
Dim. & Factor & WW & Pseudo & Moment \\ \hline
50 & 0.0157 & 0.0276 & 0.0249 & 0.0318 \\ 
100 & 0.0065 & 0.0168 & 0.0271 & 0.0316 \\ 
150 & 0.0033 & 0.0105 & 0.0412 & 0.0306 \\ 
200 & 0.0015 & 0.0057 & 0.0629 & 0.0301 \\ 
250 & 0.0010 & 0.0038 & 0.0988 & 0.0297 \\ \hline
\end{tabular}
\begin{tabular}{|rrrr|}
\multicolumn{4}{c}{Coverage} \\ \hline
Factor & WW & Pseudo* & Moment \\ \hline
0.79 & 0.99 & 0.89 & 0.80 \\ 
0.83 & 0.98 & 0.64 & 0.30 \\ 
0.89 & 0.99 & 0.53 & 0.05 \\ 
0.90 & 1.00 & 0.50 & 0.00 \\ 
0.91 & 1.00 & 0.42 & 0.00 \\ \hline
\end{tabular}%
\caption{$d = 2$ DGP~\eqref{eqn:simsAppSimple} (5\% nominal test) for 1,000
Monte Carlo rounds}
\label{tbl:AppSimulationD=2Simple}
\centering
\begin{tabular}{r|rrrr|}
& \multicolumn{4}{c}{Bias} \\ \hline
Dim. & Factor & WW & Pseudo & Moment \\ \hline
50 & 0.0254 & 0.0388 & 0.0277 & 0.0418 \\ 
100 & 0.0124 & 0.0267 & 0.0373 & 0.0405 \\ 
150 & 0.0067 & 0.0190 & 0.0412 & 0.0395 \\ 
200 & 0.0035 & 0.0126 & 0.0608 & 0.0392 \\ 
250 & 0.0023 & 0.0094 & 0.0712 & 0.0389 \\ \hline
\end{tabular}
\begin{tabular}{|rrrr|}
\multicolumn{4}{c}{Coverage} \\ \hline
Factor & WW & Pseudo* & Moment \\ \hline
0.69 & 0.96 & 0.89 & 0.66 \\ 
0.72 & 0.97 & 0.59 & 0.13 \\ 
0.80 & 0.96 & 0.44 & 0.00 \\ 
0.87 & 0.97 & 0.35 & 0.00 \\ 
0.89 & 0.98 & 0.33 & 0.00 \\ \hline
\end{tabular}%
\par
{\footnotesize {Dim. is the size of each dimension, with $N=T$. Factor is a
factor model approximation. WW is our Weighted Within estimator. Pseudo
utilises the Weighted Within estimator with pseudo-distance in \cite%
{zhang2017estimating} - *standard errors can be increasing with sample size
for Pseudo estimator. Moment is the Weighted Within estimator with
cross-sectional and time-serial means used to form distance metrics.} }
\caption{$d = 3$ DGP~\eqref{eqn:simsAppSimple} (5\% nominal test) for 1,000
Monte Carlo rounds}
\label{tbl:AppSimulationD=3Simple}
\end{table}

\begin{table}[tbp]
\centering
\begin{tabular}{r|rrrr|}
& \multicolumn{4}{c}{Bias} \\ \hline
Dim. & Factor & WW & Pseudo & Moment \\ \hline
50 & 0.0061 & 0.0142 & 0.0106 & 0.0519 \\ 
100 & 0.0011 & 0.0053 & 0.0200 & 0.0519 \\ 
150 & 0.0004 & 0.0024 & 0.0313 & 0.0519 \\ 
200 & 0.0004 & 0.0020 & 0.0371 & 0.0519 \\ 
250 & 0.0001 & 0.0008 & 0.0341 & 0.0516 \\ \hline
\end{tabular}
\begin{tabular}{|rrrr|}
\multicolumn{4}{c}{Coverage} \\ \hline
Factor & WW & Pseudo & Moment \\ \hline
0.82 & 1.00 & 0.95 & 0.54 \\ 
0.89 & 0.99 & 0.84 & 0.02 \\ 
0.90 & 1.00 & 0.80 & 0.00 \\ 
0.91 & 1.00 & 0.75 & 0.00 \\ 
0.93 & 1.00 & 0.70 & 0.00 \\ \hline
\end{tabular}%
\par
{\footnotesize {Dim. is the size of each dimension, with $N=T$. Factor is a
factor model approximation. WW is our Weighted Within estimator. Pseudo
utilises the Weighted Within estimator with pseudo-distance in \cite%
{zhang2017estimating}. Moment is the Weighted Within estimator with
cross-sectional and time-serial means used to form distance metrics.} }
\caption{$d = 1$ DGP~\eqref{eqn:simsAppSimple} with random effects; (5\%
nominal test) for 1,000 MC rounds}
\label{tbl:AppSimulationD=1SimpleRE}
\end{table}

\setlength{\bibsep}{2pt} 
\bibliographystyle{chicago3}
\bibliography{refs}

\end{document}